\documentclass[article,nojss,shortnames]{jss}

\usepackage{thumbpdf}
\usepackage{amsfonts,amstext,amsmath,amssymb,amsthm}
\usepackage{accents}
\usepackage{color}
\usepackage{rotating}
\usepackage{verbatim}

\newcommand{\blind}{1}
\newcommand{\compact}{0}
\newcommand{\compacteval}{0}

\newcommand{\etc}{\textit{etc.}}

\usepackage{booktabs}

\newcommand{\rZ}{Z}
\newcommand{\rY}{Y}
\newcommand{\rX}{\mX}
\newcommand{\rz}{z}
\newcommand{\ry}{y}
\newcommand{\rx}{\xvec}
\newcommand{\erx}{x}
\newcommand{\sA}{\mathfrak{A}}
\newcommand{\sAZ}{\mathfrak{B}}
\newcommand{\sAY}{\mathfrak{C}}

\newcommand{\esAY}{C}
\newcommand{\sam}{\Omega}
\newcommand{\samZ}{\RR}
\newcommand{\samY}{\Xi}
\newcommand{\samX}{\chi}

\newcommand{\msZ}{(\samZ, \sAZ)}
\newcommand{\msY}{(\samY, \sAY)}
\newcommand{\ps}{(\sam, \sA, \Prob)}
\newcommand{\psZ}{(\samZ, \sAZ, \Prob_\rZ)}
\newcommand{\psY}{(\samY, \sAY, \Prob_\rY)}
\newcommand{\pZ}{F_\rZ}
\newcommand{\pY}{F_\rY}
\newcommand{\oY}{O_\rY}
\newcommand{\hatpY}{\hat{F}_{\rY,N}}
\newcommand{\hatpYx}{\hat{F}_{\rY \mid \rX = \rx, N}}
\newcommand{\pN}{\Phi}
\newcommand{\pSL}{F_{\SL}}
\newcommand{\pMEV}{F_{\MEV}}
\newcommand{\pExp}{F_{\ExpD}}

\newcommand{\pYx}{F_{\rY \mid \rX = \rx}}
\newcommand{\pYA}{F_{\rY \mid \rX = A}}
\newcommand{\pYB}{F_{\rY \mid \rX = B}}
\newcommand{\qZ}{F^{-1}_\rZ}
\newcommand{\qY}{F^{-1}_\rY}
\newcommand{\dZ}{f_\rZ}
\newcommand{\dY}{f_\rY}
\newcommand{\hatdY}{\hat{f}_{\rY, N}}
\newcommand{\dYx}{f_{\rY \mid \rX = \rx}}
\newcommand{\hazY}{\lambda_\rY}
\newcommand{\HazY}{\Lambda_\rY}
\newcommand{\hathazY}{\hat{\lambda}_{\rY, N}}
\newcommand{\hatHazY}{\hat{\Lambda}_{\rY, N}}
\newcommand{\measureY}{\mu}
\newcommand{\lebesgue}{\mu_L}
\newcommand{\counting}{\mu_C}
\newcommand{\g}{g}
\newcommand{\h}{h}
\newcommand{\s}{\svec}
\newcommand{\hY}{h_\rY}
\newcommand{\hx}{h_\rx}
\newcommand{\hs}{\mathcal{H}}
\newcommand{\basisy}{\avec}
\newcommand{\bern}[1]{\avec_{\text{Bs},#1}}
\newcommand{\bernx}[1]{\bvec_{\text{Bs},#1}}
\newcommand{\basisx}{\bvec}
\newcommand{\basisyx}{\cvec}
\newcommand{\m}{m}
\newcommand{\lik}{\mathcal{L}}
\newcommand{\parm}{\varthetavec}
\newcommand{\eparm}{\vartheta}
\newcommand{\dimparm}{P}
\newcommand{\dimparmx}{Q}
\newcommand{\shiftparm}{\betavec}

\newcommand{\ie}{\textit{i.e.}~}
\newcommand{\eg}{\textit{e.g.}~}

\renewcommand{\Prob}{\mathbb{P}}
\newcommand{\Ex}{\mathbb{E}}
\newcommand{\RR}{\mathbb{R}}

\newcommand{\Null}{\mathbf{0}}

\usepackage{dsfont}
\newcommand{\I}{\mathds{1}}

 \DeclareMathOperator{\diag}{diag}

 \DeclareMathOperator*{\argmax}{{arg\,max}}

 \DeclareMathOperator{\ND}{N}
 
 \DeclareMathOperator{\UD}{U}

 \DeclareMathOperator{\ExpD}{Exp}

 \DeclareMathOperator{\SL}{SL}
 \DeclareMathOperator{\MEV}{MEV}

\def \avec {\text{\boldmath$a$}}    
\def \bvec {\text{\boldmath$b$}}    
\def \cvec {\text{\boldmath$c$}}    
    \def \mD {\text{\boldmath$D$}}
\def \evec {\text{\boldmath$e$}}    
    \def \mF {\text{\boldmath$F$}}

    \def \mI {\text{\boldmath$I$}}

\def \pvec {\text{\boldmath$p$}}

\def \svec {\text{\boldmath$s$}}

\def \xvec {\text{\boldmath$x$}}    \def \mX {\text{\boldmath$X$}}

\def \betavec         {\text{\boldmath$\beta$}}

\def \varthetavec     {\text{\boldmath$\vartheta$}}

\newtheorem{thm}{Theorem}
\newtheorem{coro}{Corollary}
\newtheorem{defn}{Definition}
\newtheorem{remark}{Remark}

\newcommand{\ubar}[1]{\underaccent{\bar}{#1}}

\newcommand{\rev}[1]{#1}

\author{Torsten Hothorn \\ Universit\"at Z\"urich \\
   \And Lisa M\"ost \\ Universit\"at M\"unchen \\
   \And Peter B\"uhlmann \\ ETH Z\"urich}
\Plainauthor{Hothorn, M\"ost, and B\"uhlmann}

\title{Most Likely Transformations} 
\Plaintitle{Most Likely Transformations} 
\Shorttitle{Most Likely Transformations}

\Abstract{

We propose and study properties of maximum likelihood estimators in the
class of conditional transformation models.  Based on a suitable explicit
parameterisation of the unconditional or conditional transformation
function, we establish a cascade of increasingly complex transformation
models that can be estimated, compared and analysed in the maximum
likelihood framework.  Models for the unconditional or conditional
distribution function of any univariate response variable can be set-up and
estimated in the same theoretical and computational framework simply by
choosing an appropriate transformation function and parameterisation
thereof.  The ability to evaluate the distribution function directly allows
us to estimate models based on the exact likelihood, especially in the
presence of random censoring or truncation.  For discrete and continuous
responses, we establish the asymptotic normality of the proposed estimators. 
A reference software implementation of maximum likelihood-based estimation
for conditional transformation models that allows the same flexibility as the
theory developed here was employed to illustrate the wide range of possible
applications.
}

\Plainkeywords{transformation model, distribution regression, conditional
distribution function, conditional quantile function, censoring,
truncation}

\Address{
  Torsten Hothorn\\
  Institut f\"ur Epidemiologie, Biostatistik und Pr\"avention \\
  Universit\"at Z\"urich \\
  Hirschengraben 84, CH-8001 Z\"urich, Switzerland \\
  \texttt{Torsten.Hothorn@uzh.ch} \\

  Lisa M\"ost \\
  Institut f\"ur Statistik \\
  Ludwig-Maximilians-Universit\"at M\"unchen \\
  Ludwigstra{\ss}e 33, DE-80539 M\"unchen, Germany \\

  Peter B\"uhlmann \\
  Seminar f\"{u}r Statistik \\
  ETH Z\"{u}rich \\
  R\"amistrasse 101, CH-8092 Z\"{u}rich, Switzerland \\

}

\begin{document}

\footnote{A shortened version of this technical report
was accepted for publication in
the \emph{Scandinavian Journal of Statistics} 2017-06-19.}

\section{Introduction}

In a broad sense, we can understand all statistical models as models of
distributions or certain characteristics thereof, especially the
mean.  All distributions $\Prob_\rY$ for at least ordered responses $\rY$
can be characterised by their distribution, quantile, density,
odds, hazard or cumulative hazard functions.  In a fully parametric setting,
all these functions have been specified up to unknown parameters, and the ease of
interpretation can guide us in looking at the appropriate function.  In the semi- and
non-parametric contexts, however, the question arises how we can obtain an
estimate of one of these functions without assuming much about their
shape.  For the direct estimation of distribution functions, we deal with
monotonic functions in the unit interval, whereas for densities, we need to
make sure that the estimator integrates to one.  The hazard function comes with a
positivity constraint, and monotonicity is required for the positive
cumulative hazard function.  These computationally inconvenient restrictions
disappear completely only when the log-hazard function is estimated, and this
explains the plethora of research papers following this path.  However, the
lack of any structure in the log-hazard function comes at a price.  A 
too-erratic behaviour of estimates of the log-hazard function has to be
prevented by some smoothness constraint; this makes classical likelihood
inference impossible.  The novel characterisation and subsequent estimation
of distributions via their transformation function in a broad class of
transformation models that are developed in this paper can be interpreted as a
compromise between structure (monotonicity) and ease of parameterisation,
estimation and inference.  This transformation approach to modelling and
estimation allows standard likelihood inference in a large class of models
that have so far commonly been \rev{dealt with} by other inference procedures.

Since the introduction of transformation models based on non-linear
transformations of some response variable by \cite{BoxCox_1964}, 
this attractive class of models has
received much interest.  In regression problems, transformation models can
be understood as models for the conditional distribution function and are
sometimes referred to as ``distribution regression'', in contrast to their
``quantile regression'' counterpart \citep{Chernozhukov_2013}.  Traditionally,
the models were actively studied and applied in the analysis of ordered
categorical or censored responses.  Recently, transformation models for the
direct estimation of conditional distribution functions for arbitrary
responses received interest in the context of counterfactual distributions
\citep{Chernozhukov_2013}, probabilistic forecasting
\citep{Gneiting_Katzfuss_2014}, distribution and quantile regression
\citep{Leorato_Peracchi_2015, Rothe_Wied_2013}, probabilistic index
models \citep{Thas_Neve_Clement_2012} and conditional transformation models
\if1\blind
\citep{Hothorn_Kneib_Buehlmann_2014}.
\else
(ANONYMOUS).
\fi
The core idea of any transformation
model is the application of a strictly monotonic transformation function $\h$
for the reformulation of an unknown distribution function $\Prob(\rY \le
\ry)$ as $\Prob(\h(\rY) \le \h(\ry))$, where the unknown transformation function $\h$ is
estimated from the data.  Transformation models have received attention
especially in situations where the likelihood contains terms involving the
conditional distribution function $\Prob(\rY \le \ry \mid \rX = \rx) =
\pZ(\h(\ry \mid \rx))$ \rev{with inverse link function $\pZ$}, 
most importantly for censored, truncated and ordered
categorical responses.  For partially linear transformation models with transformation
function $\h(\ry \mid \rx) = \hY(\ry) + \hx(\rx)$, much emphasis has been given to
estimation procedures treating the baseline transformation $\hY$ (\eg the 
log-cumulative baseline hazard function in the Cox model)
as a high-dimensional nuisance parameter.
Prominent members of these estimation procedures are the partial likelihood
estimator 
\if0\compact
\citep{Cox_1975} 
\fi
and approaches influenced by the estimation
equations introduced by 
\if0\compact
\cite{Chengetal_1995} and \cite{Chenetal_2002}. 
\else
\cite{Chengetal_1995}.
\fi
Once an estimate for the shift $\hx$ is obtained, the baseline transformation
$\hY$ is then typically estimated by the non-parametric maximum likelihood
estimator \citep[see, for example,][]{Chengetal_1997}.  An overview of the
extensive literature on the simultaneous non-parametric maximum likelihood
estimation of $\hY$ and $\hx$, \rev{\ie estimation procedures not requiring an
explicit parameterisation of $\hY$}, for censored continuous responses is given in
\cite{ZengLin_2007}.

An explicit parameterisation of $\hY$ is common in models of ordinal
responses \citep{Tutz_2012}.  For survival times,
\cite{Kooperbergetal_1995} 
\if0\compact
and \cite{KooperbergClarkson_1997} 
\fi
introduced a cubic
spline parameterisation of the log-conditional hazard function with the
possibility of \rev{response-varying effects} and estimated the  
corresponding models by maximum likelihood.  \cite{CrowtherLambert_2014}
followed up on this suggestion and used restricted cubic splines.  Many authors
studied penalised likelihood approaches for spline approximations of the
baseline hazard function in a Cox model, \eg 
\if0\compact
\cite{Jolyetal_1998, CaiBetensky_2003, Bove_Held_2013} or 
\fi
\cite{Maetal_2014}.  Less frequently, the
transformation function $\hY$ was modelled directly. 
\cite{MallickWalker_2003, Chang_Hsiung_Wu_2005} and
\cite{McLain_Ghosh_2013} used Bernstein polynomials for $\hY$, 
and \cite{Royston_Parmar_2002} proposed a maximum
likelihood approach using cubic splines for modelling $\hY$ and also
time-varying effects.  
The connection between these different transformation models is difficult to see
because most authors present their models in the relatively narrow contexts
of survival or ordinal data.  The lack of a general understanding of
transformation models made the development of novel approaches in this model
class burdensome.  
\if1\blind
\cite{Hothorn_Kneib_Buehlmann_2014} 
\else
ANONYMOUS
\fi
decoupled the
parameterisation of the conditional transformation function $\h(\ry \mid
\rx)$ from the
estimation procedure, and showed that many interesting and novel models can
be understood as transformation models.  The boosting-based optimisation of
proper scoring rules, however, was only developed for uncensored
und right-censored 
\if0\compact
\citep{Moest_Hothorn_2015}
\fi
observations in the absence of truncation and requires the numerical
approximation of the true target function.  In a similar spirit,
\cite{Chernozhukov_2013} applied the connection
$\Prob(\rY \le \ry \mid \rX = \rx) = \Ex(\I(\rY \le y) \mid \rX = \rx)$ for
estimation in the \rev{response-varying effects} transformation model $\Prob(\rY \le \ry \mid \rX =
\rx) = \pZ(\hY(\ry) - \rx^\top \shiftparm(\ry))$; this approach 
can be traced back to \cite{Foresi_Peracchi_1995}.

A drawback of all but the simplest
transformation models is the lack of a likelihood estimation procedure. 
Furthermore, although important connections to other models have been known
for some time \citep{Doksum_Gasko_1990}, it is often not easy to see how
broad and powerful the class of transformation models actually is.  We
address these issues and embed the 
estimation of unconditional and conditional distribution
functions of arbitrary univariate random variables under all forms of
random censoring and truncation into a common theoretical and computational
likelihood-based framework.
In a nutshell, we show in Section~\ref{sec:trafo} that all distributions can
be generated by a strictly monotonic transformation of some absolute
continuous random variable.  The likelihood function of the transformed
variable can then be characterised by this transformation function.  The
parameters of appropriate parameterisations of the transformation function,
and thus the parameters of the conditional distribution function in which we are
interested, can then be estimated by maximum likelihood under simple linear
constraints that allow classical asymptotic likelihood inference,
as will be shown in Section~\ref{sec:mlt}.  Many classical and
contemporary models are introduced as special cases of this framework.
\rev{In particular, all transformation models sketched in this introduction 
can be understood and estimated in this novel likelihood-based framework. 
Extensions of classical and contemporary transformation models as well as 
some novel models are derived from our unified theoretical framework of 
transformation functions} in Section~\ref{sec:appl}, and their empirical
performance is illustrated 
\if0\compacteval
and evaluated 
\fi
in Section~\ref{sec:empeval}.

\section{The Likelihood of Transformations} \label{sec:trafo}


Let $\ps$ denote a probability space and $\msY$ a measureable space with 
at least ordered sample space $\samY$. We are interested in inference about 
the distribution $\Prob_\rY$ of a random variable $\rY$, \ie the
probability space $\psY$ defined by the $\sA-\sAY$ measureable function
$\rY: \sam \rightarrow \samY$. For the sake of notational simplicity, we 
present our results for the unconditional 
\if0\compact
and ordered 
\fi
case first; 
regression models 
\if0\compact
and unordered responses 
\fi
are discussed in Section~\ref{subsec:conditional}.
The distribution $\Prob_\rY = \dY \odot \measureY$ is dominated by
some measure $\measureY$ and 
characterised by its density function $\dY$, distribution function 
$\pY(\ry) = \Prob_\rY(\{\xi \in \samY \mid \xi \le \ry\})$, quantile function 
$\qY(p) = \inf\{y \in \samY \mid \pY(y) \ge p\}$, \rev{odds function 
$\oY(\ry) = \pY(\ry) / (1 - \pY(\ry))$}, hazard function 
$\hazY(\ry) = \dY(\ry) / (1 - \pY(\ry))$ or cumulative hazard 
function $\HazY(\ry) = -\log(1 - \pY(\ry))$. For notational convenience, 
we assume strict monotonicity of $\pY$, \ie $\pY(\ry_1) < \pY(\ry_2) 
\, \forall \ry_1 < \ry_2 \in \samY$.
Our aim is to obtain an estimate $\hatpY$ of the distribution function
$\pY$ from a random sample $\rY_1, \dots, \rY_N \stackrel{\text{iid}}{\sim} \Prob_\rY$. 
In the following, we will show that one can always write this potentially complex
distribution function $\pY$ as the composition of a much simpler and
\textit{a priori} specified distribution function $\pZ$ and a strictly monotonic
transformation function $\h$. The task of estimating $\pY$ is then reduced
to obtaining an estimate $\hat{\h}_N$. The latter exercise, as we will show in this paper,
is technically and conceptually attractive.

Let $\msZ$ denote the Euclidian space with Borel $\sigma$-algebra and 
$\rZ: \sam \rightarrow \samZ$ an $\sA-\sAZ$ measureable
function such that the distribution $\Prob_\rZ = \dZ \odot \lebesgue$ is
absolutely
continuous ($\lebesgue$ denotes the Lebesgue measure) in the probability space $\psZ$.
Let $\pZ$ and $\qZ$ denote the corresponding distribution and quantile
functions. We furthermore assume $0 < \dZ(\rz) < \infty \, \forall \rz \in
\samZ$, $\pZ(-\infty) = 0$ and $\pZ(\infty) = 1$ for a log-concave density
$\dZ$ as well as the existence of the
first two derivatives of the density $\dZ(\rz)$ with respect to $\rz$; both
derivatives shall be bounded.
We do not allow any unknown parameters for this distribution. Possible choices
include the standard normal, standard logistic (SL) and minimum extreme value
(MEV) distribution with distribution functions $\pZ(\rz) = \Phi(\rz)$, $\pZ(\rz) =
\pSL(\rz) = (1 + \exp(-\rz))^{-1}$
and $\pZ(\rz) = \pMEV(\rz) = 1 -\exp(-\exp(\rz))$, respectively.
In the first step, we will show that there always exists a unique and strictly
monotonic transformation $\g$ such that the unknown and potentially complex
distribution $\Prob_\rY$ that we are interested in can be generated from the
simple and known distribution $\Prob_\rZ$ via $\Prob_\rY = \Prob_{g \circ
\rZ}$. More formally, let $\g: \samZ \rightarrow \samY$ denote a
$\sAZ-\sAY$ measureable function. The composition $\g \circ \rZ$ is
a random variable on $(\samY, \sAY, \Prob_{\g \circ \rZ})$.
We can now formulate the existence and uniqueness of $\g$ as a corollary
to the probability integral transform.
\begin{coro} \label{thm1}
For all random variables $\rY$ and $\rZ$, there exists a unique strictly monotonically increasing transformation $\g$, such that $\Prob_\rY = \Prob_{\g \circ \rZ}$.
\end{coro}

\begin{proof}
Let $\g = \pY^{-1} \circ \pZ$ and $\rZ \sim \Prob_\rZ$. Then $U := \pZ(\rZ)
\sim \UD[0, 1]$ and $\rY = \pY^{-1}(U) \sim \Prob_\rY$ by the 
probability integral transform. Let $\h: \samY \rightarrow \samZ$, such that
$\pY(\ry) = \pZ(\h(\ry))$. From
\if0\compact
\begin{eqnarray*}
\else
$
\fi
\pY(\ry) = (\pZ \circ \pZ^{-1} \circ \pY)(\ry) \iff \h = \pZ^{-1} \circ \pY
\if0\compact
\end{eqnarray*}
\else
$
\fi
we get the uniqueness of $\h$ and therefore $\g$.
The quantile function $\pZ^{-1}$ and the distribution function $\pY$ exist
by assumption and are both strictly monotonic and right-continuous.
Therefore, $\h$ is strictly monotonic and right-continuous and so is $\g$.
\end{proof}

\begin{coro}  \label{cor1}
For $\measureY = \lebesgue$, we have $\g = \h^{-1}$ and $\h^\prime(y) = 
\frac{\partial \h(\ry)}{\partial \ry} = \dZ((\pZ^{-1} \circ \pY)(\ry))^{-1}
\dY(\ry)$.
\end{coro}
This result for absolutely continuous random variables $\rY$ can be found in
many textbooks \citep[\eg][]{Lindsey_1996}, Corollary~\ref{thm1}
also covers the discrete case.

\begin{coro}
For the counting measure $\measureY = \counting$, $\h = \pZ^{-1} \circ \pY$ is a right-continuous step function
because $\pY$ is a right-continuous step function with steps at $\ry \in
\samY$.
\end{coro}



We now characterise
the distribution $\pY$ by the corresponding transformation function $\h$, set up the
corresponding likelihood of such a transformation function and estimate
the transformation function based on this likelihood.
Let $\hs = \{\h: \samY \rightarrow \samZ \mid \sAY-\sAZ \text{ measureable}, 
\h(\ry_1) < \h(\ry_2) \, \forall \ry_1 < \ry_2 \in \samY\}$ denote the 
space of all strictly monotonic transformation functions. 
With the transformation function $\h$, we can evaluate $\pY$ as 
$\pY(\ry \mid \h) = \pZ(\h(\ry)) \, \forall \ry \in \samY$. Therefore, we only
need to study the transformation function $\h$; the inverse transformation $\g
= \h^{-1}$
\if0\compact
\citep[used to define a ``group family'' or ``group model'' by][]{Lehmann_1983,Bickel_Klaassen_Ritov_1993} 
\else
\citep[used to define a ``group model'' by][]{Bickel_Klaassen_Ritov_1993}
\fi
is not necessary in what follows. The density for
absolutely continuous variables $\rY$ ($\measureY =
\lebesgue$) is now given by
\if0\compact
\begin{eqnarray*}
\else
$
\fi
\dY(\ry \mid \h) = \dZ(\h(\ry)) \h^\prime(\ry).
\if0\compact
\end{eqnarray*}
\else
$
\fi
For discrete responses $\rY$ ($\measureY = \counting$) with finite sample space $\samY = \{\ry_1, \dots, \ry_K\}$,
the density is
\begin{eqnarray*}
\dY(\ry_k \mid \h) = \left\{ \begin{array}{ll} \pZ(\h(\ry_k)) & k = 1 \\
\pZ(\h(\ry_k)) - \pZ(\h(\ry_{k-1})) & k = 2, \dots, K  - 1 \\
1 - \pZ(\h(\ry_{k-1})) & k = K \end{array} \right.
\end{eqnarray*}
and for countably infinite sample spaces $\samY = \{\ry_1, \ry_2, \ry_3, \dots\}$, we get
the density
\begin{eqnarray*}
\dY(\ry_k \mid \h) = \left\{ \begin{array}{ll} \pZ(\h(\ry_k)) & k = 1 \\
\pZ(\h(\ry_k)) - \pZ(\h(\ry_{k-1})) & k > 1.  \end{array} \right.
\end{eqnarray*}
With the conventions $\pZ(\h(\ry_0)) := \pZ(\h(-\infty)) := 0$ and
$\pZ(\h(\ry_K)) := \pZ(\h(\infty)) := 1$, we use the more compact notation
$\dY(\ry_k \mid \h) = \pZ(\h(\ry_k)) - \pZ(\h(\ry_{k-1}))$ in the sequel.

For a given transformation function $h$, the likelihood contribution of a
datum $\esAY = (\ubar{\ry},\bar{\ry}] \in \sAY$ is defined in terms of the
distribution function \citep{Lindsey_1996}: 
\begin{eqnarray*} 
\lik(\h \mid \rY \in \esAY) := \int_{\esAY} \dY(y \mid \h) d\measureY(y) =
\pZ(\h(\bar{\ry})) - \pZ(\h(\ubar{\ry})).
\end{eqnarray*}
This ``exact'' definition of the likelihood applies to most practical situations of interest and, in particular, allows discrete and (conceptually) continuous as well as
censored or truncated observations $\esAY$. For a discrete response $\ry_k$, we have $\bar{\ry} = \ry_k$ and $\ubar{\ry}
= \ry_{k -1}$, such that $\lik(\h \mid \rY = \ry_k) = \dY(\ry_k \mid \h) =
\pZ(\h(\bar{\ry})) - \pZ(\h(\ubar{\ry}))$.  For absolutely continuous random
variables $\rY$, we almost always observe an imprecise datum $(\ubar{\ry},\bar{\ry}]
\subset \RR$ and, for short intervals $(\ubar{\ry},\bar{\ry}]$, approximate
the exact likelihood $\lik(\h \mid \rY \in (\ubar{\ry},\bar{\ry}])$ by the term
$(\bar{\ry} - \ubar{\ry}) \dY(\ry \mid \h)$ or simply $\dY(\ry \mid \h)$ with $\ry
= (\ubar{\ry} + \bar{\ry})/2$ \citep{Lindsey_1999}. This approximation only
works for relatively precise measurements, \ie short intervals. If longer intervals 
are observed, one speaks of ``censoring'' and relies on the exact definition
of the likelihood contribution instead of using the above approximation \citep{Klein_Moeschberger_2003}. 
In summary, the likelihood contribution of a conceptually ``exact
continuous'' or left-, right- or interval-censored continuous or discrete 
observation $(\ubar{\ry}, \bar{\ry}]$ is given by 
\begin{eqnarray*}
\lik(\h \mid \rY \in (\ubar{\ry}, \bar{\ry}]) \left\{ \begin{array}{ll}
    \approx \dZ(\h(\ry)) \h^\prime(\ry) & \ry = (\ubar{\ry} + \bar{\ry})/2 \in \samY \quad \text{```exact continuous'''}\\
    = 1 - \pZ(\h(\ubar{\ry})) & \ry \in (\ubar{\ry}, \infty) \cap \samY \quad \text{`right-censored'} \\
    = \pZ(\h(\bar{\ry})) & \ry \in (-\infty, \bar{\ry}] \cap \samY \quad \text{`left-censored'} \\
    = \pZ(\h(\bar{\ry})) - \pZ(\h(\ubar{\ry})) & \ry \in (\ubar{\ry},\bar{\ry}] \cap \samY \quad    \text{`interval-censored',}
\end{array} \right. 
\end{eqnarray*}
under the assumption of random censoring.
The likelihood is more complex under dependent censoring
\citep{Klein_Moeschberger_2003}, but we will not elaborate on this issue. 
The likelihood contribution $\lik(\h \mid \rY \in (\ry_k, \ry_{k-1}])$
of an ordered factor in category $\ry_k$ is equivalent to the term 
$\lik(\h \mid \rY \in (\ubar{\ry}, \bar{\ry}])$
contributed by an interval-censored observation $(\ubar{\ry},\bar{\ry}]$,
when category $\ry_k$ is defined by the interval $(\ubar{\ry},\bar{\ry}]$. Thus,
the expression $\pZ(\h(\bar{\ry})) - \pZ(\h(\ubar{\ry}))$ for the likelihood
contribution reflects the equivalence of interval-censoring 
and categorisation at corresponding cut-off points. 

For truncated observations in the interval $(\ry_l, \ry_r] \subset \samY$,
the above likelihood contribution is defined in terms of the 
distribution function conditional on the truncation
\begin{eqnarray*}
\pY(\ry \mid \rY \in (\ry_l, \ry_r]) = \pZ(\h(\ry) \mid \rY \in (\ry_l, \ry_r]) = 
\frac{\pZ(\h(\ry))}{\pZ(\h(\ry_r)) - \pZ(\h(\ry_l))}  
\quad \forall \ry \in (\ry_l, \ry_r]
\end{eqnarray*}
and thus the likelihood contribution changes to \citep{Klein_Moeschberger_2003}
\begin{eqnarray*}
\frac{\lik(\h \mid \rY \in (\ubar{\ry}, \bar{\ry}])}{\pZ(\h(\ry_r)) -
\pZ(\h(\ry_l))} = \frac{\lik(\h \mid \rY \in (\ubar{\ry}, \bar{\ry}])}{\lik(\h
\mid \rY \in (\ry_l, \ry_r])} \quad \text{when } \ry_l < \ubar{\ry} < \bar{\ry} \le \ry_r.
\end{eqnarray*}
It is important to note that the likelihood is always \textit{defined} in
terms of a distribution function \citep{Lindsey_1999}, and it therefore makes
sense to directly model the distribution function of interest.  The ability
to uniquely characterise this distribution function by the
transformation function $\h$ gives rise to the following definition of an
estimator $\hat{\h}_N$.

\begin{defn}[Most likely transformation]
Let $\esAY_1, \dots, \esAY_N$ denote an independent sample of possibly
randomly censored or truncated observations from $\Prob_\rY$. The estimator
\begin{eqnarray*}
\hat{\h}_N := \argmax_{\tilde{\h} \in \hs} \sum_{i = 1}^{N} \log(\lik(\tilde{\h} \mid
\rY \in \esAY_i)) 
\end{eqnarray*}
is called the most likely transformation (MLT).
\end{defn}
Log-concavity of $\dZ$ ensures concavity of the log-likelihood (except 
when all observations are right-censored) and thus
ensures the existence and uniqueness of $\hat{\h}_N$.

\if0\compact
\paragraph{Example} For an absolutely continuous response $\rY$ the likelihood and
log-likelihood for $\h$ are approximated by the density and log-density 
evaluated at $\ry = (\ubar{\ry} + \bar{\ry})/2$, respectively:
\begin{eqnarray*}
\lik(\h \mid \rY \in (\ubar{\ry}, \bar{\ry}]) & \approx & \dZ(\h(\ry))\h^\prime(\ry) \\
\log(\lik(\h \mid \rY \in (\ubar{\ry}, \bar{\ry}])) & \approx & \log(\dZ(\h(\ry))) + \log(\h^\prime(\ry)).
\end{eqnarray*}
Strict monotonicity of the transformation function $\h$ is required;
otherwise the likelihood is not defined. The term $\log(\h^\prime(\ry))$ is
not a penalty term, but the likelihood favours transformation functions with
a large positive derivative at the observations. If we assume $\rY \sim \ND(\alpha, \sigma^2)$ 
and for the choice $\rZ \sim \ND(0, 1)$ with $\pZ = \Phi$ and $\dZ = \phi$,
we can restrict $\h$ to linear functions $\h(\ry) = (\ry - \alpha) \sigma^{-1}$.  
The likelihood reduces to
\begin{eqnarray*}
\lik(\h \mid \rY \in (\ubar{\ry}, \bar{\ry}]) \approx \dZ(\h(\ry))\h^\prime(\ry) = \phi((\ry - \alpha)
\sigma^{-1}) \sigma^{-1} = \phi_{\alpha, \sigma^2}(\ry) = \dY(\ry \mid \alpha,
\sigma^2).
\end{eqnarray*}
In this simple location-scale family, the most likely transformation is
characterised by the parameters of the normal distribution of $\rY$.  It is
important to note that for other choices of $\pZ$, the most likely
transformation is non-linear; however, the distribution function $\pY =
\pZ(\h(\ry))$ is invariant with respect to $\pZ$ because we can always write
$\h$ as $\pZ^{-1} \circ \pY$.  In other words, with $\pZ \neq \Phi$, we can
still model normal responses $\rY$; however, a non-linear transformation
function $\h$ is required.
\fi

Many distributions are defined by a transformation function $\h$,
for example, the Box-Cox power exponential family
\if0\compact
\citep{Rigby_Stasinopoulos_2004}, 
\else
\citep{Stasinopoulos_Rigby_2007},
\fi
the sinh-arcsinh distributions \citep{Jones_Pewsey_2009}, or
the T-X family of distributions \citep{Alzaatreh_Lee_Famoye_2013}.
\if0\compact
The parameters of these distributions can, for
example, be estimated by the GAMLSS approach
\citep{Rigby_Stasinopoulos_2005}.  
\fi
In what follows, we do not assume any
specific form of the transformation function but parameterise $\h$ in terms
of basis functions.  We now introduce such a parameterisation, a
corresponding family of distributions, a maximum likelihood estimator and a
large class of models for unconditional and conditional distributions.

\section{Transformation Analysis} \label{sec:mlt}

We parameterise the transformation function $\h(\ry)$ as a linear function of its
basis-trans\-formed argument $\ry$ using a basis function $\basisy: \samY
\rightarrow \RR^\dimparm$, such that $\h(\ry) = \basisy(\ry)^\top \parm,
\parm \in \RR^\dimparm$.  The choice of the basis function $\basisy$ is
problem specific and will be discussed in Section~\ref{sec:appl}.  The
likelihood $\lik$ only requires evaluation of $\h$, and only an
approximation thereof using the Lebesgue density of ``exact continuous'' observations 
makes the evaluation of the 
first derivative of $\h(\ry)$ with respect to $\ry$ necessary.  In this
case, the derivative with respect to $\ry$ is given by
$\h^\prime(\ry) = \basisy^\prime(\ry)^\top \parm$, and we assume that
$\basisy^\prime$ is available.
In the following, we will write $\h = \basisy^\top
\parm$ and $\h^\prime = {\basisy^\prime}^\top \parm$ for the transformation
function and its first derivative, omitting the argument $\ry$,
and we assume that both functions are bounded away from $-\infty$ and
$\infty$. For a
specific choice of $\pZ$ and $\basisy$, the transformation family of
distributions consists of all distributions $\Prob_\rY$ whose distribution
function $\pY$ is given as the composition $\pZ \circ \basisy^\top \parm$;
this family can be formally defined as follows.

\begin{defn}[Transformation family] 
The distribution family
\if0\compact
\begin{eqnarray*}
\Prob_{\rY, \Theta} = \{ \pZ \circ \basisy^\top\parm \mid \parm \in \Theta \}
\end{eqnarray*}
\else
$\Prob_{\rY, \Theta} = \{ \pZ \circ \basisy^\top\parm \mid \parm \in \Theta \}$
\fi
with parameter space $\Theta = \{\parm \in \RR^\dimparm \mid \basisy^\top\parm
\in \hs\}$ is called 
transformation family of distributions $\Prob_{\rY, \parm}$ with transformation functions
$\basisy^\top\parm \in \hs$, $\mu$-densities $\dY(\ry\mid \parm), \ry \in
\samY$, 
and error distribution function $\pZ$.
\end{defn}

The classical definition of a transformation family relies on the idea of
invariant distributions, \ie only the parameters of a distribution are
changed by a transformation function but the distribution itself is not changed. 
The normal family characterised by affine transformations is the most
well-known example \citep[\eg][]{Fraser_1968, Lindsey_1996}. 
Here, we explicitly allow and encourage transformation functions that 
change the shape of the distribution.  The transformation function
$\basisy^\top\parm$ is, at least in principle, flexible enough to generate
any distribution function $\pY = \pZ \circ \basisy^\top\parm$ from the
distribution function $\pZ$.  We borrow the term ``error distribution''
function for $\pZ$ from \cite{Fraser_1968}, because $\rZ$ can be understood
as an error term in some of the models discussed in Section~\ref{sec:appl}. 
The problem of estimating the unknown transformation function $\h$, and thus
the unknown distribution function $\pY$, reduces to the problem of
estimating the parameter vector $\parm$ through maximisation of the
likelihood function. We assume that the basis function $\basisy$ is
such that the parameters $\parm$ are identifiable.

\begin{defn}[Maximum likelihood estimator]
\if0\compact
\begin{eqnarray*}
\hat{\parm}_N := \argmax_{\parm \in \Theta} \sum_{i = 1}^N
\log(\lik(\basisy^\top \parm \mid \rY \in \esAY_i))
\end{eqnarray*}
\else
$\hat{\parm}_N := \argmax\limits_{\parm \in \Theta} \sum\limits_{i = 1}^N
\log(\lik(\basisy^\top \parm \mid \rY \in \esAY_i))$
\fi
\end{defn}

Based on the maximum likelihood estimator $\hat{\parm}_N$, we
define plug-in estimators of the most likely transformation function and
the corresponding estimator of our target distribution $\pY$ as 
$\hat{\h}_N := \basisy^\top \hat{\parm}_N$ and 
$\hatpY := \pZ \circ \hat{\h}_N$.
Because the problem of estimating an unknown distribution function is now
embedded in the maximum likelihood framework, the asymptotic analysis
benefits from standard results on the asymptotic behaviour of maximum likelihood
estimators.
We begin with deriving the score function and Fisher information.
The score contribution of an ``exact continuous'' 
observation $\ry = (\ubar{\ry} +
\bar{\ry})/2$ from an absolutely continuous distribution is approximated by the
gradient of the log-density
\if0\compact
\begin{eqnarray}
\s(\parm \mid \rY \in (\ubar{\ry}, \bar{\ry}]) \approx
\frac{\partial \log(\dY(\ry \mid \parm))}{\partial \parm} & = & 
\frac{\partial \log(\dZ(\basisy(\ry)^\top \parm))) +
\log({\basisy^\prime(\ry)}^\top \parm)}{\partial \parm} \nonumber \\
& = & \basisy(\ry) \frac{\dZ^\prime(\basisy(\ry)^\top \parm)}{\dZ(\basisy(\ry)^\top \parm)}
    + \frac{\basisy^\prime(\ry)}{{\basisy^\prime(\ry)}^\top \parm}. \label{f:s_exact}
\end{eqnarray}
\else 
\begin{eqnarray}
\s(\parm \mid \rY \in (\ubar{\ry}, \bar{\ry}]) \approx
\basisy(\ry) \frac{\dZ^\prime(\basisy(\ry)^\top \parm)}{\dZ(\basisy(\ry)^\top \parm)}
    + \frac{\basisy^\prime(\ry)}{{\basisy^\prime(\ry)}^\top \parm}. \label{f:s_exact}
\end{eqnarray}
\fi
For an interval-censored or discrete observation $\ubar{\ry}$ and
$\bar{\ry}$ (the constant terms $\pZ(\basisy(-\infty)^\top \parm) =
\pZ(-\infty) = 0$ and  $\pZ(\basisy(\infty)^\top \parm) =
\pZ(\infty) = 1$ vanish), the score contribution is
\if0\compact
\begin{eqnarray}
\s(\parm \mid \rY \in (\ubar{\ry}, \bar{\ry}]) & = & \frac{\partial \log(\lik(\basisy^\top \parm \mid \rY \in (\ubar{\ry},
\bar{\ry}]))}{\partial \parm} \nonumber \\
& = & \frac{\partial \log(\pZ(\basisy(\bar{\ry})^\top \parm) - \pZ(\basisy(\ubar{\ry})^\top \parm))}{\partial \parm}  \nonumber \\
& = & \frac{\dZ(\basisy(\bar{\ry})^\top \parm)\basisy(\bar{\ry}) - \dZ(\basisy(\ubar{\ry})^\top
\parm) \basisy(\ubar{\ry})}{\pZ(\basisy(\bar{\ry})^\top \parm) - \pZ(\basisy(\ubar{\ry})^\top
\parm)}. \label{f:s_interval}
\end{eqnarray}
\else
\begin{eqnarray}
\s(\parm \mid \rY \in (\ubar{\ry}, \bar{\ry}]) = 
\frac{\dZ(\basisy(\bar{\ry})^\top \parm)\basisy(\bar{\ry}) - \dZ(\basisy(\ubar{\ry})^\top
\parm) \basisy(\ubar{\ry})}{\pZ(\basisy(\bar{\ry})^\top \parm) - \pZ(\basisy(\ubar{\ry})^\top
\parm)}. \label{f:s_interval}
\end{eqnarray}
\fi
For a truncated observation, the score function is $\s(\parm \mid \rY \in
(\ubar{\ry}, \bar{\ry}]) - \s(\parm \mid \rY \in (\ry_l, \ry_r])$.

The contribution of an ``exact continuous''  observation $\ry$ from an
absolutely
continuous distribution to the Fisher information is approximately
\if0\compact
\begin{eqnarray}
\mF(\parm \mid \rY \in (\ubar{\ry}, \bar{\ry}]) & \approx & 
-\frac{\partial^2 \log(\dY(\ry \mid \parm))}{\partial \parm
\partial \parm^\top} \nonumber \\
& = & - \left(
\basisy(\ry) \basisy(\ry)^\top \left\{
    \frac{\dZ^{\prime\prime}(\basisy(\ry)^\top \parm)}{\dZ(\basisy(\ry)^\top \parm)}
   -\left[\frac{\dZ^{\prime}(\basisy(\ry)^\top \parm)}{\dZ(\basisy(\ry)^\top \parm)}\right]^2\right\} \right. \\ \nonumber 
& &   \left. \quad - \frac{\basisy^\prime(\ry){\basisy^\prime(\ry)}^\top}{{(\basisy^\prime(\ry)}^\top\parm)^2}\right) \label{f:F_exact}
\end{eqnarray}
(NB: the weight to $\basisy(\ry) \basisy(\ry)^\top$ is constant one for
$\pZ = \Phi$). 
\else
\begin{eqnarray}
\mF(\parm \mid \rY \in (\ubar{\ry}, \bar{\ry}])  \approx  - \left(
\basisy(\ry) \basisy(\ry)^\top \left\{
    \frac{\dZ^{\prime\prime}(\basisy(\ry)^\top \parm)}{\dZ(\basisy(\ry)^\top \parm)}
   -\left[\frac{\dZ^{\prime}(\basisy(\ry)^\top \parm)}{\dZ(\basisy(\ry)^\top \parm)}\right]^2\right\}
   - \frac{\basisy^\prime(\ry){\basisy^\prime(\ry)}^\top}{{(\basisy^\prime(\ry)}^\top\parm)^2}\right). \label{f:F_exact}
\end{eqnarray}
\fi
For a censored or discrete observation, we have the following
contribution to the Fisher information
\if0\compact
\begin{eqnarray}
\mF(\parm \mid \rY \in (\ubar{\ry}, \bar{\ry}]) & = & -\frac{\partial^2 \log(\lik(\basisy^\top \parm \mid
\rY \in (\ubar{\ry}, \bar{\ry}]))}{\partial \parm \partial \parm^\top} \nonumber \\
& = & 
- \left\{\frac{\dZ^\prime(\basisy(\bar{\ry})^\top \parm)\basisy(\bar{\ry})\basisy(\bar{\ry})^\top -
      \dZ^\prime(\basisy(\ubar{\ry})^\top \parm) \basisy(\ubar{\ry}) \basisy(\ubar{\ry})^\top}
     {\pZ(\basisy(\bar{\ry})^\top \parm) - \pZ(\basisy(\ubar{\ry})^\top \parm)}
\right. \label{f:F_interval} \\
& &  \quad -\frac{[\dZ(\basisy(\bar{\ry})^\top \parm)\basisy(\bar{\ry}) - 
       \dZ(\basisy(\ubar{\ry})^\top \parm) \basisy(\ubar{\ry})] }
     {[\pZ(\basisy(\bar{\ry})^\top \parm) - \pZ(\basisy(\ubar{\ry})^\top \parm]^2} \times \nonumber \\
& & \left. \qquad      [\dZ(\basisy(\bar{\ry})^\top \parm)\basisy(\bar{\ry})^\top - 
       \dZ(\basisy(\ubar{\ry})^\top \parm) \basisy(\ubar{\ry})^\top]
\right\}. \nonumber
\end{eqnarray}
\else
\begin{eqnarray}
& & \mF(\parm \mid \rY \in (\ubar{\ry}, \bar{\ry}]) = - \left\{\frac{\dZ^\prime(\basisy(\bar{\ry})^\top \parm)\basisy(\bar{\ry})\basisy(\bar{\ry})^\top -
      \dZ^\prime(\basisy(\ubar{\ry})^\top \parm) \basisy(\ubar{\ry}) \basisy(\ubar{\ry})^\top}
     {\pZ(\basisy(\bar{\ry})^\top \parm) - \pZ(\basisy(\ubar{\ry})^\top \parm)} \right. \label{f:F_interval} \\
& &  \left. \quad -\frac{[\dZ(\basisy(\bar{\ry})^\top \parm)\basisy(\bar{\ry}) - 
       \dZ(\basisy(\ubar{\ry})^\top \parm) \basisy(\ubar{\ry})] }
     {[\pZ(\basisy(\bar{\ry})^\top \parm) - \pZ(\basisy(\ubar{\ry})^\top \parm]^2} \times 
[\dZ(\basisy(\bar{\ry})^\top \parm)\basisy(\bar{\ry})^\top - 
       \dZ(\basisy(\ubar{\ry})^\top \parm) \basisy(\ubar{\ry})^\top]
\right\}. \nonumber
\end{eqnarray}
\fi
For a truncated observation, the Fisher information is given by
$\mF(\parm \mid \rY \in (\ubar{\ry}, \bar{\ry}]) - \mF(\parm \mid \rY \in (\ry_l,
\ry_r])$.

We will first discuss the asymptotic properties of the maximum likelihood estimator $\hat{\parm}_N$
in the parametric setting with fixed parameters $\parm$ in both the discrete and
continuous case. For continuous variables $\rY$ and a transformation
function parameterised using a Bernstein polynomial, results for sieve maximum
likelihood estimation, where the number of parameters increases with $N$,
are then discussed in Subsection~\ref{subsec:nonparam}.

\subsection{Parametric Inference} \label{subsec:param}

Conditions on the densities of the error distribution $\dZ$ and the basis
functions $\basisy$ ensuring consistency and asymptotic normality of the
sequence of maximum likelihood estimators $\hat{\parm}_N$ and an estimator
of their asymptotic covariance matrix are given in the following three
theorems.  Due to the full parameterisation of the model, the proofs are
simple standard results for likelihood asymptotics, and a more complex
analysis 
\if0\compact
\citep[as required for estimation equations in the presence of a
nuisance parameter $\h_\rY$, for example in][]{Chengetal_1995,Chenetal_2002}
\else
\citep[as required for estimation equations in the presence of a
nuisance parameter $\h_\rY$, for example in][]{Chengetal_1995}
\fi
is not necessary.  We will restrict ourselves to absolutely continuous or discrete
random variables $\rY$, where the likelihood is given in terms of the density
$\dY(\ry \mid \parm)$. Furthermore, we will only study the case of a correctly
specified transformation $\h = \basisy^\top \parm$ and refer the reader to
\if1\blind
\cite{Hothorn_Kneib_Buehlmann_2014},
\else
ANONYMOUS
\fi
where consistency results for arbitrary $\h$ are given.

\begin{thm} \label{thm:consistency}
For $\rY_1, \dots, \rY_N \stackrel{\text{iid}}{\sim} \Prob_{\rY, \parm_0}$
and under the assumptions \textnormal{(A1)} the parameter space $\Theta$ is compact and
\textnormal{(A2)} $\Ex_{\parm_0} [\sup_{\parm \in \Theta} |\log(\dY(\rY \mid \parm))|] < \infty$
where $\parm_0$ is well-separated: 
\begin{eqnarray*}
\sup_{\parm; |\parm - \parm_0| \ge \epsilon} \Ex_{\parm_0} [\log(\dY(\rY \mid \parm))] <
  \Ex_{\parm_0}[\log(\dY(\rY \mid \parm_0))],
\end{eqnarray*}
the sequence of estimators $\hat{\parm}_N$ converges to $\parm_0$ in
probability, $\hat{\parm}_N \stackrel{\Prob}{\rightarrow} \parm_0$, as $N
\to \infty$. 
\end{thm}
\begin{proof}
The log-likelihood is continuous in $\parm$, and due to (A2), each log-likelihood
contribution is dominated by an integrable function. Thus,
the result follows from \cite{vdVaart_1998} (Theorem 5.8 with Example 19.7;
see note at bottom of page 46). 
\end{proof}

\begin{remark}
Assumption \textnormal{(A1)} is made for convenience, and relaxations of such a condition
are given in \cite{vdGeer_2000} or \cite{vdVaart_1998}. The assumptions
in \textnormal{(A2)} are rather weak: the first one holds if the functions $\basisy$ are
not arbitrarily ill-posed, and the second one holds if the
function $\Ex_{\parm_0} [\log(\dY(\rY \mid \parm))]$ is strictly convex in
$\parm$ (if the assumption would not hold, we would still have convergence
to the set $\mbox{argmax}_{\parm} \Ex_{\parm_0} [\log(\dY(\rY \mid \parm))]$).
\end{remark}

\begin{thm} \label{thm:normality}
Under the assumptions of Theorem~\ref{thm:consistency} and in addition
\textnormal{(A3)}
\begin{eqnarray*}
\Ex_{\parm_0} \left( \sup_\parm \left|\left| \frac{\partial \log \dY(\rY
  \mid \parm)}{\partial \parm} \right|\right|^2 \right) < \infty,
\end{eqnarray*}
\textnormal{(A4)} 
$\Ex_{\parm_0}(\basisy(\rY) \basisy(\rY)^\top)$ and (for the absolutely
continuous case $\mu =
  \lebesgue$ only) 
          $\Ex_{\parm_0}(\basisy^\prime(\rY) \basisy^\prime(\rY)^\top)$ are
          nonsingular, and
\textnormal{(A5)} 
$0 < \dZ < \infty$, $\sup | \dZ^\prime | < \infty$ and $\sup | \dZ^{\prime\prime}| <
\infty$,
the sequence $\sqrt{N}(\hat{\parm}_N - \parm_0)$ is asymptotically normal
with mean zero and covariance matrix 
\begin{eqnarray*}
\Sigma_{\parm_0} = \left(\Ex_{\parm_0}\left(-\frac{\partial^2 \log \dY(\rY \mid \parm)}{\partial \parm
\partial \parm^\top}\right)\right)^{-1},
\end{eqnarray*}
as $N \to \infty$.
\end{thm}

\begin{proof}
Because the map $\parm \mapsto \sqrt{\dY(\ry \mid \parm)}$ is continuously
differentiable in $\parm$ for all $\ry$ in both the discrete and
absolutely continuous case and the matrix
\begin{eqnarray*}
\Ex_{\parm_0}\left(\left[\frac{\partial \log \dY(\rY \mid \parm)}{\partial
\parm}\right] \left[\frac{\partial \log \dY(\rY \mid \parm)}{\partial
\parm}\right]^\top \right)
\end{eqnarray*}
is continuous in $\parm$ as given in (\ref{f:s_exact}) and
(\ref{f:s_interval}), the transformation family $\Prob_{\rY, \Theta}$ is
differentiable in quadratic mean with Lemma 7.6 in \cite{vdVaart_1998}.
Furthermore, assumptions (A4-5) ensure that the expected Fisher information
matrix is nonsingular at $\parm_0$. With the consistency and (A3), 
the result follows from Theorem~5.39 in 
\cite{vdVaart_1998}.
\end{proof}

\begin{remark}
Assumption \textnormal{(A4)} is valid for the densities $\dZ$ of 
the normal, logistic and minimum extreme value distribution. The Fisher
information (\ref{f:F_exact}) and (\ref{f:F_interval}) evaluated at the
maximum likelihood estimator $\hat{\parm}_N$ can be used to
estimate the covariance matrix $\Sigma_{\parm_0}$.
\end{remark}

\begin{thm} \label{thm:sigma}
Under the assumptions of Theorem~\ref{thm:normality} and assuming
  $\Ex_{\parm_0}|\mF(\parm_0\mid \rY)| < \infty$, 
a consistent estimator for $\Sigma_{\parm_0}$ is given by
\begin{eqnarray*}
\hat{\Sigma}_{\parm_0, N} = \left(N^{-1} \sum_{i = 1}^{N} \mF(\hat{\parm}_N \mid \rY_i)\right)^{-1}.
\end{eqnarray*}
\end{thm}

\begin{proof}
With the law of large numbers we have
\begin{eqnarray*}
N^{-1} \sum_{i = 1}^{N} \mF(\parm_0 \mid \rY_i) = 
N^{-1} \sum_{i = 1}^{N}  -\frac{\partial^2 \log
\dY(\rY_i \mid \parm)}{\partial \parm \partial \parm^\top}
\stackrel{\Prob}{\rightarrow} 
\Ex_{\parm_0}\left(-\frac{\partial^2 \log
\dY(\rY \mid \parm)}{\partial \parm \partial \parm^\top}\right) =
\Sigma_{\parm_0}^{-1}.
\end{eqnarray*}
Because the map $\parm \mapsto \mF(\parm \mid \ry)$ is continuous for all $\ry$
(as can be seen from (\ref{f:F_exact}) and (\ref{f:F_interval})),
the result follows with Theorem~\ref{thm:consistency}.
\end{proof}
Based on Theorems \ref{thm:consistency}-\ref{thm:sigma}, we can perform standard likelihood inference on the
model parameters $\parm$.  In particular, we can construct confidence
intervals and confidence bands for the conditional distribution function
from confidence intervals and bands for the linear functions $\basisy^\top
\parm$.  We complete this part by formally defining the class of
transformation models.
\begin{defn}[Transformation model]
The triple $(\pZ, \basisy, \parm)$ is called transformation model.
\end{defn}
The transformation model $(\pZ, \basisy, \parm)$ fully defines the
distribution of $\rY$ via $\pY = \pZ \circ \basisy^\top \parm$ and thus the
corresponding likelihood $\lik(\basisy^\top \parm \mid \rY \in (\ubar{\ry},
\bar{\ry}])$.  Our definition of transformation models as $(\pZ, \basisy, \parm)$ is
strongly tied to the idea of structural inference \citep{Fraser_1968} and 
\if0\compact
group families \citep{Lehmann_1983} or group models \citep{Bickel_Klaassen_Ritov_1993}.  
\else
group models \citep{Bickel_Klaassen_Ritov_1993}.  
\fi
\cite{Fraser_1968} described a measurement model $\Prob_\rY$ for $\rY$ by an error distribution
$\Prob_\rZ$ and a structural equation $\rY = \g \circ \rZ$, where $\g$ is a
linear function, thereby extending the location-scale family
\if0\compact
$\rY = \alpha + \sigma \rZ$ introduced by \cite{Fisher_1934} and refined by \cite{Pitman_1939}. 
\else
$\rY = \alpha + \sigma \rZ$.
\fi
Group models consist of distributions generated by possibly non-linear $\g$.
The main difference to these classical approaches is that we 
parameterise $\h$ instead of $\g = \h^{-1}$.
By extending the linear transformation functions $\g$ dealt with by
\cite{Fraser_1968} to non-linear transformations, 
we approximate the potentially non-linear transformation
functions $\h = \g^{-1} = \pZ^{-1} \circ \pY$ by $\basisy^\top \parm$,
with subsequent estimation of the parameters $\parm$. For given
parameters $\parm$, a sample from $\Prob_\rY$ can be drawn by the 
probability integral transform, \ie $\rZ_1, \dots, \rZ_N \stackrel{\text{iid}}{\sim} \Prob_\rZ$ is
drawn and then $\rY_i = \inf\{\ry \in \samY \mid \basisy(\ry)^\top \parm \ge
\rZ_i\}$. 
\if0\compact
This generalises the method published by \cite{Bender_Augustin_Blettner_2005}
from the Cox model to all conditional transformation models.
\fi

\subsection{Non-parametric Inference} \label{subsec:nonparam}

For continuous responses $\rY$, any unknown transformation $\h$ can be
approximated by Bernstein polynomials of increasing order
\citep{Farouki_2012}.  For uncensored and right-censored responses and under
the same conditions for $\pZ$ as stated in Subsection~\ref{subsec:param},
\cite{McLain_Ghosh_2013} showed that the non-parametric sieve maximum
likelihood estimator is consistent with rate of convergence $N^{2/5}$ for
$\h$ with continuous bounded second derivatives in unconditional and linear
transformation models (see Subsection~\ref{subsec:ltm}).  In the latter
class, the linear shift parameters $\shiftparm$ are asympotically normal and
semi-parametrically efficient.  Numerical approximations to the observed
Fisher information $\mF(\hat{\parm}_N \mid \rY \in (\ubar{\ry}, \bar{\ry}])$
were shown to lead to appropriate standard errors of $\hat{\shiftparm}_N$ by
\cite{McLain_Ghosh_2013}.  \cite{Hothorn_Kneib_Buehlmann_2014} established
the consistency of boosted non-parametric conditional transformation models
(see Subsection~\ref{subsec:conditional}).  For sieve maximum likelihood
estimation in the class of conditional transformation models, the techniques
employed by \cite{McLain_Ghosh_2013} require minor technical extensions,
which are omitted here.

In summary, the same limiting distribution arises under both the parametric
and the non-parametric paradigm for transformation functions parameterised or
approximated using Bernstein polynomials, respectively.  In the latter case,
the target is then the best approximated transformation function with
Bernstein polynomials, say $h_N^\star$ (where the index $N$ indicates that
we use a more complex approximation when $N$ increases).  If the
approximation error $h_N^\star - h$ is of smaller order than the convergence
rate of the estimator, the estimator's target becomes the true underlying
transformation function $h$, and otherwise a bias for estimating $h$
remains.

\section{Applications} \label{sec:appl}

The definition of transformation models tailored for specific situations
``only'' requires the definition of a suitable basis function $\basisy$ and
a choice of $\pZ$.  
In this section, we will discuss specific 
transformation models for unconditional and conditional distributions of
ordered 
\if0\compact
and unordered 
\fi
categorical, discrete and continuous responses $\rY$.  Note that the
likelihood function $\lik$ allows all these models to be fitted to
arbitrarily censored or truncated responses; for brevity, we will not 
elaborate on the details.

\subsection{Unconditional Transformation Models}

\paragraph{Finite Sample Space} For ordered categorical responses $\rY$ from a
finite sample space $\samY = \{\ry_1, \dots, \ry_K\}$, we assign one
parameter to each element of the sample space except $\ry_K$.  This
corresponds to the basis function $\basisy(\ry_k) = \evec_{K - 1}(k)$, where
$\evec_{K-1}(k)$ is the unit vector of length $K - 1$, with its
$k$th element being one.  The transformation function $\h$ is
\begin{eqnarray*}
\h(\ry_k) & = & \evec_{K - 1}(k)^\top \parm = \eparm_k \in \RR, \quad 1 \le k < K, \quad
\text{st} \quad \eparm_1 < \dots <
\eparm_{K – 1},
\end{eqnarray*}
with $\h(\ry_K) = \infty$,  and the unconditional distribution function of $\pY$ is 
$\pY(\ry_k) = \pZ(\eparm_k)$. This parameterisation underlies the common
proportional odds and proportional hazards model for ordered categorical
data \citep{Tutz_2012}. Note that monotonicity of $\h$ is guaranteed by the
$K - 2$ linear constraints $\eparm_2 - \eparm_1 > 0, \dots, \eparm_{K -1} -
\eparm_{K -2} > 0$ when constrained optimisation is performed.
In the absence of censoring or truncation and with $\eparm_0 = -\infty, \eparm_K
= \infty$, we
obtain the maximum likelihood estimator for $\parm$ as
\begin{eqnarray*}
\hat{\parm}_N & = & \argmax_{\eparm_1 < \dots < \eparm_{K - 1}}
\sum_{i = 1}^N \log(\pZ(\eparm_{k(i)}) - \pZ(\eparm_{k(i) - 1})) \\
& = & \left(\pZ^{-1}\left(N^{-1} \sum_{i = 1}^N \I(\rY_i \le \ry_1)\right), \dots, 
            \pZ^{-1}\left(N^{-1} \sum_{i = 1}^N \I(\rY_i \le \ry_{K - 1})\right)\right)^\top
\end{eqnarray*}
because $\hat{\pi}_k = N^{-1}\sum_{i = 1}^N \I(\rY_i = \ry_k), 1  \le k < K$ 
maximises the equivalent multinomial (or empirical) log-likelihood
$\sum_{i = 1}^N \log(\pi_{k(i)})$,
and we can rewrite this estimator as 
\begin{eqnarray*}
\hat{\pi}_k = N^{-1}\left(\sum_{i = 1}^N \I(\rY_i
\le \ry_k) - \I(k > 1)\sum_{i = 1}^N \I(\rY_i \le \ry_{k - 1})\right), 1 \le k < K.
\end{eqnarray*}
The estimated distribution function $\hatpY = \pZ \circ \hat{\h}_N$ is invariant with respect to
$\pZ$.
Assumption (A4) is valid for these basis functions because we
have $\Ex_{\parm_0}( \evec_{K - 1}(\rY)  \evec_{K - 1}(\rY)^\top) = \diag(\Prob(\rY =
\ry_k)), 1 \le k < K$ for $\rY \sim \Prob_{\rY, \parm_0}$.

If we define the sample space $\samY$ as the set of unique observed values
and the probability measure as the empirical cumulative distribution
function (ECDF), putting mass $N^{-1}$ on each observation, we see that this
particular parameterisation is equivalent to an empirical likelihood approach
and we get $\hat{\h}_N = \pZ^{-1} \circ \text{ECDF}$.  Note
that although the transformation function depends on the choice of $\pZ$,
the estimated distribution function $\hatpY = \pZ \circ \hat{\h}_N = \text{ECDF}$ 
does not and is simply the non-parametric
empirical maximum likelihood estimator.  A
smoothed version of this estimator for continuous responses is discussed in the next paragraph.

\paragraph{Infinite Sample Space} For continuous responses $\rY$, the
parameterisation $\h(\ry) = \basisy(\ry)^\top \parm$, and thus also
$\hatpY$, should be smooth in $\ry$; therefore, any polynomial or spline basis is a
suitable choice for $\basisy$.  For the empirical experiments in
Section~\ref{sec:empeval}, we applied Bernstein polynomials \citep[for an
Overview, see][]{Farouki_2012} of order $M$
($\dimparm = M + 1$) defined on the interval $[\ubar{\imath}, \bar{\imath}]$ with 
\begin{eqnarray*}
\bern{M}(\ry) & = & (M + 1)^{-1}(f_{\text{Be}(1, M + 1)}(\tilde{\ry}), \dots, 
                            f_{\text{Be}(m, M - m + 1)}(\tilde{\ry}), \dots, 
                            f_{\text{Be}(M + 1, 1)}(\tilde{\ry}))^\top \in \RR^{M + 1} \\
\h(\ry) & = & \bern{M}(\ry)^\top \parm = 
              \sum_{m = 0}^{M} \eparm_m f_{\text{Be}(m + 1, M - m + 1)}(\tilde{\ry}) / (M + 1) \\
\h^\prime(\ry) & = & \bern{M}^\prime(\ry)^\top \parm = 
              \sum_{m = 0}^{M - 1} (\eparm_{m + 1} - \eparm_m) f_{\text{Be}(m + 1, M - m)}(\tilde{\ry}) M /
((M + 1) (\bar{\imath} - \ubar{\imath})),
\end{eqnarray*}
where $\tilde{\ry} = (\ry -\ubar{\imath}) / (\bar{\imath} - \ubar{\imath}) \in [0,
1]$ and $f_{\text{Be}(m, M)}$ is the density of the Beta distribution with
parameters $m$ and $M$.  This choice is computationally attractive because
strict monotonicity can be formulated as a set of $M$ linear constraints on the
parameters $\eparm_m < \eparm_{m + 1}$ for all $m = 0, \dots, M$
\citep{Curtis_Ghosh_2011}.  Therefore, application of constrained optimisation guarantees 
monotonic estimates $\hat{\h}_N$. The basis
contains an intercept.  We obtain smooth plug-in estimators for the
distribution, density, hazard and cumulative hazard functions as
$\hatpY = \pZ \circ \bern{M}^\top \hat{\parm}_N$, 
$\hatdY = \dZ \circ \bern{M}^\top \hat{\parm}_N \times {\bern{M}^\prime}^\top
\hat{\parm}_N$, $\hathazY = \hatdY / (1 - \hatpY)$ and 
$\hatHazY = -\log(1 - \hatpY)$.
\rev{The estimator $\hatpY = \pZ \circ \bern{M}^\top \hat{\parm}_N$ must
not be confused with the estimator $\hatpY = \bern{M}^\top \hat{\pvec}$
for $\rY \in [0, 1]$ obtained from the smoothed empirical distribution function 
with coefficients $\hat{\pvec}_{m + 1} = \sum_{i = 1}^N \I(\rY_i \le m/M)
/ N$ corresponding to probabilities evaluated at the quantiles $m / M$ for
$m = 0, \dots, M$ \citep{Babu_Canty_Chaubey_2002}.}

The question arises how the degree of the polynomial affects the estimated
distribution function.  On the one hand, the model $(\Phi, \bern{1}, \parm)$ only allows
linear transformation functions of a standard normal and $\pY$ is restricted to the
normal family.  On the other hand, $(\Phi, \bern{N - 1}, \parm)$ has one
parameter for each observation and $\hatpY$ is the non-parametric maximum
likelihood estimator $\text{ECDF}$, which, by the Glivenko-Cantelli lemma,
converges to $\pY$.  In this sense, we cannot choose a ``too large'' value for $M$. This
is a consequence of the monotonicity constraint on the estimator 
$\basisy^\top \hat{\parm}_N$, which, in this extreme case, just interpolates
the step function $\pZ^{-1} \circ \text{ECDF}$. Empirical evidence for the
insensitivity of results when $M$ is large can be found in \cite{vign:mlt.docreg}
and in the discussion. 

\subsection{Conditional Transformation Models} \label{subsec:conditional}

In the following, we will discuss a cascade of increasingly complex transformation
models where the transformation function $\h$ may depend on explanatory
variables $\rX \in \samX$.  We are interested in estimating the conditional
distribution of $\rY$ given $\rX = \rx$.  The corresponding distribution
function $\pYx$ can be written as $\pYx(\ry) = \pZ(\h(\ry \mid \rx))$.  The
transformation function $\h(\cdot \mid \rx): \samY \rightarrow \RR$ is said to
be conditional on $\rx$.  Following the arguments presented in the proof of
Corollary~\ref{thm1}, it is easy to see that for each $\rx$, there exists a
strictly monotonic transformation function $\h(\cdot \mid \rx) = \pZ^{-1} \circ
\pYx$ such that $\pYx(\ry) = \pZ(\h(\ry \mid \rx))$. Because this class of
conditional transformation models and suitable parameterisations were
introduced by \cite{Hothorn_Kneib_Buehlmann_2014}, we will only sketch the most
important aspects here.

Let $\basisx: \samX \rightarrow \RR^Q$ denote a basis transformation of the
explanatory variables. The joint basis for both $\ry$ and $\rx$ 
is called $\basisyx: \samY \times \samX \rightarrow \RR^{d(\dimparm,
\dimparmx)}$; its dimension $d(\dimparm, \dimparmx)$ depends on the way
the two basis functions $\basisy$ and $\basisx$ are combined (\eg
$\basisyx = (\basisy^\top, \basisx^\top)^\top \in
\RR^{\dimparm + \dimparmx}$ or $\basisyx = (\basisy^\top \otimes \basisx^\top)^\top
\in \RR^{\dimparm\dimparmx}$). The conditional transformation function
is now parameterised as $\h(\ry \mid \rx) = \basisyx(\ry, \rx)^\top \parm$.
One important special case is the simple transformation function
$\h(\ry \mid \rx) = \hY(\ry) + \hx(\rx)$, where the explanatory variables
only contribute a shift $\hx(\rx)$ to the conditional transformation
function. Often this shift is assumed to be linear in $\rx$; therefore, we
use the function $\m(\rx) = \basisx(\rx)^\top
\shiftparm = \tilde{\rx}^\top \shiftparm$ to denote linear shifts. 
Here $\basisx(\rx) = \tilde{\rx}$ is one row of the design matrix without
intercept. These
simple models correspond to the joint basis
$\basisyx(\ry, \rx)^\top \parm = \basisy(\ry)^\top \parm_1 +
\basisx(\rx)^\top \parm_2$, with $\hY(\ry) = \basisy(\ry)^\top \parm_1$
and $\hx(\rx) = \basisx(\rx)^\top \parm_2 = \m(\rx) = \tilde{\rx}^\top \shiftparm$.
The results presented in Section~\ref{sec:mlt}, including Theorems~\ref{thm:consistency}, 
\ref{thm:normality} and \ref{thm:sigma}, carry over in the fixed design case when $\basisy$ is
replaced by $\basisyx$.

In the rest of this section, we will present classical models that can be embedded
in the larger class of conditional transformation models and some novel models
that can be implemented in this general framework.

\subsection{Classical Transformation Models} \label{subsec:ltm}

\paragraph{Linear Model} The normal linear regression model 
$\rY \sim \ND(\alpha + \m(\rx), \sigma^2)$ with conditional distribution function
\if0\compact
\begin{eqnarray*}
\pYx(\ry) = \Phi\left(\frac{\ry - (\alpha + \m(\rx))}{\sigma}\right)
\end{eqnarray*}
\else 
$\pYx(\ry) = \Phi(\sigma^{-1}(\ry - (\alpha + \m(\rx))))$
\fi
can be understood as a transformation model with transformation function
$\h(\ry \mid \rx) = \ry / \sigma - \alpha/\sigma - \m(\rx) / \sigma$
parameterised via basis functions $\basisy(\ry) = (\ry, 1)^\top,
\basisx(\rx) = \tilde{\rx}$ and $\basisyx = (\basisy^\top,
\basisx^\top)^\top$ with parameters $\parm = (\sigma^{-1}, -\sigma^{-1}
\alpha, -\sigma^{-1} \shiftparm^\top)^\top$ under the constraint $\sigma >
0$, or in more compact notation $(\pN, (\ry, 1, \tilde{\rx}^\top)^\top,
\parm)$.  The parameters of the model are the inverse
standard deviation and the inverse negative coefficient of variation instead of the
mean and variance of the original normal distribution.
For ``exact continuous'' observations, the likelihood $\lik$ is equivalent to least-squares,
which can be maximised with respect to $\alpha$ and $\shiftparm$
without taking $\sigma$ into account. This is not possible for
censored or truncated observations, where we need to evaluate the conditional
distribution function that depends on all parameters; this model is
called Type I Tobit model 
\if0\compact
\citep{Tobin_1958} 
\fi
 (although only the likelihood changes under censoring
and truncation, but the model does not). Using an alternative
basis function $\basisyx$ would allow
arbitrary non-normal conditional distributions of $\rY$ and the
simple shift model $\basisyx(\ry, \rx)^\top \parm = \basisy(\ry)^\top \parm_1 +
\basisx(\rx)^\top \parm_2$ is then a generalisation of additive models and 
leads to the interpretation
\if0\compact
\begin{eqnarray*}
\Ex_{\rY \mid \rX = \rx}(\basisy(\rY)^\top \parm_1) = - \basisx(\rx)^\top \parm_2.
\end{eqnarray*}
\else
$\Ex_{\rY \mid \rX = \rx}(\basisy(\rY)^\top \parm_1) = - \basisx(\rx)^\top
\parm_2.$
\fi
The choice $\basisy = (1, \log)^\top$ implements the log-normal model for
$\rY > 0$. Implementation of a Bernstein basis $\basisy = \bern{M}$ allows arbitrarily shaped
distributions, \ie a transition from the normal family to the transformation
family, and thus likelihood inference on $\parm_2$ without strict assumptions on the
distribution of~$\rY$. \rev{The transformation $\bern{M}(\ry)^\top \parm_1$ must 
increase monotonically in $\ry$. Maximisation of the log-likelihood under the
linear inequality constraint $\mD_{M + 1} \parm_1 > 0$, with $\mD_{M + 1}$ representing first-order
differences, implements this requirement.}

\paragraph{Continuous ``Survival Time'' Models} 
For a continuous response $\rY > 0$, the model
\if0\compact
\begin{eqnarray*}
\pYx(\ry) = \pZ\left(\frac{\log(\ry) - (\alpha + \m(\rx))}{\sigma}\right)
\end{eqnarray*}
\else
$\pYx(\ry) = \pZ(\sigma^{-1}(\log(\ry) - (\alpha + \m(\rx))))$
\fi
with basis functions $\basisy(\ry) = (1, \log(\ry))$ and $\basisx(\rx) =
\tilde{\rx}$ and parameters $\parm = (-\alpha, \sigma^{-1}, -\shiftparm^\top )^\top$ under the constraint $\sigma > 0$ is
called the accelerated failure time (AFT) model.  The model $(\pMEV, (1, \log,
\tilde{\rx}^\top)^\top, (-\eparm_1, 1, -\shiftparm^\top)^\top)$ 
with $\sigma \equiv 1$ (and thus
fixed transformation function $\log$) is the exponential AFT model because it implies
an exponential distribution of $\rY$. When the parameter $\sigma > 0$ is
estimated from the data,  the model $(\pMEV, (1, \log,
\tilde{\rx}^\top)^\top, \parm)$ is called the Weibull model, $(\pSL, (1, \log,
\tilde{\rx}^\top)^\top, \parm)$ is the log-logistic AFT model and $(\Phi,
(1, \log, \tilde{\rx}^\top)^\top, \parm)$ is the log-normal AFT model. 
For a continuous (not
necessarily positive) response $Y$, the model $\pYx(\ry) = \pMEV(\hY(\ry) -
\m(\rx))$ is called the proportional hazards, relative risk or Cox model.  The
transformation function $\hY$ equals the log-cumulative baseline hazard and
is treated as a nuisance parameter in the partial likelihood framework, where
only the regression coefficients $\shiftparm$ are estimated.  Given
$\hat{\shiftparm}$, non-parametric maximum likelihood estimators are
typically applied to obtain $\hat{\h}_\rY$.  Here, we parameterise this
function as $\hY(\ry) = \log(\HazY(\ry)) = \basisy(\ry)^\top \parm_1$ (for
example, using $\basisy = \bern{M}$) and fit all parameters in the model
$(\pMEV, (\basisy^\top, \tilde{\rx}^\top)^\top, (\parm_1^\top,
-\shiftparm^\top)^\top)$ simultaneously.  The model is highly popular
because $\m(\rx)$ is the log-hazard ratio to
$\m(\Null)$. For the special case of right-censored survival times,
this parameterisation of the Cox model was studied theoretically and
empirically by \cite{McLain_Ghosh_2013}.
Changing the distribution function in
the Cox model from $\pMEV$ to $\pSL$ results in the proportional odds model
$(\pSL, (\basisy^\top,
\tilde{\rx}^\top)^\top,(\parm_1^\top,-\shiftparm^\top)^\top)$; its name comes from the
interpretation of $\m(\rx)$ as the constant log-odds ratio of the odds 
$\oY(\ry \mid \rX = \rx)$ and $\oY(\ry \mid \rx = \Null)$.
\rev{An additive hazards model with the conditional hazard function $\hazY(\ry \mid
\rX = \rx) = \hazY(\ry \mid \rX = \Null) - \tilde{\rx}^\top \shiftparm$ results
from the choice $\pZ(\rz) = \pExp(\rz) = 1 - \exp(-\rz)$ \citep{ArandaOrdaz_1983}
under the additional constraint $\hazY(\ry \mid \rX = \rx) > 0$.
In this case, the function $\basisy(\ry)^\top\parm_1 > 0$ is the positive baseline cumulative hazard function 
$\HazY(\ry \mid \rX = \Null)$.}

\paragraph{Discrete Models} 
For ordered categorical responses $\ry_1 < \dots < \ry_K$, the conditional distribution
$\pYx(\ry_k) = \pZ(\hY(\ry_k) - \m(\rx))$ is a transformation model with
$\basisy(\ry_k) = \evec_{K - 1}(k)$.  The model $(\pSL, (\basisy^\top,
\tilde{\rx}^\top)^\top, (\parm_1^\top, -\shiftparm^\top)^\top)$ is called the discrete proportional odds model
and $(\pMEV, (\basisy^\top, \tilde{\rx}^\top)^\top, (\parm_1^\top,
-\shiftparm^\top)^\top)$ is the discrete
proportional hazards model.  Here, $\m(\rx)$ is the log-odds ratio or log-hazard
ratio to $\m(\Null)$ independent of $k$; details are given in 
\cite{Tutz_2012}.
For the special case of a binary response ($K
= 2$), the transformation model $(\pSL, (\I(k = 1),
\tilde{\rx}^\top)^\top, (\eparm_1, -\shiftparm^\top)^\top)$ 
is the logistic regression model, $(\pN,
(\I(k = 1), \tilde{\rx}^\top)^\top, (\eparm_1, -\shiftparm^\top)^\top)$ is the probit model and
$(\pMEV, (\I(k = 1), \tilde{\rx}^\top)^\top, (\eparm_1,
-\shiftparm^\top)^\top)$ is called the
complementary log-log model.  Note that the transformation function $\hY$ is
given by the basis function $\basisy = \I(k = 1)$, \ie $\eparm_1$ is just the
intercept. The connection between standard binary regression models and
transformation models is explained in more detail by \cite{Doksum_Gasko_1990}.

\paragraph{Linear Transformation Model} The transformation model
$(\pZ, (\basisy^\top, \tilde{\rx}^\top)^\top, (\parm_1^\top,
-\shiftparm^\top)^\top)$ for any $\basisy$ and
$\pZ$ is called the linear transformation model and contains all models
discussed in this subsection.  Note that the transformation
of the response $\hY(\ry) = \basisy(\ry)^\top \parm_1$ is non-linear in all
models of interest (AFT, Cox, \etc), and the term ``linear'' only refers to a
linear shift $\m(\rx)$ of the explanatory variables.  Partially
linear or additive transformation models allow non-linear shifts
as part of a partially smooth basis $\basisx$, \ie in the form of an
additive model. The number of constraints only depends on the basis
$\basisy$ but not on the explanatory variables.

\subsection{Extension of Classical Transformation Models}

A common property of all classical transformation models is the additivity
of the response transformation and the shift, \ie the decomposition
$\h(\ry \mid \rx) = \hY(\ry) + \hx(\rx)$ of the conditional transformation
function. This assumption is relaxed by the following extensions of the
classical models. Allowing for deviations from this simple model is also the
key aspect for the development of novel transformation models in the rest of
this section.

\paragraph{Discrete Non-Proportional Odds and Hazards Models} For ordered
categorical responses, the model $\pYx(\ry_k) = \pZ(\hY(\ry_k) - \m_k(\rx))$
allows a category-specific shift $\m_k(\rx) = \tilde{\rx}^\top
\shiftparm_k$; with $\pSL$, this cumulative model is called the non-proportional odds model, and
with $\pMEV$, it is the non-proportional hazards model.  Both models can be cast into
the transformation model framework by defining the joint basis
$\basisyx(\ry_k, \rx) =
(\basisy(\ry_k)^\top, \basisy(\ry_k)^\top \otimes \basisx(\rx)^\top)^\top$ as the Kronecker
product of the two simple basis functions $\basisy(\ry_k) = \evec_{K -
1}(k)$ and $\basisx(\rx) = \tilde{\rx}$ (assuming that $\basisx$ does not
contain an intercept term). Note that the conditional transformation
function $\h(\ry \mid \rx)$ includes an interaction term between $\ry$ and
$\rx$. 
\if0\compact
{
It is also worth noting that for unordered categorical responses
$\rY \in \samY = \{\ry_1, \dots, \ry_K\}$, the multinomial model can be
estimated by the model $(\pZ, \basisyx, \parm)$ under any ordering of the
response categories because the corresponding conditional density
\begin{eqnarray*}
\dYx(\ry_k) = 
\left\{
\begin{array}{ll}
\pZ(\basisyx(\ry_k, \rx)^\top \parm) & k = 1 \\
\pZ(\basisyx(\ry_k, \rx)^\top \parm) - \pZ(\basisyx(\ry_{k - 1}, \rx)^\top \parm) & 1 < k < K \\
1 - \pZ(\basisyx(\ry_{k - 1}, \rx)^\top \parm) & k = K
\end{array}
\right.
\end{eqnarray*}
is invariant with respect to the ordering applied. For $\pSL$, this
model is called the cumulative logit or partial proportional odds model
\citep{Tutz_2012}. The classical multinomial logit (with $\pZ = \pSL$) 
or probit (with $\pZ = \Phi$) models for
an unordered response $\rY \in \samY = \{\ry_1, \dots, \ry_K\}$ 
can be written as models for the density
\begin{eqnarray*}
\Prob(\rY = \ry_1 \mid \rX = \rx) & = & \pZ(\tilde{\rx}^\top \shiftparm_1) \\
\Prob(\rY = \ry_k \mid \rX = \rx) & = & \pZ(\tilde{\rx}^\top \shiftparm_k) - 
                         \pZ(\tilde{\rx}^\top \shiftparm_{k - 1}) \quad 1 < k < K \\
\Prob(\rY = \ry_K \mid \rX = \rx) & = & 1 - \pZ(\tilde{\rx}^\top \shiftparm_{K – 1}),
\end{eqnarray*}
where the parameters $\shiftparm_k$ correspond to parameters in
a cumulative model $\Prob(\rY \le \ry_k) = \pZ(\tilde{\rx}^\top
\shiftparm_k)$ for any ordering of the sample space. Of course, the
parameter estimates $\hat{\shiftparm}_k$ change when the ordering changes
and therefore must not be interpreted directly, 
but the estimated densities $\hat{\Prob}(\rY = \ry_k \mid \rX = \rx)$ are invariant
with respect to the ordering applied.
}
\fi

\paragraph{Time-Varying Effects} One often studied extension of the Cox
model is $\pYx(\ry) = \pZ(\hY(\ry) - \tilde{\rx}^\top \shiftparm(\ry))$,
where the regression coefficients $\shiftparm(\ry)$ may change with time
$\ry$.  The Cox model is included with $\shiftparm(\ry) \equiv \shiftparm$,
and the model is often applied to check the proportional hazards assumption. 
With a smooth parameterisation of time $\ry$, for example via $\basisy =
\bern{M}$, and linear basis $\basisx(\rx) = \tilde{\rx}$, the transformation
model $(\pMEV, (\basisy^\top, \basisy^\top \otimes \basisx^\top)^\top,
\parm)$ implements this Cox model with time-varying (linear) effects. This
model (with arbitrary $\pZ$) has also been presented in \cite{Foresi_Peracchi_1995} and 
is called distribution regression in \cite{Chernozhukov_2013}. 

\subsection{Novel Transformation Models}

Due to the broadness of the transformation family, it is straightforward to
set up new models for interesting situations by allowing more complex
transformation functions $\h(\ry \mid \rx)$. We will illustrate this
possibility for two simple cases – the independent two-sample situation
and regression models for count data.  The generic and most complex
transformation model is called the conditional transformation model and is explained
at the end of this section.

\paragraph{Beyond Shift Effects} Assume we observe samples from two groups
$A$ and $B$ and want to model the conditional distribution functions
$\pYA(\ry)$ and $\pYB(\ry)$ of the response $\rY$ in the two groups.  Based
on this model, it is often interesting to infer whether the two distributions are
equivalent and, if this is not the case, to characterise how they differ. 
Using an appropriate basis function $\basisy$ and the basis $\basisx(\rx) =
(1, \I(B))^\top$, the model $(\pZ, (\basisy^\top \otimes
\basisx^\top)^\top, \parm)$ parameterises the conditional transformation
function as $\h(\ry \mid A) = \basisy(\ry)^\top \parm_1$ and 
$\h(\ry \mid B) = \h(\ry \mid A) + \h_{B - A}(\ry)
= \basisy(\ry)^\top \parm_1 + \I(B)\basisy(\ry)^\top \parm_2$.  
Clearly, the second term is constant zero ($\h_{B - A}(y) \equiv 0$) iff the two distributions 
are equivalent ($\pYA(y) = \pYB(y)$ for all $\ry$).  For the deviation
function $\h_{B -
A}(y) = \basisy^\top \parm_2$, we can apply standard likelihood inference
procedures for $\hat{\parm}_2$ to construct a confidence band 
or use a test statistic like $\max(\hat{\parm}_2 /
\text{se}(\hat{\parm}_2))$ to assess deviations from zero.  If there is
evidence for a group effect, we can use the model to check whether the deviation
function is constant, \ie $\h_{B - A}(y) \equiv c \neq 0$.  In this case, the
simpler model $(\pZ, (\basisy^\top, \I(B))^\top, (\parm_1^\top,
-\beta)^\top)$ with shift $\beta = -\parm_2$ might be easier to interpret. 
This model actually corresponds to a normal ANOVA model with $\pZ = \Phi$ and $\basisy(\ry)^\top
= (1, \ry)^\top$ or the Cox proportional hazards model with $(\pMEV,
(\bern{M}^\top, \I(B))^\top, (\parm_1^\top, -\beta)^\top)$.

\paragraph{Count Regression ``Without Tears''} Simple models for count data
$\samY = \{0, 1, 2, \dots \}$ almost always suffer from over-dispersion or excess
zeros. The linear transformation model
\if0\compact
\begin{eqnarray*}
\pYx(\ry) = \pZ(\hY(\ry) - \m(\rx))
\end{eqnarray*}
\else
$\pYx(\ry) = \pZ(\hY(\ry) - \m(\rx))$
\fi
can be implemented using the basis function $\basisy(\ry) = 
\bern{M}(\lfloor \ry \rfloor)$, and then the parameters of the transformation 
model $(\pZ, (\basisy^\top,
\tilde{\rx}^\top)^\top, \parm)$ are not affected by over- or
under-dispersion because higher moments are handled by $\hY$ independently
of the effects of the explanatory variables $\m(\rx)$. If there are excess zeros, we
can set up a joint transformation model 
\if0\compact
\begin{eqnarray*}
\pYx(\ry) = \pZ(\hY(\ry) - \m(\rx) + \I(\ry = 0) (\alpha_0 - \m_0(\rx)))
\end{eqnarray*}
\else
$\pYx(\ry) = \pZ(\hY(\ry) - \m(\rx) + \I(\ry = 0) (\alpha_0 - \m_0(\rx)))$
\fi
such that we have a two-components mixture model consisting of the count distribution
\if0\compact
\begin{eqnarray*}
\pYx(y) = \pZ(\hY(\ry) - \m(\rx)), \ry \in \samY
\end{eqnarray*}
\else
$\pYx(y) = \pZ(\hY(\ry) - \m(\rx))$ for $y \in \samY$
\fi
and the probability of an excess zero
\begin{eqnarray*}
\dYx(0) = \pZ(\hY(0) - \m(\rx) + (\alpha_0 - \m_0(\rx)))
 = \pZ(\hY(0) + \alpha_0 - \tilde{\rx}^\top(\shiftparm + \shiftparm_0))
\end{eqnarray*}
when $\m_0(\rx) = \tilde{\rx}^\top \shiftparm_0$. Hence, the transformation analogue to a hurdle model with hurdle at zero is the transformation model
$(\pZ, (\basisy^\top, \tilde{\rx}^\top, \I(y = 0), \I(y = 0) \tilde{\rx}^\top)^\top, 
(\parm_1^\top, \shiftparm^\top, \alpha_0, \shiftparm_0^\top)^\top)$. 

\paragraph{Conditional Transformation Models} When the conditional
transformation function is parameterised by multiple basis functions
$\basisy_j(\ry), \basisx_j(\rx), j = 1, \dots J$ via the joint basis
\begin{eqnarray*}
\basisyx = (\basisy_1^\top \otimes \basisx_1^\top, \dots, \basisy_J^\top \otimes
\basisx_J^\top)^\top,
\end{eqnarray*}
models of the class $(\cdot, \basisyx, \parm)$ are called conditional transformation
models (CTMs) with $J$ partial transformation functions parameterised as $\basisy_j^\top
\otimes \basisx_j^\top$ and include all special cases discussed in this
section. 
It is convenient to assume monotonicity for each of the partial
transformation functions; thus, the linear constraints for
$\basisy_j$ are repeated for each basis function in $\basisx_j$
\rev{\citep[detailed descriptions of linear constraints for different models
in this class are available in][]{vign:mlt.docreg}.}
\if1\blind
\cite{Hothorn_Kneib_Buehlmann_2014} 
\else
ANONYMOUS
\fi
introduced this general model class and
proposed a boosting algorithm for the estimation of transformation functions
$\h$ for ``exact continuous'' responses $\rY$. In the
likelihood framework presented here, conditional transformation models can
be fitted under arbitrary schemes of censoring and truncation and classical
likelihood inference for the model parameters $\parm$ becomes feasible. Of
course, unlike in the boosting context, the number of model terms $J$ and their 
complexity are limited in the likelihood world because the likelihood does not 
contain any penalty terms that induce smoothness in the $\rx$-direction. 
\if0\compact
A more detailed overview on the class of conditional transformation models
can be found in \citep{Moest_Schmid_Faschingbauer_2014}.
\fi

\rev{A systematic overview of linear transformation models with potentially
response-varying effects is given in Table~\ref{overview}. Model 
nomenclature and interpretation of the corresponding model parameters is
mapped to specific transformation functions $\h$ and distribution
functions $\pZ$. To the best of our knowledge, models without names have not yet
been discussed in the literature, and their specific properties 
await closer investigation.}

\begin{sidewaystable}
\begin{center}
\begin{tabular}{lllcccc}
\toprule
& & & \multicolumn{4}{c}{$\pZ$} \\
& & & $\Phi$ & $\pSL$ & $\pExp$ & $\pMEV$ \\
\cmidrule{4-7}
$\samY$ & $\h$ & Meaning of & & & &  \\ 
\midrule
$K = 2$ & \multicolumn{5}{c}{\textbf{Binary Regression}} \\ 
& $\I(k = 1) \eparm_1 - \tilde{\rx}^\top \shiftparm$ & & probit BGLM & logistic BGLM & clog BGLM & cloglog BGLM \\
& & $\eparm_1$ &  & $\log(\oY(\ry \mid \rX = \Null))$  & $\HazY(\ry \mid \rX = \Null)$ & $\log(\HazY(\ry \mid
\rX = \Null))$ \\
& & $\shiftparm$ &  & log-OR & AH & log-HR \\

$K > 2$ &  \multicolumn{5}{c}{\textbf{Polytomous Regression}} \\ 
 & $ \evec_{K - 1}(k)^\top \parm_1 - \tilde{\rx}^\top \shiftparm$ & & & discrete PO & & discrete PH \\
& & $\evec_{K - 1}(k)^\top \parm_1$ & & $\log(\oY(\ry \mid \rX = \Null))$  & $\HazY(\ry \mid \rX = \Null)$ & $\log(\HazY(\ry \mid \rX = \Null))$ \\
& & $\shiftparm$ & & log-OR & AH & log-HR \\

  & $ \evec_{K - 1}(k)^\top \parm_1 - \tilde{\rx}^\top \shiftparm(k)$ & &  & non-PO & & non-PH \\
& & $\evec_{K - 1}(k)^\top \parm_1$ & & $\log(\oY(\ry \mid \rX = \Null))$  & $\HazY(\ry \mid \rX = \Null)$ & $\log(\HazY(\ry \mid \rX = \Null))$ \\

$\mathbb{N}$  &  \multicolumn{5}{c}{\textbf{Count Regression}} \\ 
& $ \bern{M}(\lfloor \ry \rfloor)^\top \parm_1 - \tilde{\rx}^\top \shiftparm$ & &  & & &  \\
& & $\bern{M}(\lfloor \ry \rfloor)^\top \parm_1$ & & $\log(\oY(\ry \mid \rX = \Null))$  & $\HazY(\ry \mid \rX = \Null)$ & $\log(\HazY(\ry \mid \rX = \Null))$ \\
& & $\shiftparm$ & & log-OR & AH & log-HR \\

  & $ \bern{M}(\lfloor \ry \rfloor)^\top \parm_1 - \tilde{\rx}^\top \shiftparm(\lfloor \ry \rfloor)$ & &  & & &  \\
& & $\bern{M}(\lfloor \ry \rfloor)^\top \parm_1$ & & $\log(\oY(\ry \mid \rX = \Null))$  & $\HazY(\ry \mid \rX = \Null)$ & $\log(\HazY(\ry \mid
\rX = \Null))$ \\

$\RR^+$  &  \multicolumn{5}{c}{\textbf{Survival Analysis}} \\ 
 & $ \log(\ry) + \eparm_1 - \tilde{\rx}^\top \shiftparm$ & & & & & Exponential AFT  \\
& & $\shiftparm$ & & log-OR & & log-HR \\

         & $ (1, \log(\ry))^\top \parm_1 - \tilde{\rx}^\top \shiftparm$ & & log-normal AFT & log-logistic AFT & & Weibull AFT   \\
& & $(1, \log(\ry))^\top \parm_1$ & & $\log(\oY(\ry \mid \rX = \Null))$ & $\HazY(\ry \mid \rX = \Null)$ & $\log(\HazY(\ry \mid \rX = \Null))$ \\
& & $\shiftparm$ & & log-OR & AH & log-HR \\

$\RR$  &  \multicolumn{5}{c}{\textbf{Continuous Regression and Survival Analysis}} \\ 
& $ \ry \eparm_1 - \eparm_2 - \tilde{\rx}^\top \shiftparm$ & & NLRM &  & &   \\
&  & $ \eparm_1^{-2} $ & variance &  & &   \\
&  & $ (\eparm_2 + \tilde{\rx}^\top \shiftparm) / \eparm_1 $ & mean &  & &   \\

      & $ \bern{M}(\ry)^\top \parm_1 - \tilde{\rx}^\top \shiftparm$ & &  & & Aalen AH & Cox PH  \\
& & $\bern{M}(\ry)^\top \parm_1 $ &  & $\log(\oY(\ry \mid \rX = \Null))$ & $\HazY(\ry \mid \rX = \Null)$ & $\log(\HazY(\ry \mid \rX = \Null))$ \\
& & $\shiftparm$ &  & log-OR & AH & log-HR \\

      & $ \bern{M}(\ry)^\top \parm_1 - \tilde{\rx}^\top \shiftparm(\ry)$ & & \multicolumn{3}{|c|}{Distribution Regression} & Time-varying Cox   \\
& & $\bern{M}(\ry)^\top \parm_1 $ &  & $\log(\oY(\ry \mid \rX = \Null))$ & $\HazY(\ry \mid \rX = \Null)$ & $\log(\HazY(\ry \mid \rx = \Null))$ \\

\bottomrule
\end{tabular}
\caption{Non-exhaustive Overview of Conditional Transformation Models. Abbreviations: Proportional hazards (PH),
proportional odds (PO), additive hazards (AH), odds ratio (OR), hazard ratio (HR), complementary log (clog), complementary log-log (cloglog),
normal linear regression model (NLRM), binary generalised linear model (BGLM), accelerated failure time
(AFT). \label{overview}}
\end{center}
\end{sidewaystable}

\section{Empirical Evaluation} \label{sec:empeval}

We will illustrate the range of possible applications of likelihood-based
conditional transformation 
\if0\compacteval
{
models. In Subsection~\ref{subsec:sim}, we will present a small
simulation experiment highlighting the possible advantage of indirectly
modelling conditional distributions with transformation functions.

\subsection{Illustrations}

}
\else
{
models. Further applications and a simulation study can be found in 
\if0\blind
ANONYMOUS.
\else
\cite{Hothorn_Moest_Buehlmann_2016}.
\fi
}
\fi

\paragraph{Density Estimation: Old Faithful Geyser}

The duration of eruptions and the waiting time between eruptions of the Old
Faithful geyser in the Yellowstone National Park became a standard benchmark
for non-parametric density 
\if0\compact
estimation \citep[the original data were given by][]{Azzalini_Bowman_1990}. 
\else
estimation.
\fi
The nine parameters of the transformation model $(\Phi,
\bern{8}(\text{waiting}), \parm)$ were fitted by maximisation of the approximate
log-likelihood (treating the waiting times as ``exact'' observations) under
the eight linear constraints \rev{$\mD_9 \parm > 0$}. The model depicted in 
Figure~\ref{fig_1_geyser-u}A reproduces the classic bimodal unconditional
density of waiting time along with a kernel density estimate.  It is
important to note that the transformation model was fitted
likelihood-based, whereas the kernel density estimate relied on a
cross-validated bandwidth.  An unconditional density estimate for the
duration of the eruptions needs to deal with censoring because exact duration times
are only available for the daytime measurements.  At night, the
observations were either left-censored (``short'' eruption),
interval-censored (``medium'' eruption) or right-censored 
\if0\compact
(``long'' eruption) as explained by \cite{Azzalini_Bowman_1990}.  
\else
(``long'' eruption).
\fi
This censoring was widely
ignored in analyses of the Old Faithful data because most non-parametric
kernel techniques cannot deal with 
\if0\compact
censoring \citep[see for example][]{Hyndman_Yao_2002}.  
\else
censoring.
\fi
We applied the transformation model $(\Phi,
\bern{8}(\text{duration}), \parm)$ based on the exact log-likelihood function
under eight linear constraints and obtained the unconditional density depicted in 
Figure~\ref{fig_1_geyser-u}B. In \cite{vign:mlt.docreg}, results for
$M = 40$ are computed, which led to almost identical estimates of the distribution function.

\begin{figure}[t]
\begin{center}
\includegraphics{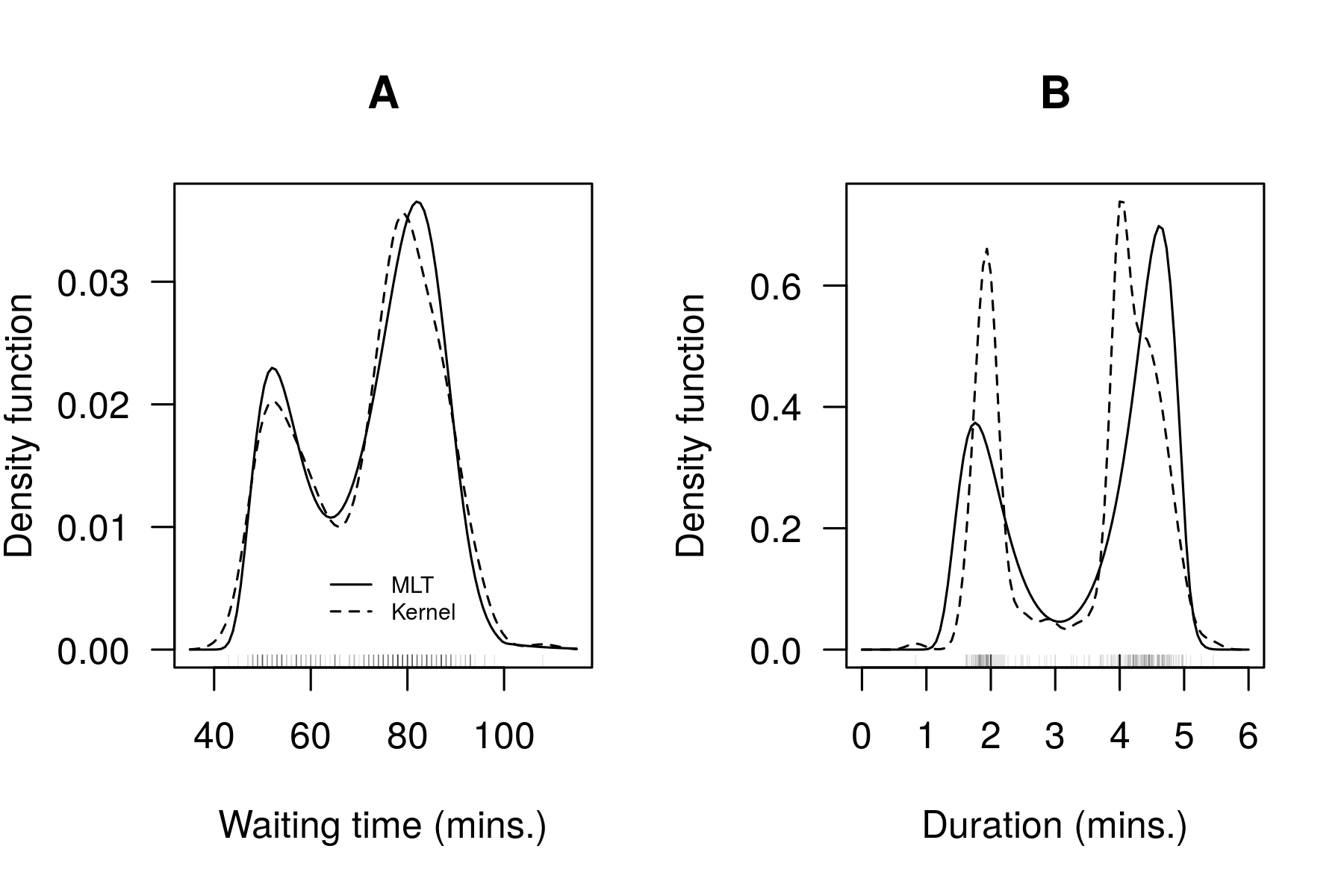}
\caption{Old Faithful Geyser. Estimated density for waiting times (A) between and duration (B) of 
         eruptions 
         by the most likely transformation model (MLT)
         and kernel 
\if0\compact
smoothing (function \code{npudens()} from package \pkg{np} in \proglang{R}). 
\else
smoothing.
\fi
         Note that the kernel estimator was based on the imputed duration times $2, 3$ and $4$ for
         short, medium and long eruptions at night (as are the rugs in B).
         \label{fig_1_geyser-u}}
\end{center}
\end{figure}

\if0\compact
{
In addition, we modelled the conditional distribution of the censored
duration times given the waiting times using the transformation model $(\Phi,
(\bern{7}(\text{duration})^\top \otimes \, \bernx{3}(\text{waiting})^\top)^\top,
\parm)$.  This conditional transformation model allows a smooth conditional
distribution function of duration smoothly varying with waiting time and was fitted by
maximisation of the exact log-likelihood under $7 \times 4$ linear constraints. The
corresponding conditional density in Figure~\ref{fig_2_geyser} shows that
the marginal bimodality can be explained by relatively long durations for
short waiting times and two clusters of short and long durations after
waiting times longer than $70$ minutes. Note that in this model, 
the choice of $\pZ$ does not influence the estimated 
conditional distribution or density function if the parameterisation of 
$\h$ is flexible enough to compensate for this change.

\begin{figure}[t]
\begin{center}
\includegraphics{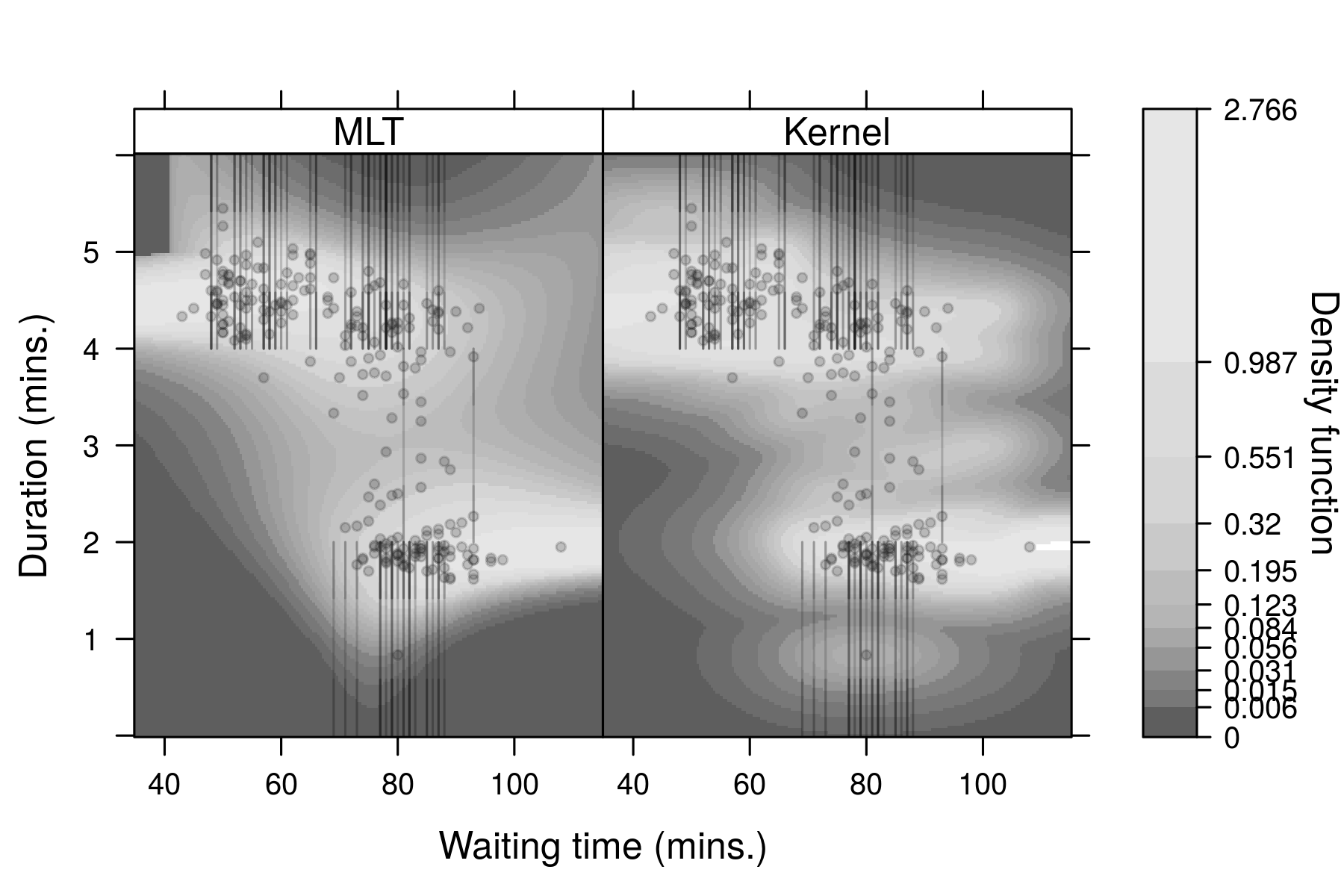}
\caption{Old Faithful Geyser. Conditional distribution of duration given waiting time
         estimated with a most likely transformation model (MLT) and a conditional 
         kernel density estimator (function \code{npcdens()} from package \pkg{np}). 
         Exact observations are given as dots, censored observations as lines.
         \label{fig_2_geyser}}
\end{center}
\end{figure}

}
\fi

\paragraph{Quantile Regression: Head Circumference}

The Fourth Dutch Growth Study 
\if0\compact
\citep{Fredriks_Buuren_Burgmeijer_2000} 
\fi
is a
cross-sectional study on growth and development of the Dutch population
younger than $22$ years.  \cite{Stasinopoulos_Rigby_2007} fitted
a growth curve to head circumferences (HC) of $7040$ boys using a GAMLSS
model with a Box-Cox $t$ distribution describing the first four moments of
head circumference conditionally on age.  The model showed evidence of
kurtosis, especially for older boys.  We fitted the same growth curves by the
conditional transformation model $(\Phi, (\bern{3}(\text{HC})^\top \otimes
\bernx{3}(\text{age}^{1/3})^\top)^\top, \parm)$ by maximisation of the approximate
log-likelihood under $3 \times 4$ linear constraints \rev{$(\mD_4 \otimes
\mI_4) \parm > 0$}. Figure~\ref{fig_4_HC} shows the
data overlaid with quantile curves obtained via inversion of the estimated
conditional distributions.  The figure very closely reproduces the growth
curves presented in Figure~16 of \cite{Stasinopoulos_Rigby_2007} and also
indicates a certain asymmetry towards older boys.

\begin{figure}[t]
\begin{center}
\includegraphics{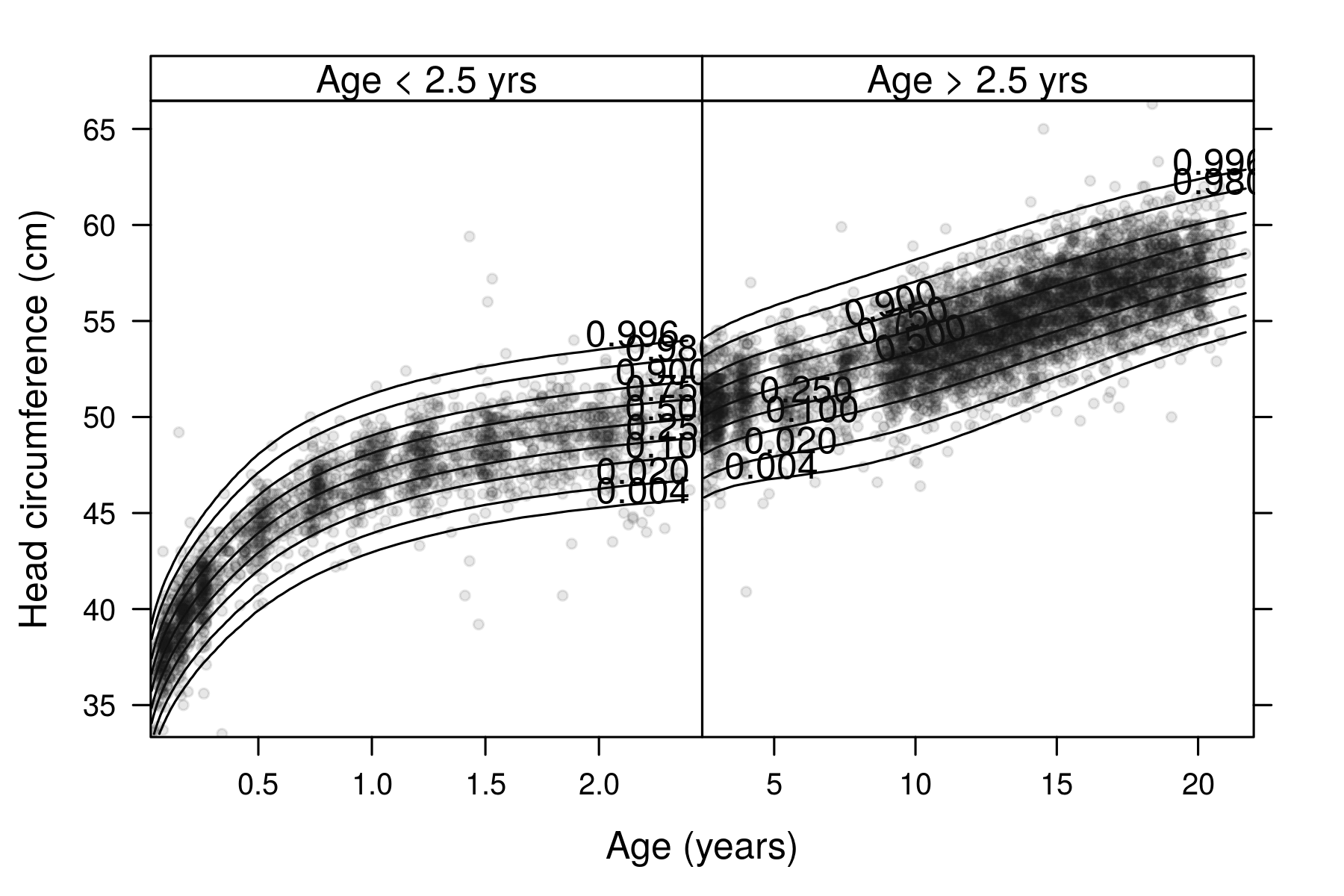}
\caption{Head Circumference Growth. Observed head circumference and age of
         $7040$ boys with estimated quantile curves for
         $p = 0.04, 0.02, 0.1, 0.25, 0.5, 0.75, 0.9, 0.98, 0.996$.
         \label{fig_4_HC}}
\end{center}
\end{figure}

\paragraph{Survival Analysis: German Breast Cancer Study Group-2 Trial} 

This prospective, controlled clinical trial on the treatment of node-positive
breast cancer patients was conducted by the German Breast Cancer Study 
\if0\compact
Group \citep[GBSG-2,][]{gbsg2:1994}.  Patients not older than $65$ years with
positive regional lymph nodes but no distant metastases were included in
the study.  
\else
Group.
\fi
Out of $686$ women, $246$ received hormonal therapy, whereas the
control group of $440$ women did not.  Additional
variables include age, menopausal status, tumour size, tumour grade, number of
positive lymph nodes, progesterone receptor and oestrogen receptor.  The
right-censored recurrence-free survival time is the response variable of
interest.

The Cox model $(\pMEV, (\bern{10}^\top, \I(\text{hormonal
therapy}))^\top, \parm)$ implements the transformation function $\h(\ry \mid
\text{treatment}) = \bern{10}(\ry)^\top \parm_1 + \I(\text{hormonal therapy})
\beta$, where $\bern{10}^\top \parm_1$ is the log-cumulative baseline hazard function
parameterised by a Bernstein polynomial and $\beta \in \RR$ is the
log-hazard ratio of hormonal therapy. This is the classical Cox model with one treatment parameter $\beta$ 
but with fully parameterised baseline transformation function, which was fitted by the exact
log-likelihood under ten linear constraints. The model assumes proportional hazards, an assumption whose
appropriateness we wanted to assess using the non-proportional hazards model
$(\pMEV, (\bern{10}^\top \otimes (1, \I(\text{hormonal therapy})))^{\top}, \parm)$ with the
transformation function 
\begin{eqnarray*}
\h(\ry \mid \text{treatment}) = \bern{10}(\ry)^\top \parm_1 + \I(\text{hormonal therapy}) \bern{10}(\ry)^\top \parm_2. 
\end{eqnarray*}
\rev{The function $\bern{10}(\ry)^\top \parm_2$ is the time-varying difference of the
log-hazard functions of women without and with hormonal therapy
and can be interpreted as the deviation from a constant log-hazard ratio treatment effect
of hormonal therapy.}  Under the null hypothesis of no
treatment effect, we would expect $\parm_2 \equiv \bold{0}$.  This monotonic
deviation function adds ten linear constraints \rev{$\mD_{11} \parm_1 + \mD_{11}
\parm_2 > 0$, which also ensure monotonicity of the transformation function for
treated patients.} We first
compared the fitted survivor functions obtained from the model including 
a time-varying treatment effect with the Kaplan-Meier
estimators in both treatment groups.  
Figure~\ref{fig_5_logrank}A shows a nicely smoothed version of the survivor
functions obtained from this transformation model.  
Figure~\ref{fig_5_logrank}B shows the time-varying treatment effect
$\bern{10}(\ry)^\top \hat{\parm}_2$, together with a $95\%$ confidence band
computed from the joint normal distribution of $\hat{\parm}_2$ for a grid
over 
\if0\compact
time as described by \cite{Hothorn_Bretz_Westfall_2008}; 
\else
time;
\fi
the method is
much simpler than other methods for inference on time-varying
effects \citep[\eg][]{Sun_Sundaram_Zhao_2009}.  The $95\%$
confidence interval around the log-hazard ratio $\hat{\beta}$ is also plotted, 
and as the latter is fully covered by the confidence band for
the time-varying treatment effect, there is no reason to question the
treatment effect computed under the proportional hazards assumption.
\if0\compact
An alternative method for this type of analysis has been recently suggested by
\cite{Yang_Prentice_2015}.
\fi

\begin{figure}[t]
\begin{center}
\includegraphics{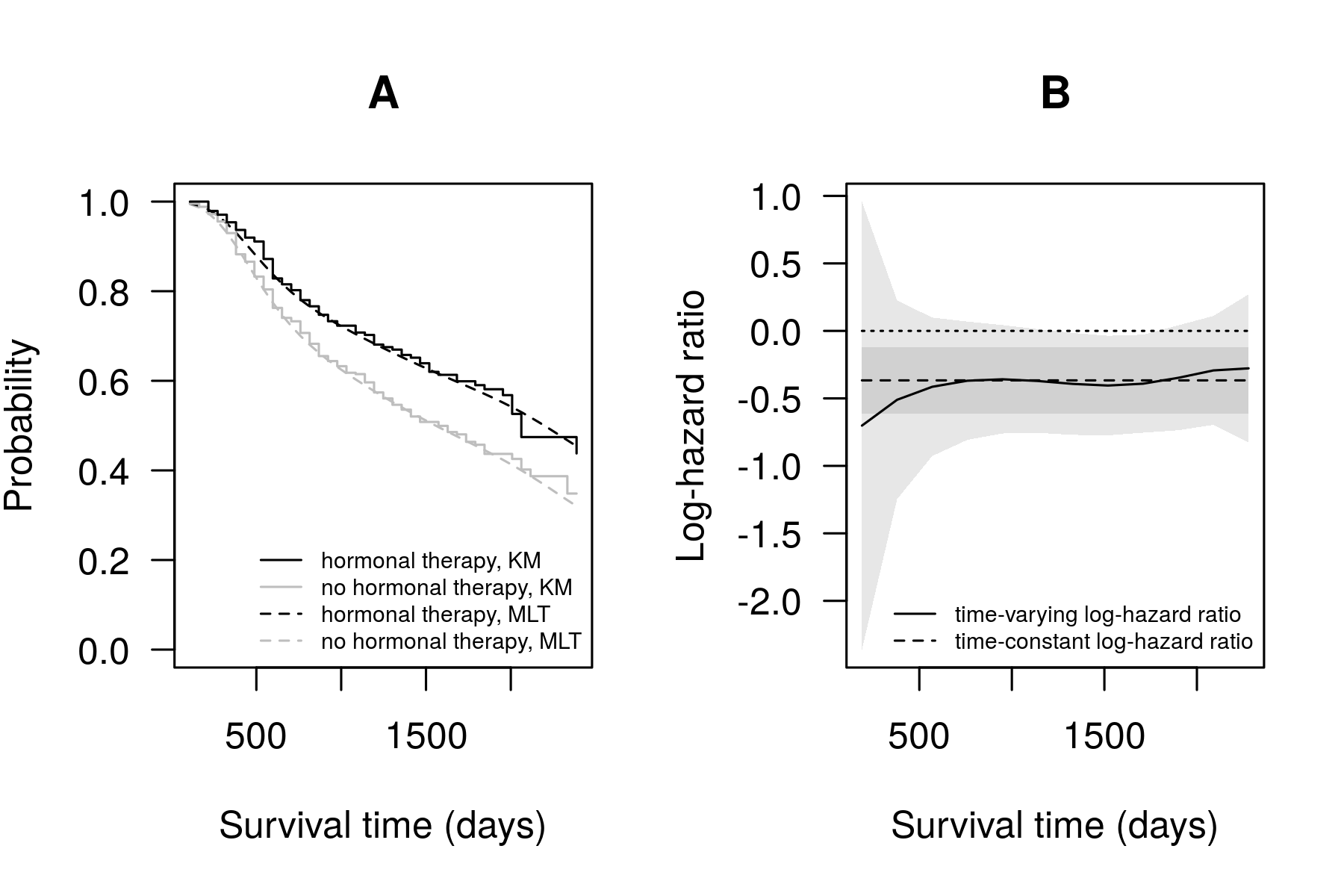}
\caption{GBSG-2. Estimated survivor functions 
         by the most likely transformation model (MLT) and the Kaplan-Meier (KM) estimator in the two 
         treatment groups (A). 
         Verification of proportional hazards (B): The log-hazard ratio $\hat{\beta}$
         (dashed line) with $95\%$ confidence interval (dark grey) is fully 
         covered by a $95\%$ confidence band for the time-varying treatment effect
         (the time-varying log-hazard ratio is in light grey, the estimate is the solid line)
         computed from a non-proportional hazards model. \label{fig_5_logrank}}
\end{center}
\end{figure}

In the second step, we allowed an age-varying treatment effect to be included
in the model $(\pMEV, (\bern{10}(\ry)^\top \otimes
       (\I(\text{hormonal therapy}), 1 - \I(\text{hormonal therapy}))
       \otimes \bernx{3}(\text{age})^\top)^\top, \parm)$. For both treatment
groups, we estimated a conditional transformation function of survival time
$\ry$ given age parameterised as the tensor basis of two Bernstein bases. Each of the
two basis functions comes with $10 \times 3$ linear constraints; therefore, the model
was fitted under $60$ linear constraints. 
Figure~\ref{fig_6_GBSG2} allows an assessment of the prognostic and
predictive properties of age.  As the survivor functions were clearly larger
for all patients treated with hormones, the positive treatment effect
applied to all patients.  However, the size of the treatment effect varied
greatly.  The effect was most pronounced for women younger than $30$ and
levelled off a little for older patients.  In general, the survival times were
longest for women between $40$ and $60$ years old.  Younger women suffered
the highest risk; for women older than $60$ years, the risk started to
increase again.  This effect was shifted towards younger women when hormonal treatment was applied.

\begin{figure}[t]
\begin{center}
\includegraphics{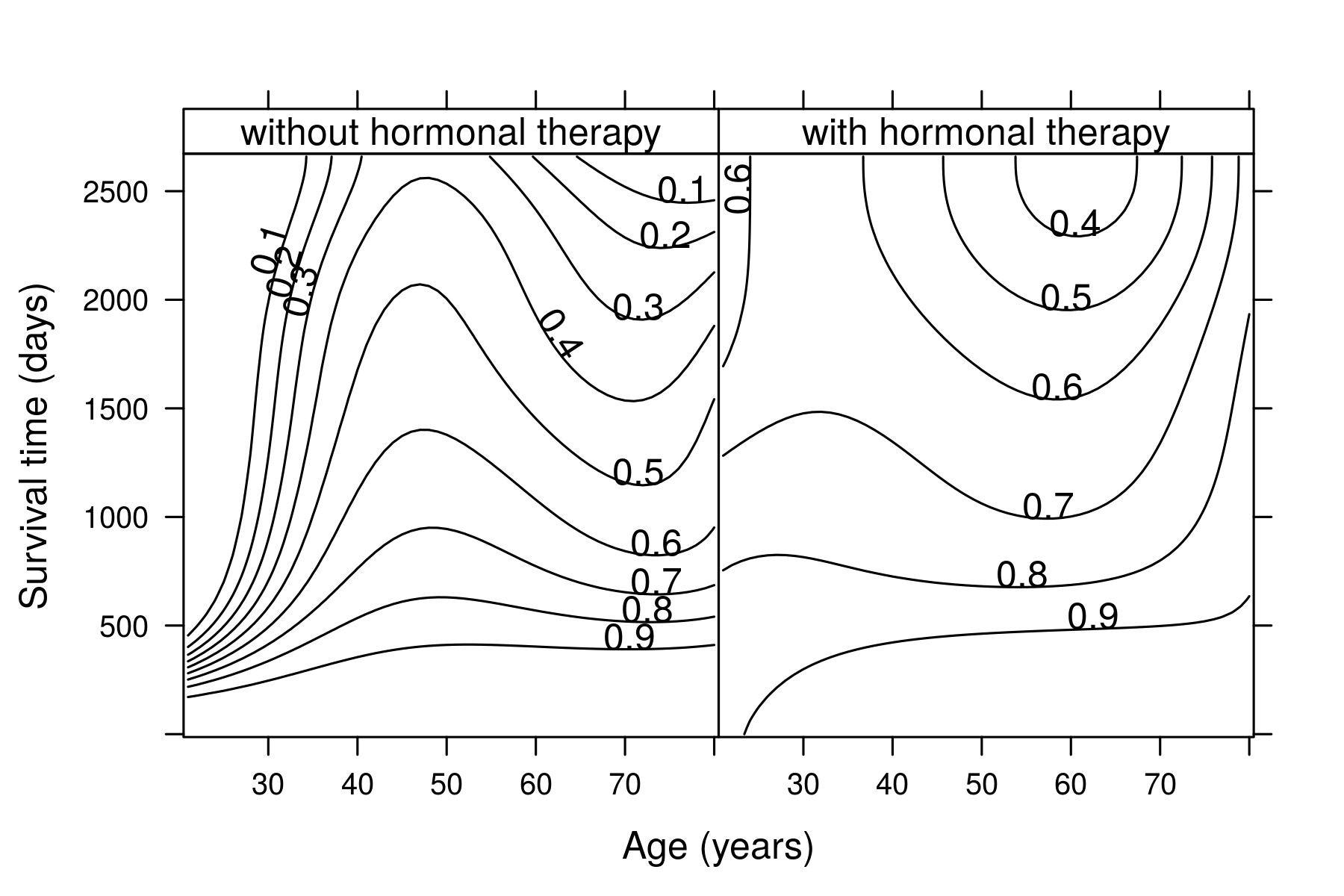}
\caption{GBSG-2. Prognostic and predictive effect of age. The contours depict the
         conditional survivor functions given treatment and age of the patient. 
         \label{fig_6_GBSG2}}
\end{center}
\end{figure}

%
%
%

\if0\compact
{

\paragraph{Count Regression: Tree Pipit Counts}

\cite{Mueller_Hothorn_2004} reported data on the number of tree pipits
\textit{Anthus trivialis}, a small passerine bird, counted on $86$ forest plots in a light gradient
ranging from open and sunny stands (small cover storey) to dense and dark
stands (large cover storey).  We modelled the conditional distribution of the
number of tree pipits on one plot given the cover storey on this plot 
by the transformation
model $(\Phi, (\basisy^\top \otimes \bernx{4}(\text{cover storey})^\top)^\top,
\parm)$, where $\basisy(y) = \evec_5(y + 1), y = 0, \dots, 4$; the model
was fitted under $4 \times 5$ linear constraints. In this
model of count data, the conditional distribution depends on
both the number of counted birds and the cover storey and the effect of
cover storey may change with different numbers of birds observed.  
Figure~\ref{fig_8_treepipit}A depicts the observations, and
Figure~\ref{fig_8_treepipit}B shows the conditional distribution function evaluated for $0,
\dots, 5$ observed birds.  The conditional distribution function obtained
from a generalised additive Poisson (GAM) model with smooth mean effect of cover
storey is given in Figure~\ref{fig_8_treepipit}C.  Despite some overfitting, this model is
more restrictive than our transformation model because one mean function determines the whole distribution
(the local minima of the conditional distributions as a function of cover storey were
constant in Figure~\ref{fig_8_treepipit}C, whereas they were shifted towards higher values of
cover storey in Figure~\ref{fig_8_treepipit}B).

\begin{figure}[t]
\begin{center}
\includegraphics{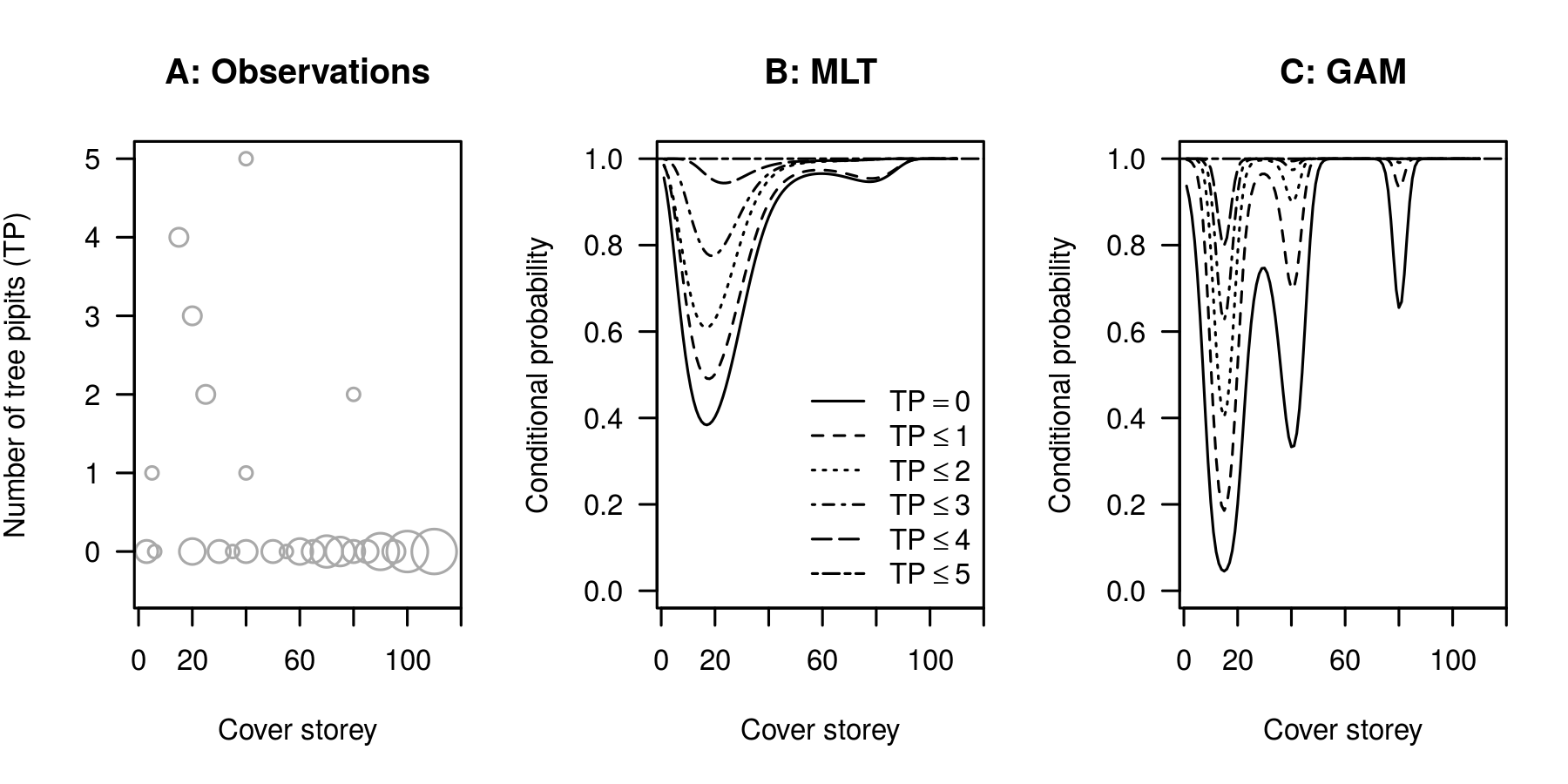}
\caption{Tree Pipit Counts. Observations (A, the size of the points is
         proportional to the number of observations) and estimated conditional distribution
         of number of tree pipits given cover storey by the most likely transformation model (MLT, B)
         and a generalised additive Poisson model (function \code{gam()} in package \pkg{mgcv}, 
         GAM, C). \label{fig_8_treepipit}}
\end{center}
\end{figure}

}
\fi

\if0\compacteval
{

\subsection{Simulation Experiment} \label{subsec:sim}

The transformation family includes linear as well as very flexible models,
and we therefore illustrate the potential gain of
modelling a transformation function $\h$ by comparing a very simple
transformation model to a fully parametric approach and to a non-parametric approach
using a data-generating process introduced by
\cite{Hothorn_Kneib_Buehlmann_2014}.

In the transformation model $(\Phi, ((1, \ry) \otimes (1, \rx^\top))^\top, \parm)$,
two explanatory variables $\rx = (\erx_1, \erx_2)^\top$
influence both the conditional mean and the conditional variance of
a normal response $\rY$. 
Although the transformation function is linear in
$\ry$ with three linear constraints, the mean and variance of $\rY$ given $\rx$
depend on $\rx$ in a non-linear way. The choices
$\erx_1 \sim \UD[0, 1], \erx_2 \sim \UD[-2, 2]$ with
$\parm = (0, 0, -1, .5, 1, 0)$ lead to the heteroscedastic
varying coefficient model
\begin{eqnarray} \label{simmod}
\rY = \frac{1}{\erx_1 + 0.5}\erx_2 + 
      \frac{1}{\erx_1 + 0.5}\varepsilon, \quad \varepsilon \sim \ND(0, 1),
\end{eqnarray}
where the variance of $\rY$ ranges between $0.44$ and $4$ depending on  
$\erx_1$. This model can be fitted in the GAMLSS framework under the       
assumptions that the mean of the normal response depends on a    
smoothly varying regression coefficient $(\erx_1 + 0.5)^{-1}$ for $\erx_2$ and
that the variance is a smooth function of $\erx_1$. This model is therefore fully parametric. 
As a non-parametric counterpart, we used a kernel estimator for
estimating the conditional distribution function of $\rY$ 
as a function of the two explanatory variables.

From the transformation model, the GAMLSS and kernel estimators, we obtained estimates of
$\pYx(\ry)$ over a grid on
$\ry, \erx_1, \erx_2$ and computed the mean absolute deviation (MAD)
of the true and estimated probabilities
\begin{eqnarray*}
\text{MAD}(\erx_1, \erx_2) = \frac{1}{n} \sum_{\ry}
    |\pYx(\ry) - \hatpYx(\ry)| 
\end{eqnarray*}
for each pair of $\erx_1$ and $\erx_2$. Then, the minimum, the median and
the maximum of the MAD values for all $\erx_1$ and $\erx_2$ were computed as
summary statistics.  The most likely transformation approach and its two
competitors were estimated and evaluated for $100$ random samples of size $N
= 200$ drawn from model (\ref{simmod}).  Cross-validation was used to
determine the bandwidths for the kernel-based estimators \citep[function
\code{npcdist()} in package \pkg{np}; for details,
see][]{Hayfield_Racine_2008}.  We fitted the GAMLSS models by boosting; the
number of boosting iterations was determined via sample splitting
\citep{Mayr_Fenske_Hofner_2012}.  To investigate the stability of the three
procedures under non-informative explanatory variables, we added to the data $p = 1,
\dots, 5$ uniformly distributed variables without association to the
response and included them as potential explanatory variables in
the three models.  The case $p = 0$ corresponds to model~(\ref{simmod}).

\begin{sidewaysfigure}
\begin{center}
\includegraphics{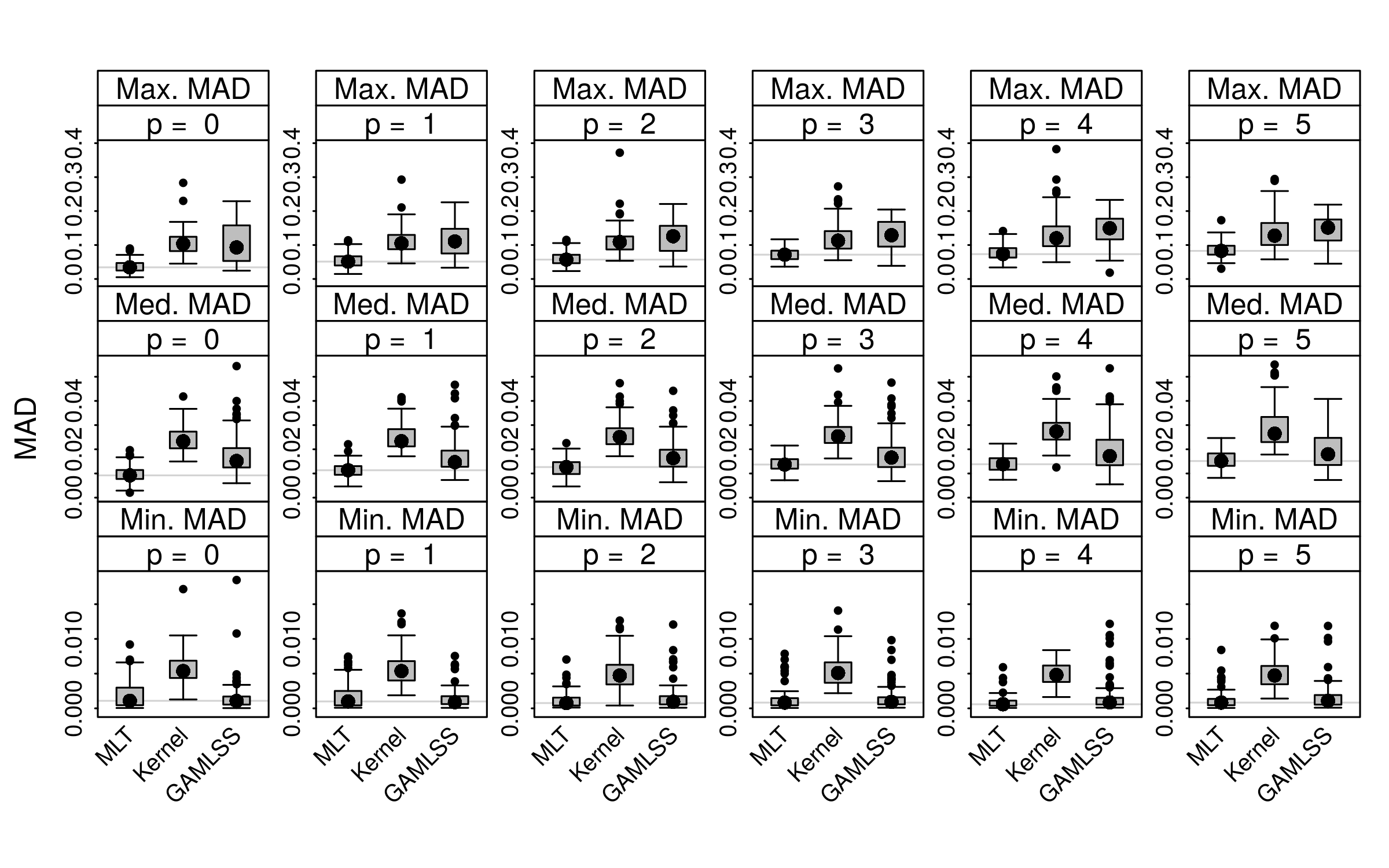}
\caption{Empirical Evaluation. Minimum, median and maximum of the mean absolute deviation (MAD)
         between true and estimated probabilities for most likely transformation
         models (MLT), non-parametric kernel distribution function estimation (Kernel)
         and generalised additive models for location, scale and shape (GAMLSS) for
         $100$ random samples.
         Values on the ordinate can be interpreted as absolute differences of probabilities.
         The grey horizontal lines correspond to the median of MLT.
         \label{sim-plot}}
\end{center}
\end{sidewaysfigure}

Figure~\ref{sim-plot} shows the empirical distributions of the minimum,
median and maximum MAD for the three competitors.  Except for the minimum
MAD in the absence of any irrelevant explanatory variables ($p = 0$), the
conditional distributions fitted by the transformation models were closer to the true conditional
distribution function by means of the MAD.  This result was obtained because
the transformation model only had to estimate a simple transformation
function, whereas the other two procedures had a difficult time approximating this
simple transformation model on another scale.
However, the comparison illustrates the potential improvement one can achieve  
when fitting simple models for the transformation function instead of more   
complex models for the mean (GAMLSS) or distribution function   
(Kernel). The kernel
estimator led to the largest median MAD values but seemed more robust than
GAMLSS with respect to the maximum MAD.  These results were remarkably robust
in the presence of up to five non-informative explanatory variables,
although of course the MAD increased with the number of 
non-informative variables $p$.

}
\fi

\section{Discussion}

The contribution of a likelihood approach for the general class of
conditional transformation models is interesting both from a theoretical and a
practical perspective.  With the range of simple to very complex
transformation functions introduced in Section~\ref{sec:appl} and illustrated
in Section~\ref{sec:empeval}, it becomes possible to understand classical 
parametric, semi-parametric and non-parametric models as special cases of
the same model class. Thus, analytic comparisons between
models of different complexity become possible.  The
transformation family $\Prob_{\rY, \Theta}$, the corresponding likelihood
function and the most likely transformation estimator are easy to
understand.  This makes the approach appealing also from a teaching
perspective.  Connections between standard parametric models (for example,
the normal linear model) and potentially complex models for survival or
ordinal data can be outlined in very simple notation, placing emphasis on the
modelling of (conditional) distributions instead of just modelling
(conditional) means.  Computationally, the log-likelihood $\log \circ \lik$ is linear in
the number of observations $N$ and, for contributions of ``exact
continuous'' responses, 
only requires the
evaluation of the derivative $\h^\prime$ of the transformation function $\h$
instead of integrals thereof.  
\if0\compact
{
Standard optimisers for linearly constrained
problems can be applied, and AIC- or BIC-based model selection is
possible.  Transformation models directly suggest a way to specify
conceptual $(\pZ, \basisyx, \parm)$ and fitted $(\pZ, \basisyx,
\hat{\parm})$ models in computer languages because ``only'' the
conditional transformation function $\h = \basisyx^\top \parm$ has to be specified in
addition to a ``simple'' distribution function $\pZ$
\citep{pkg:mlt,vign:mlt.docreg}.

The results presented in Section~\ref{sec:empeval} are based on only roughly
$1000$ lines of \textsf{R} code.  We tested this reference implementation of
most likely transformations against all linear
transformation models available in \textsf{R} to date and obtained equivalent
regression coefficients $\hat{\shiftparm}$ and their corresponding
covariance  matrices.
Thus, the framework helps to reduce the code base
and thereby helps to considerably reduce the number of possible errors in implementations of 
linear transformation models, while at the same time it allows novel and more
complex models to be fitted \citep{pkg:mlt,vign:mlt.docreg}.
}
\fi

Based on the general understanding of transformation models outlined in this
paper, it will be
interesting to study these models outside the strict likelihood world.  A
mixed transformation model for cluster data 
\if0\compact
\citep{Caietal_2002, HuberCarolVonta_2004, Zengetal_2005, ChoiHuang_2012} 
\else
\citep{Caietal_2002, Zengetal_2005, ChoiHuang_2012} 
\fi
is often based on the transformation function
$\h(\ry \mid \rx, i) = \hY(\ry) + \delta_i + \hx(\rx)$ with random intercept
(or ``frailty'' term) $\delta_i$ for the $i$th observational unit.  Conceptually, a more complex
deviation from the global model could by formulated as $\h(\ry \mid \rx, i) =
\hY(\ry) + \hY(\ry, i) + \hx(\rx)$, \ie each observational unit is assigned
its own ``baseline'' transformation $\hY(\ry) + \hY(\ry, i)$, where the
second term is an integral zero deviation from $\hY$. For longitudinal data
with possibly time-varying explanatory variables, the model
$\h(\ry \mid \rx(t), t) = \h_\rY(\ry, t) + \rx(t)\shiftparm(t)$
\citep{Ding_Tian_Yu_2012, Wu_Tian_2013} can also be
understood as a mixed version of a conditional transformation model.
The penalised log-likelihood $\log(\lik(\h \mid \ry)) - \text{pen}(\shiftparm)$ for the
linear transformation model $\h(\ry \mid \rx) = \hY(\ry) - \tilde{\rx}^\top
\shiftparm$ leads to Ridge- or Lasso-type regularised models, depending on
the form of the penalty term.  Priors for all model parameters $\parm$ allow
a fully Bayesian treatment of transformation models.
\if0\compact
Instead of the relatively strict model assumption $\pY = \pZ \circ
\basisy^\top \parm$, one could of course allow arbitrary unknown
transformation functions $\h \in \hs$ in the spirit of non-parametric
regression and study the quality of the approximation $\basisy^\top
\hat{\parm}$ for $\h$ for the fixed and random design case.
\fi

\rev{It is possible to relax the assumption that $\pZ$ is known. The simultaneous
estimation of $\pZ$ in the model $\Prob(\rY \le \ry \mid \rX = \rx) =
\pZ(\hY(\ry) - \tilde{\rx}^\top \shiftparm)$ was studied by \cite{Horowitz_1996} and later
extended by \cite{lintonetal_2008} to non-linear functions $\hx$ with parametric baseline
transformation $\hY$ and kernel estimates for $\pZ$ and $\hx$.  For AFT models, \cite{ZhangDavidian_2008} applied
smooth approximations for the density $\dZ$ in an exact censored likelihood
estimation procedure. In a similar setup, \cite{Huang_2014} proposed a
method to jointly estimate the mean function and the error distribution in
a generalised linear model. The estimation of $\pZ$ is noteworthy in additive
models of the form $\hY + \hx$ because these models assume additivity of
the contributions of $\ry$ and $\rx$ on the scale of $\pZ^{-1}(\Prob(\rY \le
\ry \mid \rX = \rx))$.  If this model assumption seems questionable, 
one can either allow unknown $\pZ$ or move to a 
transformation model featuring a more complex
transformation function. 
From this
point of view, the distribution function $\pZ$ in flexible
transformation models is only a computational device mapping the unbounded
transformation function $\h$ into the unit interval strictly monotonically,
making the evaluation of the likelihood easy.  Then, $\pZ$ has no
further meaning or interpretation as error distribution. 
A compromise could be the family of distributions
$\pZ(\rz \mid \rho) = 1 - (1 + \rho \exp(\rz))^{-1/\rho}$ for $\rho > 0$ 
\citep[suggested by][]{McLain_Ghosh_2013} with
simultaneous maximum likelihood estimation of $\parm$ and $\rho$ for
additive transformation functions $\h = \hY + \hx$, as these models
are flexible and still relatively easy to interpret.
}

\if0\compact
{
Another comment concerns the connection between transformation models and
quantile regression. For some probability $p \in [0, 1]$, the conditional quantile function is
\begin{eqnarray*}
Q_\rY(p \mid \rX = \rx) := \pY^{-1}(p \mid \rX = \rx) = \inf\{\ry \in \samY \mid \pZ(\h(\ry \mid \rx)) \ge p\},
\end{eqnarray*}
and for absolutely continuous $\rY$, we then get $Q_\rY(p \mid \rX = \rx) =
\h^{-1}(\pZ^{-1}(p) \mid \rX = \rx)$.  A linear quantile regression model
\citep{Koenker_2005} assumes $Q_\rY(p \mid \rX = \rx) = \alpha(p) + \tilde{\rx}^\top
\shiftparm(p)$, which cannot be written as a linear transformation model. 
Thus, non-linear transformation functions are necessary in order to achieve
the same flexibility as a linear quantile regression model; however, the
full conditional distribution function can be estimated in one step, which avoids
computational problems such as quantile crossing.  In the most complex case
of a transformation model that is invariant with respect to the choice of $\pZ$, 
there is a correspondingly complex
quantile regression model such that both models are equivalent, as was also
noted by \cite{Chernozhukov_2013}.  In a certain sense, we can understand
transformation models as ``inverse quantile regression''.  A more detailed
analysis of the connection between transformation models (called
distribution regression there) and quantile regression can be found in
\cite{Leorato_Peracchi_2015}.
}
\fi

\begin{figure}[h!]
\begin{center}
\includegraphics{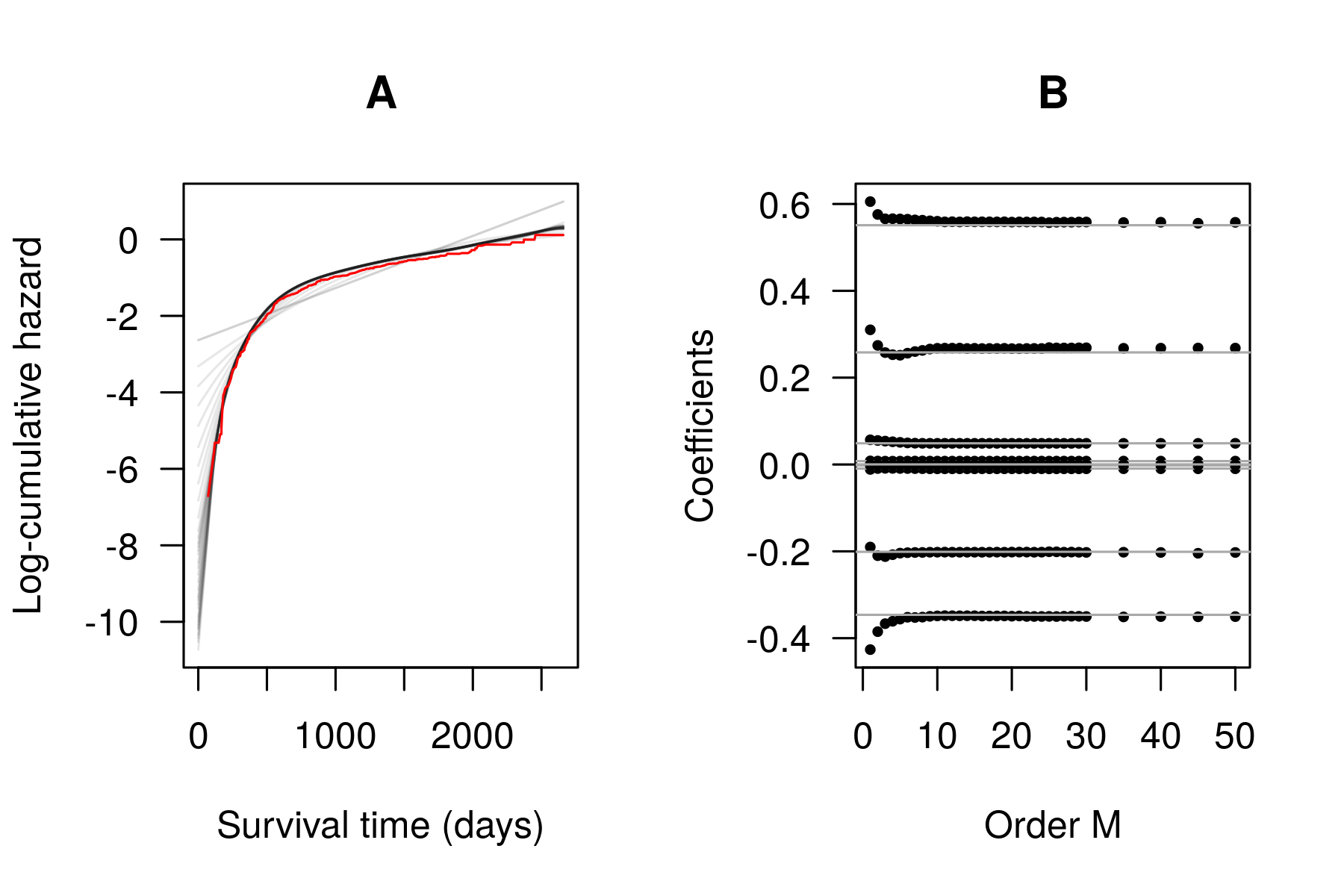}
\caption{GBSG-2. Comparison of exact and partial likelihood for order $M = 1, \dots, 30, 35,40, 45, 50$
         of the Bernstein polynomial approximating the log-cumulative baseline hazard function $\h_\rY$.
         The estimated log-cumulative baseline hazard functions for varying
         $M$ are shown in grey and the Nelson-Aalen-Breslow estimator
         is shown in red (A). The right panel (B) shows the trajectories of the
         regression coefficients $\hat{\shiftparm}$ obtained for varying $M$, which are represented as 
         dots. The horizontal lines represent the partial likelihood estimates.
         \label{fig:GBSG2-coxph_mlt}}
\end{center}
\end{figure}

In light of the empirical results discussed in this paper and the
theoretical work of \cite{McLain_Ghosh_2013} on a Cox model with
log-cumulative baseline hazard function parameterised in terms of a
Bernstein polynomial with increasing order $M$, one might ask where the
boundaries between parametric, semi-parametric and non-parametric statistics
lie.  The question how the order $M$ affects results practically has been
repeatedly raised; therefore, we will close our discussion by looking at a Cox
model with increasing $M$ for the GBSG-2 data.  All eight baseline variables were included in
the linear predictor, and we fitted the model with orders $M = 1, \dots, 30,
35,40, 45, 50$ of the Bernstein polynomial parameterising the log-cumulative
baseline hazard function.  In Figure~\ref{fig:GBSG2-coxph_mlt}A,
the log-cumulative baseline hazard functions start with a linear function
($M = 1$) and quickly approach a function that is essentially a smoothed
version of the Nelson-Aalen-Breslow estimator plotted in red.  In
Figure~\ref{fig:GBSG2-coxph_mlt}B, the
trajectories of the estimated regression coefficients become very similar to
the partial likelihood estimates as $M$ increased.  For $M \ge 10$, for instance, the
results of the ``semi-parametric'' and the ``fully parametric'' Cox models
are practically equivalent.  An extensive collection of such head-to-head
comparisons of most likely transformations with their classical counterparts
can be found in \cite{vign:mlt.docreg}.  Our work for this paper and
practical experience with its reference software implementation 
convinced us that rethinking classical models in terms of
fully parametric transformations is intellectually and practically a
fruitful exercise.

\section*{Acknowledgements}

Torsten Hothorn received financial support by Deutsche
Forschungsgemeinschaft (DFG) under grant number HO 3242/4-1. We thank Karen
A.~Brune for improving the language.

\bibliography{mlt,packages}

\begin{thebibliography}{77}
\newcommand{\enquote}[1]{``#1''}
\providecommand{\natexlab}[1]{#1}
\providecommand{\url}[1]{\texttt{#1}}
\providecommand{\urlprefix}{URL }
\expandafter\ifx\csname urlstyle\endcsname\relax
  \providecommand{\doi}[1]{doi:\discretionary{}{}{}#1}\else
  \providecommand{\doi}{doi:\discretionary{}{}{}\begingroup
  \urlstyle{rm}\Url}\fi
\providecommand{\eprint}[2][]{\url{#2}}

\bibitem[{Alzaatreh \emph{et~al.}(2013)Alzaatreh, Lee, and
  Famoye}]{Alzaatreh_Lee_Famoye_2013}
Alzaatreh A, Lee C, Famoye F (2013).
\newblock \enquote{A New Method for Generating Families of Continuous
  Distributions.}
\newblock \emph{Metron}, \textbf{71}(1), 63--79.
\newblock \doi{10.1007/s40300-013-0007-y}.

\bibitem[{Aranda-Ordaz(1983)}]{ArandaOrdaz_1983}
Aranda-Ordaz FJ (1983).
\newblock \enquote{An Extension of the proportional Hazards Model for Grouped
  Data.}
\newblock \emph{Biometrics}, \textbf{39}(1), 109--117.

\bibitem[{Azzalini and Bowman(1990)}]{Azzalini_Bowman_1990}
Azzalini A, Bowman AW (1990).
\newblock \enquote{A Look at Some Data on the Old Faithful Geyser.}
\newblock \emph{Journal of the Royal Statistical Society: Series C (Applied
  Statistics)}, \textbf{39}(3), 357--365.
\newblock \doi{10.2307/2347385}.

\bibitem[{Babu \emph{et~al.}(2002)Babu, Canty, and
  Chaubey}]{Babu_Canty_Chaubey_2002}
Babu GJ, Canty AJ, Chaubey YP (2002).
\newblock \enquote{Application of {Bernstein} Polynomials for Smooth Estimation
  of a Distribution and Density Function.}
\newblock \emph{Journal of Statistical Planning and Inference},
  \textbf{105}(2), 377--392.
\newblock \doi{10.1016/S0378-3758(01)00265-8}.

\bibitem[{Bender \emph{et~al.}(2005)Bender, Augustin, and
  Blettner}]{Bender_Augustin_Blettner_2005}
Bender R, Augustin T, Blettner M (2005).
\newblock \enquote{Generating Survival Times to Simulate {Cox} Proportional
  Hazards Models.}
\newblock \emph{Statistics in Medicine}, \textbf{24}(11), 1713--1723.
\newblock \doi{10.1002/sim.2059}.

\bibitem[{Bickel \emph{et~al.}(1993)Bickel, Klaassen, Ritov, and
  Wellner}]{Bickel_Klaassen_Ritov_1993}
Bickel PJ, Klaassen CAJ, Ritov Y, Wellner JA (1993).
\newblock \emph{Efficient and Adaptive Estimation for Semiparametric Models}.
\newblock The Johns Hopkins University Press, Baltimore, U.S.A. and London,
  U.K.

\bibitem[{Box and Cox(1964)}]{BoxCox_1964}
Box GEP, Cox DR (1964).
\newblock \enquote{An Analysis of Transformations.}
\newblock \emph{Journal of the Royal Statistical Society: Series B (Statistical
  Methodology)}, \textbf{26}(2), 211--252.

\bibitem[{Cai and Betensky(2003)}]{CaiBetensky_2003}
Cai T, Betensky RA (2003).
\newblock \enquote{Hazard Regression for Interval-Censored Data with Penalized
  Spline.}
\newblock \emph{Biometrics}, \textbf{59}(3), 570--579.
\newblock \doi{10.1111/1541-0420.00067}.

\bibitem[{Cai \emph{et~al.}(2002)Cai, Cheng, and Wei}]{Caietal_2002}
Cai T, Cheng SC, Wei LJ (2002).
\newblock \enquote{Semiparametric Mixed-Effects Models for Clustered Failure
  Time Data.}
\newblock \emph{Journal of the American Statistical Association},
  \textbf{97}(458), 514--522.
\newblock \doi{10.1198/016214502760047041}.

\bibitem[{Chang \emph{et~al.}(2005)Chang, Hsiung, Wu, and
  Yany}]{Chang_Hsiung_Wu_2005}
Chang I, Hsiung CA, Wu YJ, Yany CC (2005).
\newblock \enquote{Bayesian Survival Analysis Using {Bernstein} Polynomials.}
\newblock \emph{Scandinavian Journal of Statistics}, \textbf{32}(3), 447--466.
\newblock \doi{10.1111/j.1467-9469.2005.00451.x}.

\bibitem[{Chen \emph{et~al.}(2002)Chen, Jin, and Ying}]{Chenetal_2002}
Chen K, Jin Z, Ying Z (2002).
\newblock \enquote{Semiparametric Analysis of Transformation Models with
  Censored Data.}
\newblock \emph{Biometrika}, \textbf{89}(3), 659--668.
\newblock \doi{10.1093/biomet/89.3.659}.

\bibitem[{Cheng \emph{et~al.}(1995)Cheng, Wei, and Ying}]{Chengetal_1995}
Cheng SC, Wei LJ, Ying Z (1995).
\newblock \enquote{Analysis of Transformation Models with Censored Data.}
\newblock \emph{Biometrika}, \textbf{82}(4), 835--845.
\newblock \doi{10.1093/biomet/82.4.835}.

\bibitem[{Cheng \emph{et~al.}(1997)Cheng, Wei, and Ying}]{Chengetal_1997}
Cheng SC, Wei LJ, Ying Z (1997).
\newblock \enquote{Predicting Survival Probabilities with Semiparametric
  Transformation Models.}
\newblock \emph{Journal of the American Statistical Association},
  \textbf{92}(437), 227--235.
\newblock \doi{10.1080/01621459.1997.10473620}.

\bibitem[{Chernozhukov \emph{et~al.}(2013)Chernozhukov, Fern{\'a}ndez-Val, and
  Melly}]{Chernozhukov_2013}
Chernozhukov V, Fern{\'a}ndez-Val I, Melly B (2013).
\newblock \enquote{Inference on Counterfactual Distributions.}
\newblock \emph{Econometrica}, \textbf{81}(6), 2205--2268.
\newblock \doi{10.3982/ECTA10582}.

\bibitem[{Choi and Huang(2012)}]{ChoiHuang_2012}
Choi S, Huang X (2012).
\newblock \enquote{A General Class of Semiparametric Transformation Frailty
  Models for Nonproportional Hazards Survival Data.}
\newblock \emph{Biometrics}, \textbf{68}(4), 1126--1135.
\newblock \doi{10.1111/j.1541-0420.2012.01784.x}.

\bibitem[{Cox(1975)}]{Cox_1975}
Cox DR (1975).
\newblock \enquote{Partial Likelihood.}
\newblock \emph{Biometrika}, \textbf{62}(2), 269--276.

\bibitem[{Crowther and Lambert(2014)}]{CrowtherLambert_2014}
Crowther MJ, Lambert PC (2014).
\newblock \enquote{A General Framework for Parametric Survival Analysis.}
\newblock \emph{Statistics in Medicine}, \textbf{33}(30), 5280--5297.
\newblock \doi{10.1002/sim.6300}.

\bibitem[{Curtis and Ghosh(2011)}]{Curtis_Ghosh_2011}
Curtis SM, Ghosh SK (2011).
\newblock \enquote{A Variable Selection Approach to Monotonic Regression with
  {Bernstein} Polynomials.}
\newblock \emph{Journal of Applied Statistics}, \textbf{38}(5), 961--976.
\newblock \doi{10.1080/02664761003692423}.

\bibitem[{Ding \emph{et~al.}(2012)Ding, Tian, Yu, and Guo}]{Ding_Tian_Yu_2012}
Ding AA, Tian S, Yu Y, Guo H (2012).
\newblock \enquote{A Class of Discrete Transformation Survival Models With
  Application to Default Probability Prediction.}
\newblock \emph{Journal of the American Statistical Association},
  \textbf{107}(499), 990--1003.
\newblock \doi{10.1080/01621459.2012.682806}.

\bibitem[{Doksum and Gasko(1990)}]{Doksum_Gasko_1990}
Doksum KA, Gasko M (1990).
\newblock \enquote{On a Correspondence Between Models in Binary Regression
  Analysis and in Survival Analysis.}
\newblock \emph{International Statistical Review}, \textbf{58}(3), 243--252.

\bibitem[{Farouki(2012)}]{Farouki_2012}
Farouki RT (2012).
\newblock \enquote{The {Bernstein} Polynomial Basis: A Centennial
  Retrospective.}
\newblock \emph{Computer Aided Geometric Design}, \textbf{29}(6), 379--419.
\newblock \doi{10.1016/j.cagd.2012.03.001}.

\bibitem[{Fisher(1934)}]{Fisher_1934}
Fisher RA (1934).
\newblock \enquote{Two New Properties of Mathematical Likelihood.}
\newblock \emph{Proceedings of the Royal Society of London. Series A,
  Containing Papers of a Mathematical and Physical Character}, \textbf{53},
  285--306.

\bibitem[{Foresi and Peracchi(1995)}]{Foresi_Peracchi_1995}
Foresi S, Peracchi F (1995).
\newblock \enquote{The Conditional Distribution of Excess Returns: {An}
  Empirical Analysis.}
\newblock \emph{Journal of the American Statistical Association},
  \textbf{90}(430), 451--466.
\newblock \doi{10.1080/01621459.1995.10476537}.

\bibitem[{Fraser(1968)}]{Fraser_1968}
Fraser DAS (1968).
\newblock \emph{The Structure of Inference}.
\newblock John Wiley \& Sons, New York, U.S.A.

\bibitem[{Fredriks \emph{et~al.}(2000)Fredriks, van Buuren, Burgmeijer,
  Meulmeester, Beuker, Brugman, Roede, Verloove-Vanhorick, and
  Wit}]{Fredriks_Buuren_Burgmeijer_2000}
Fredriks AM, van Buuren S, Burgmeijer RJF, Meulmeester JF, Beuker RJ, Brugman
  E, Roede MJ, Verloove-Vanhorick SP, Wit J (2000).
\newblock \enquote{Continuing Positive Secular Growth Change in {The}
  {Netherlands} 1955--1997.}
\newblock \emph{Pediatric Research}, \textbf{47}(3), 316--323.

\bibitem[{Gneiting and Katzfuss(2014)}]{Gneiting_Katzfuss_2014}
Gneiting T, Katzfuss M (2014).
\newblock \enquote{Probabilistic Forecasting.}
\newblock \emph{Annual Review of Statistics and Its Application},
  \textbf{1}(1), 125--151.
\newblock \doi{10.1146/annurev-statistics-062713-085831}.

\bibitem[{Hayfield and Racine(2008)}]{Hayfield_Racine_2008}
Hayfield T, Racine JS (2008).
\newblock \enquote{Nonparametric Econometrics: The np Package.}
\newblock \emph{Journal of Statistical Software}, \textbf{27}(5), 1--32.
\newblock \urlprefix\url{http://www.jstatsoft.org/v27/i05}.

\bibitem[{Hofner \emph{et~al.}(2016)Hofner, Mayr, Fenske, and
  Schmid}]{pkg:gamboostLSS}
Hofner B, Mayr A, Fenske N, Schmid M (2016).
\newblock \emph{gamboostLSS: Boosting Methods for 'GAMLSS'}.
\newblock R package version 1.2-2,
  \urlprefix\url{https://CRAN.R-project.org/package=gamboostLSS}.

\bibitem[{Horowitz(1996)}]{Horowitz_1996}
Horowitz JL (1996).
\newblock \enquote{Semiparametric Estimation of a Regression Model with an
  Unknown Transformation of the Dependent Variable.}
\newblock \emph{Econometrica}, \textbf{64}(1), 103--137.

\bibitem[{Hothorn(2017{\natexlab{a}})}]{pkg:mlt}
Hothorn T (2017{\natexlab{a}}).
\newblock \emph{mlt: Most Likely Transformations}.
\newblock R package version 0.1-3,
  \urlprefix\url{https://CRAN.R-project.org/package=mlt}.

\bibitem[{Hothorn(2017{\natexlab{b}})}]{vign:mlt.docreg}
Hothorn T (2017{\natexlab{b}}).
\newblock \emph{Most Likely Transformations: The mlt Package}.
\newblock R package vignette version 0.1-5,
  \urlprefix\url{https://CRAN.R-project.org/package=mlt.docreg}.

\bibitem[{Hothorn \emph{et~al.}(2008)Hothorn, Bretz, and
  Westfall}]{Hothorn_Bretz_Westfall_2008}
Hothorn T, Bretz F, Westfall P (2008).
\newblock \enquote{Simultaneous Inference in General Parametric Models.}
\newblock \emph{Biometrical Journal}, \textbf{50}(3), 346--363.
\newblock \doi{10.1002/bimj.200810425}.

\bibitem[{Hothorn \emph{et~al.}(2014)Hothorn, Kneib, and
  B{\"u}hlmann}]{Hothorn_Kneib_Buehlmann_2014}
Hothorn T, Kneib T, B{\"u}hlmann P (2014).
\newblock \enquote{Conditional Transformation Models.}
\newblock \emph{Journal of the Royal Statistical Society: Series B (Statistical
  Methodology)}, \textbf{76}(1), 3--27.
\newblock \doi{10.1111/rssb.12017}.

\bibitem[{Huang(2014)}]{Huang_2014}
Huang A (2014).
\newblock \enquote{Joint Estimation of the Mean and Error Distribution in
  Generalized Linear Models.}
\newblock \emph{Journal of the American Statistical Association},
  \textbf{109}(505), 186--196.
\newblock \doi{10.1080/01621459.2013.824892}.

\bibitem[{Huber-Carol and Vonta(2004)}]{HuberCarolVonta_2004}
Huber-Carol C, Vonta I (2004).
\newblock \enquote{Frailty Models for Arbitrarily Censored and Truncated Data.}
\newblock \emph{Lifetime Data Analysis}, \textbf{10}(4), 369--388.
\newblock \doi{10.1007/s10985-004-4773-y}.

\bibitem[{Hyndman and Yao(2002)}]{Hyndman_Yao_2002}
Hyndman RJ, Yao Q (2002).
\newblock \enquote{Nonparametric Estimation and Symmetry Tests for Conditional
  Density Functions.}
\newblock \emph{Journal of Nonparametric Statistics}, \textbf{14}(3), 259--278.
\newblock \doi{10.1080/10485250212374}.

\bibitem[{Joly \emph{et~al.}(1998)Joly, Commenges, and
  Letenneur}]{Jolyetal_1998}
Joly P, Commenges D, Letenneur L (1998).
\newblock \enquote{A Penalized Likelihood Approach for Arbitrarily Censored and
  Truncated Data: Application to Age-Specific Incidence of Dementia.}
\newblock \emph{Biometrics}, \textbf{54}(1), 185--194.

\bibitem[{Jones and Pewsey(2009)}]{Jones_Pewsey_2009}
Jones MC, Pewsey A (2009).
\newblock \enquote{Sinh-arcsinh Distributions.}
\newblock \emph{Biometrika}, \textbf{96}(4), 761--780.
\newblock \doi{10.1093/biomet/asp053}.

\bibitem[{Klein and Moeschberger(2003)}]{Klein_Moeschberger_2003}
Klein JP, Moeschberger MK (2003).
\newblock \emph{Survival Analysis. Techniques for Censored and Truncated Data.}
\newblock 2nd edition. Springer, New York, U.S.A.

\bibitem[{Koenker(2005)}]{Koenker_2005}
Koenker R (2005).
\newblock \emph{Quantile Regression}.
\newblock Economic Society Monographs. Cambridge University Press, New York,
  U.S.A.

\bibitem[{Kooperberg and Clarkson(1997)}]{KooperbergClarkson_1997}
Kooperberg C, Clarkson DB (1997).
\newblock \enquote{Hazard Regression with Interval-Censored Data.}
\newblock \emph{Biometrics}, \textbf{53}(4), 1485--1494.

\bibitem[{Kooperberg \emph{et~al.}(1995)Kooperberg, Stone, and
  Truong}]{Kooperbergetal_1995}
Kooperberg C, Stone CJ, Truong YK (1995).
\newblock \enquote{Hazard Regression.}
\newblock \emph{Journal of the American Statistical Association},
  \textbf{90}(429), 78--94.

\bibitem[{Lehmann(1983)}]{Lehmann_1983}
Lehmann EL (1983).
\newblock \emph{Theory of Point Estimation}.
\newblock John Wiley \& Sons, New York, U.S.A.

\bibitem[{Leorato and Peracchi(2015)}]{Leorato_Peracchi_2015}
Leorato S, Peracchi F (2015).
\newblock \enquote{Comparing Distribution and Quantile Regression.}
\newblock \emph{Technical Report 1511}, Einaudi Institute for Economics and
  Finance, Rome, Italy.
\newblock \urlprefix\url{https://ideas.repec.org/p/eie/wpaper/1511.html}.

\bibitem[{Lindsey(1996)}]{Lindsey_1996}
Lindsey JK (1996).
\newblock \emph{Parametric Statistical Inference}.
\newblock Clarendon Press, Oxford, UK.

\bibitem[{Lindsey(1999)}]{Lindsey_1999}
Lindsey JK (1999).
\newblock \enquote{Some Statistical Heresies.}
\newblock \emph{Journal of the Royal Statistical Society: Series D (The
  Statistician)}, \textbf{48}(1), 1--40.

\bibitem[{Linton \emph{et~al.}(2008)Linton, Sperlich, and {van
  Keilegom}}]{lintonetal_2008}
Linton O, Sperlich S, {van Keilegom} I (2008).
\newblock \enquote{Estimation of a Semiparametric Transformation Model.}
\newblock \emph{The Annals of Statistics}, \textbf{36}(2), 686--718.
\newblock \doi{10.1214/009053607000000848}.

\bibitem[{Ma \emph{et~al.}(2014)Ma, Heritier, and L\^{o}}]{Maetal_2014}
Ma J, Heritier S, L\^{o} SN (2014).
\newblock \enquote{On the Maximum Penalized Likelihood Approach for
  Proportional Hazard Models with Right Censored Survival Data.}
\newblock \emph{Computational Statistics \& Data Analysis}, \textbf{74},
  142--156.
\newblock \doi{10.1016/j.csda.2014.01.005}.

\bibitem[{Mallick and Walker(2003)}]{MallickWalker_2003}
Mallick BK, Walker S (2003).
\newblock \enquote{A {B}ayesian Semiparametric Transformation Model
  Incorporating Frailties.}
\newblock \emph{Journal of Statistical Planning and Inference},
  \textbf{112}(1--2), 159--174.
\newblock \doi{10.1016/S0378-3758(02)00330-0}.

\bibitem[{Mayr \emph{et~al.}(2012)Mayr, Fenske, Hofner, Kneib, and
  Schmid}]{Mayr_Fenske_Hofner_2012}
Mayr A, Fenske N, Hofner B, Kneib T, Schmid M (2012).
\newblock \enquote{{GAMLSS} for High-dimensional Data--{A} Flexible Approach
  Based on Boosting.}
\newblock \emph{Journal of the Royal Statistical Society: Series C (Applied
  Statistics)}, \textbf{61}(3), 403--427.
\newblock \doi{10.1111/j.1467-9876.2011.01033.x}.

\bibitem[{McLain and Ghosh(2013)}]{McLain_Ghosh_2013}
McLain AC, Ghosh SK (2013).
\newblock \enquote{Efficient Sieve Maximum Likelihood Estimation of
  Time-Transformation Models.}
\newblock \emph{Journal of Statistical Theory and Practice}, \textbf{7}(2),
  285--303.

\bibitem[{M\"ost and Hothorn(2015)}]{Moest_Hothorn_2015}
M\"ost L, Hothorn T (2015).
\newblock \enquote{Conditional Transformation Models for Survivor Function
  Estimation.}
\newblock \emph{International Journal of Biostatistics}, \textbf{11}(1),
  23--50.
\newblock \doi{10.1515/ijb-2014-0006}.

\bibitem[{M\"ost \emph{et~al.}(2016)M\"ost, Schmid, Faschingbauer, and
  Hothorn}]{Moest_Schmid_Faschingbauer_2014}
M\"ost L, Schmid M, Faschingbauer F, Hothorn T (2016).
\newblock \enquote{Predicting Birth Weight with Conditionally Linear
  Transformation Models.}
\newblock \emph{Statistical Methods in Medical Research}, \textbf{25}(6),
  2781--2810.
\newblock \doi{10.1177/0962280214532745}.

\bibitem[{M{\"u}ller and Hothorn(2004)}]{Mueller_Hothorn_2004}
M{\"u}ller J, Hothorn T (2004).
\newblock \enquote{Maximally Selected Two-Sample Statistics as a new Tool for
  the Identification and Assessment of Habitat Factors with an Application to
  Breeding Bird Communities in Oak Forests.}
\newblock \emph{European Journal of Forest Research}, \textbf{123}, 218--228.
\newblock \doi{10.1007/s10342-004-0035-5}.

\bibitem[{Pitman(1939)}]{Pitman_1939}
Pitman EJG (1939).
\newblock \enquote{The Estimation of the Location and Scale Parameters of a
  Continuous Population of Any Given Form.}
\newblock \emph{Biometrika}, \textbf{30}(3/4), 391--421.

\bibitem[{{R Core Team}(2017)}]{R}
{R Core Team} (2017).
\newblock \emph{R: A Language and Environment for Statistical Computing}.
\newblock R Foundation for Statistical Computing, Vienna, Austria.
\newblock \urlprefix\url{http://www.R-project.org/}.

\bibitem[{Racine and Hayfield(2014)}]{pkg:np}
Racine JS, Hayfield T (2014).
\newblock \emph{np: Nonparametric kernel smoothing methods for mixed data
  types}.
\newblock R package version 0.60-2,
  \urlprefix\url{https://CRAN.R-project.org/package=np}.

\bibitem[{Rigby and Stasinopoulos(2004)}]{Rigby_Stasinopoulos_2004}
Rigby RA, Stasinopoulos DM (2004).
\newblock \enquote{Smooth Centile Curves for Skew and Kurtotic Data Modelled
  Using the {Box-Cox} Power Exponential Distribution.}
\newblock \emph{Statistics in Medicine}, \textbf{23}(19), 3053--3076.
\newblock \doi{10.1002/sim.1861}.

\bibitem[{Rigby and Stasinopoulos(2005)}]{Rigby_Stasinopoulos_2005}
Rigby RA, Stasinopoulos DM (2005).
\newblock \enquote{Generalized Additive Models for Location, Scale and Shape.}
\newblock \emph{Journal of the Royal Statistical Society: Series C (Applied
  Statistics)}, \textbf{54}(3), 507--554.
\newblock \doi{10.1111/j.1467-9876.2005.00510.x}.

\bibitem[{Rothe and Wied(2013)}]{Rothe_Wied_2013}
Rothe C, Wied D (2013).
\newblock \enquote{Misspecification Testing in a Class of Conditional
  Distributional Models.}
\newblock \emph{Journal of the American Statistical Association},
  \textbf{108}(501), 314--324.
\newblock \doi{10.1080/01621459.2012.736903}.

\bibitem[{Royston and Parmar(2002)}]{Royston_Parmar_2002}
Royston P, Parmar MKB (2002).
\newblock \enquote{Flexible Parametric Proportional-hazards and
  Proportional-odds Models for Censored Survival Data, with Application to
  Prognostic Modelling and Estimation of Treatment Effects.}
\newblock \emph{Statistics in Medicine}, \textbf{21}(15), 2175--2197.
\newblock \doi{10.1002/sim.1203}.

\bibitem[{{Saban{\'e}s Bov{\'e}} and Held(2013)}]{Bove_Held_2013}
{Saban{\'e}s Bov{\'e}} D, Held L (2013).
\newblock \enquote{Comment on {Cai} and {Betensky} (2003), {On} the {Poisson}
  Approximation for Hazard Regression.}
\newblock \emph{Biometrics}, \textbf{69}(3), 795--795.
\newblock \doi{10.1111/biom.12088}.

\bibitem[{Schumacher \emph{et~al.}(1994)Schumacher, Basert, Bojar, H\"ubner,
  Olschewski, Sauerbrei, Schmoor, Beyerle, Neumann, and {Rauschecker, H. F.,
  for the German Breast Cancer Study Group}}]{gbsg2:1994}
Schumacher M, Basert G, Bojar H, H\"ubner K, Olschewski M, Sauerbrei W, Schmoor
  C, Beyerle C, Neumann RLA, {Rauschecker, H F, for the German Breast Cancer
  Study Group} (1994).
\newblock \enquote{Randomized $2\times2$ trial evaluating hormonal treatment
  and the duration of chemotherapy in node-positive breast cancer patients.}
\newblock \emph{Journal of Clinical Oncology}, \textbf{12}, 2086--2093.

\bibitem[{Stasinopoulos and Rigby(2007)}]{Stasinopoulos_Rigby_2007}
Stasinopoulos DM, Rigby RA (2007).
\newblock \enquote{Generalized Additive Models for Location Scale and Shape
  {(GAMLSS)} in {R}.}
\newblock \emph{Journal of Statistical Software}, \textbf{23}(7), 1--46.
\newblock \urlprefix\url{http://www.jstatsoft.org/v23/i07}.

\bibitem[{Sun \emph{et~al.}(2009)Sun, Sundaram, and
  Zhao}]{Sun_Sundaram_Zhao_2009}
Sun Y, Sundaram R, Zhao Y (2009).
\newblock \enquote{Empirical Likelihood Inference for the {Cox} Model with
  Time-dependent Coefficients via Local Partial Likelihood.}
\newblock \emph{Scandinavian Journal of Statistics}, \textbf{36}(3), 444--462.
\newblock \doi{10.1111/j.1467-9469.2008.00634.x}.

\bibitem[{Thas \emph{et~al.}(2012)Thas, Neve, Clement, and
  Ottoy}]{Thas_Neve_Clement_2012}
Thas O, Neve JD, Clement L, Ottoy JP (2012).
\newblock \enquote{Probabilistic Index Models.}
\newblock \emph{Journal of the Royal Statistical Society: Series B (Statistical
  Methodology)}, \textbf{74}(4), 623--671.
\newblock \doi{10.1111/j.1467-9868.2011.01020.x}.

\bibitem[{Tobin(1958)}]{Tobin_1958}
Tobin J (1958).
\newblock \enquote{Estimation of Relationships for Limited Dependent
  Variables.}
\newblock \emph{Econometrica}, \textbf{26}(1), 24--36.
\newblock \doi{10.2307/1907382}.

\bibitem[{Tutz(2012)}]{Tutz_2012}
Tutz G (2012).
\newblock \emph{Regression for Categorical Data}.
\newblock Cambridge University Press, New York, U.S.A.

\bibitem[{van~de Geer(2000)}]{vdGeer_2000}
van~de Geer S (2000).
\newblock \emph{Empirical Processes in M-Estimation}.
\newblock Cambridge University Press, Cambridge, UK.

\bibitem[{van~der Vaart(1998)}]{vdVaart_1998}
van~der Vaart AW (1998).
\newblock \emph{Asymptotic Statistics}.
\newblock Cambridge University Press, Cambridge, UK.

\bibitem[{Varadhan(2015)}]{pkg:alabama}
Varadhan R (2015).
\newblock \emph{alabama: Constrained Nonlinear Optimization}.
\newblock R package version 2015.3-1,
  \urlprefix\url{https://CRAN.R-project.org/package=alabama}.

\bibitem[{Wood(2017)}]{pkg:mgcv}
Wood S (2017).
\newblock \emph{mgcv: Mixed GAM Computation Vehicle with GCV/AIC/REML
  Smoothness Estimation}.
\newblock R package version 1.8-17,
  \urlprefix\url{https://CRAN.R-project.org/package=mgcv}.

\bibitem[{Wu and Tian(2013)}]{Wu_Tian_2013}
Wu CO, Tian X (2013).
\newblock \enquote{Nonparametric Estimation of Conditional Distributions and
  Rank-Tracking Probabilities With Time-Varying Transformation Models in
  Longitudinal Studies.}
\newblock \emph{Journal of the American Statistical Association},
  \textbf{108}(503), 971--982.

\bibitem[{Yang and Prentice(2015)}]{Yang_Prentice_2015}
Yang S, Prentice RL (2015).
\newblock \enquote{Assessing Potentially Time-dependent Treatment Effect From
  Clinical Trials and Observational Studies for Survival Data, With
  Applications to the {Women's} {Health} {Initiative} Combined Hormone Therapy
  Trial.}
\newblock \emph{Statistics in Medicine}, \textbf{34}(11), 1801--1817.
\newblock \doi{10.1002/sim.6453}.

\bibitem[{Zeng and Lin(2007)}]{ZengLin_2007}
Zeng D, Lin DY (2007).
\newblock \enquote{Maximum Likelihood Estimation in Semiparametric Regression
  Models with Censored Data.}
\newblock \emph{Journal of the Royal Statistical Society: Series B (Statistical
  Methodology)}, \textbf{69}(4), 507--564.
\newblock \doi{10.1111/j.1369-7412.2007.00606.x}.

\bibitem[{Zeng \emph{et~al.}(2005)Zeng, Lin, and Yin}]{Zengetal_2005}
Zeng D, Lin DY, Yin G (2005).
\newblock \enquote{Maximum Likelihood Estimation for the Proportional Odds
  Model with Random Effects.}
\newblock \emph{Journal of the American Statistical Association},
  \textbf{100}(470), 470--483.
\newblock \doi{10.1198/016214504000001420}.

\bibitem[{Zhang and Davidian(2008)}]{ZhangDavidian_2008}
Zhang M, Davidian M (2008).
\newblock \enquote{Smooth Semiparametric Regression Analysis for Arbitrarily
  Censored Time-to-Event Data.}
\newblock \emph{Biometrics}, \textbf{64}(2), 567--576.
\newblock \doi{10.1111/j.1541-0420.2007.00928.x}.

\end{thebibliography}


\begin{appendix}
\section*{Appendix}

\subsection*{Computational Details}

A reference implementation of most likely transformation models is available
in the 
\if1\blind
\pkg{mlt} package \citep{pkg:mlt}. 
\else
\pkg{NONAME} package (ANONYMOUS).
\fi
All data analyses can be reproduced in the dynamic document 
\if1\blind
\cite{vign:mlt.docreg}. 
\else
ANONYMOUS.
\fi
Augmented Lagrangian Minimization implemented
in the \code{auglag()} function of package \pkg{alabama} \citep{pkg:alabama}
was used for optimising the log-likelihood.
Package \pkg{gamboostLSS} \citep[version~1.2-2,][]{pkg:gamboostLSS} was used to fit GAMLSS models and kernel
density and distribution estimation was performed using package
\pkg{np} \citep[version~0.60-2,][]{pkg:np}.
\if0\compact
{
The additive Poisson model was fitted by \pkg{mgcv} \citep[version~1.8-17,][]{pkg:mgcv}.
}
\fi
All computations were performed using \textsf{R} version 3.4.0 
\citep{R}.

\subsection*{Review History: Version 1 by Journal 1 (August 2015--January 2016)}

\noindent Review for version 1 (\url{http://arxiv.org/abs/1508.06749v1}).
Comments by the authors are printed in \textit{italics}.

\subsubsection*{Associate Editor}

The paper proposes a maximum likelihood approach to a general class of
transformation models.  The paper is nicely written.  The proposed framework
can incorporate different types of response variables: continuous, count,
censored, truncated.  The likelihood framework is attractive as it allows
one to make use of many tools developed in the classical maximum likelihood
setting.

The following are my main comments on the paper. 

\begin{enumerate}
\item   The abstract says ``We propose and study properties of maximum
likelihood estimators in the class of conditional transformation models''. 
However, the paper focuses on the one-sample case, that is, transforming a
univariate random variable (without covariates).  Hence, one does not see
much of the ``cascade of increasingly complex transformation models'' claimed
in the abstract.  For this paper to be of more interest to JRSSB readers,
the authors should instead focus on the more interesting and practically
useful conditional transformations (i.e., the regression case), which were
only sketched briefly in Section 4.2?

\textit{Regression models are discussed in much detail in Section 4.2 (pages
12-16), in three out of four examples presented in Section~5.1, and in the
simulation model described in Section~5.2. The theory was presented in the
unconditional case for notational and educational convenience only and
carries over to the conditional case as stated in the 2nd paragraph of
Section 4.2.}

\item	Referee 1 is concerned that ``once the transformations are
parameterized, the asymptotic results in Section 3 are obvious from standard
textbooks''.  This seems to be a valid concern.  

\textit{We see this as an advantage, see also the corresponding comment to
Referee 1.}

I have my own concerns on
the theoretical results: the main theorems (Theorem 1-3) are derived under
the assumption $Y_1, \ldots, Y_n$ are $i.i.d.$ from $P_{Y, \theta_0}$, then
it's not much different from a parametric model.  To me, it seems that a
more reasonable theory should consider: (1) approximate the unknown
transformation $h$ using $a^\top\theta$; (2) then apply MLE at the second step. 
The current theory is only for the second step but ignores the spline
approximation error in the first step.  For transformation models, the
theory would be more relevant to consider approximating $h$ using the
possibly misspecified $a^\top\theta$, for which one usually needs to let the
number of basis functions goes to infinity at certain speed.  

\textit{A corresponding consistency result for boosted conditional
transformation models treating $h$ as unknown can be found in Hothorn et al.~2014. The more useful
asymptotic normality would be much harder to derive in this situation and we therefore decided
to stay within the bounds of the modern likelihood world. Note that the
comment only applies to continuous models as the transformation function is
discrete in discrete models.}

There is
another related question: $h$ is required to be strictly monotone but the
basis function approximation does not automatically guarantee it.  This
needs to be carefully discussed.

\textit{Both parameterisations for discrete and continuous models (Bernstein
polynomial) in Section~4.1 allow monotone transformation functions to be estimated
under linear constraints.
We used a spectral projected gradient method (see Appendix) for optimisation
of the likelihood under such linear constraints. The estimated
transformation functions are guaranteed to be monotone in $y$; the nature
and number of constraints are better explained in the revised version.}

\item	There is the important question that given a real data example,
which transformation model should be chosen?  When will the new method be
more preferable than the alternative methods in the literature, such as
Cheng et al.~(1995), Chen et al.~(2002)?  

\textit{Classical linear transformation models treat the ``baseline''
transformation $h_Y$ as a nuisance parameter. Fitting the same models under
the full likelihood as suggested here leads to practically the same
estimates of the regression coefficients (see Section~5) and a smooth
interpolation of the post-hoc NPML estimate of $h_Y$. It is then relatively
simple to construct asymptotic confidence bands for the distribution of the
survivor function (see also Section~5). Most importantly, it is hard or even
impossible to fit more complex models, for example one with time-varying effects or a
conditional transformation model, using the classical approaches but it is
straightforward (as we have shown in Section~5) using the most likely
transformation approach.}

The ``Extension of the classical
transformation models'' appear to have some overlap with the authors' recent
work (Hothorn et al.~(2014), Conditional Transformation Models, JRSSB)? 

\textit{Yes, the class of ``conditional transformation models'' was introduced
in Hothorn et al.~2014, along with a boosting algorithm for the minimisation
of proper scoring rules.  The 2014 paper is restricted to exact continuous
responses.  The novel contribution of the present manuscript is a full
likelihood framework for arbitrary responses and corresponding likelihood
inference procedures.}

When will the method proposed in this paper be preferable to those in
Hothorn et al.~(2014)?

\textit{Whenever one is interested in likelihood inference for possibly
censored, truncated or discrete responses under weak 
parametric assumptions in a fairly large model class.}

\end{enumerate}

\subsubsection*{Referee 1}

The paper aims to study a unified framework based on transformations to
address semiparametric/nonparametric estimation for discrete, continuous or
censored data.  First, any distribution functions are equivalent to some
known distributions of transformed random variables.  Then the paper
proposes a class of parametric models for transformations so estimates the
parameters via maximum likelihood estimation.  The theoretical results of
the parameter estimators are provided.  The paper describes a list of
examples to illustrate transformations and corresponding results in Section~4.  
Some numerical evidence is given in Section 5.  My serious concern is
about the new material presented here.  The key motivation for using
transformation (the results in Section 2) is available in standard
textbooks.  

\textit{All classical texts study and parameterise transformations $Y =
g(Z)$, for example under the name location-scale models, structural
inference, group families or group models.  We study the case $h(Y) = Z$ and
parameterise $h = g^{-1}$.  The resulting class of ``conditional
transformation models'' was introduced in Hothorn et al.~2014, along with a
boosting algorithm for exact continuous responses.  The novel contribution of
this manuscript is a full likelihood framework for arbitrary responses and
corresponding likelihood inference procedures. This has been clarified in
the revision.}

Once the transformations are parameterized, the asymptotic
results in Section 3 are obvious from standard textbooks.  

\textit{The simplicity of the approach is an advantage because more complex
models can be fitted and analysed in this sound and well-established
likelihood theory, whereas much more complex procedures are typically
applied in (linear) transformation models. The partial likelihood is the
most prominent representative of such ``non-standard'' estimation approaches for models
with high- or even infinite-dimensional nuisance parameters.}

Section 4 provides a list of examples to demonstrate how a variety of models can be
viewed as from certain transformations, which can be also found in Bickel et
al.~(1993).  

\textit{Section~4.2 contains three subsections ``Classical Transformation
Models'', ``Extensions of Classical Transformation Models'' and ``Novel
Transformation Models''. Well-known linear transformation models are
shown to be part of the unifying theory presented here in the first part, 
along with appropriate references. The book by Bickel et al. is an early but
not the first reference for a unifying view on linear transformation models
(and is restricted to the continuous case).
We cited the (as far as we know) first paper in this field by Doksum and
Gasko (1990) instead. The models discussed in the remaining two subsections
have not been discussed in this context so far or are completely novel.}

My specific comments include:

\begin{enumerate}
\item in the introduction, it is not clear which transformations were actually
used in the literature.  Is it $h_Y(y)$ or $h(Y |x)$?

\textit{With the few exceptions mentioned in the Introduction, $h_Y(y)$ is
typically treated as a nuisance parameter and therefore, it is not parameterised nor estimated directly. The
transformation $h(Y |x)$ is a shifted (by a linear predictor) version of
$h_Y(y)$ as explained in detail in Section~4.2.}

\item the bottom of page 4, what if $Y$ contains both discrete and continuous
components?

\textit{$Y$ is univariate and either discrete or continuous. For multiple
observations $Y_i, i = 1, \dots, N$ of mixed type, the corresponding likelihood contributions 
(page 5) are simply multiplied (for example in ``double censoring'').}

\item in Definition 1, does the maximum always exist? If not, what are the
conditions for existence?

\textit{The distribution of $Z$ needs to be log-concave and the basis $a$ needs to
ensure identifiability of the distribution of $Y$. Both conditions have been
added to the revision.}

\item in Section 3, originally I thought that the paper would consider
nonparametric transformations but it turns out that the actual
transformations are parametric in this section.  Obviously, a big issue is
how to choose basis functions and how much parameterization one needs.  The
paper has little discussion on this.

\textit{We used the term ``non-parametric'' for indicating the invariance
with respect to the distribution of $Z$ which was non-standard and has been
revised. The type of the basis functions is problem-specific. In the
discrete case, we have $K-1$ parameters for $K$ levels, in the continuous
case we used Bernstein polynomials of order $M$ for modelling smooth
transformation functions. For count data, one could use either
parameterisation, depending on the number of counts. In the discrete case,
the number of basis functions is fix. In the continuous case (as explained
in Section 4) the number of basis functions is often limited to two (in all
transformation functions allowing linear transformations of $Z$ only) or can
be relatively large due to the monotonicity constraints (see Section 4.2). For the Geyser
example, we tried up to $M = 70$ basis functions and the
resulting distribution function closely approximates the ECDF as one would
expect from the comments in subsection 4.1 (results not shown in paper).}

\item another concern is that the parameterization should satisfy $h(y_1) <
h(y_2)$ if $y_1 < y_2$ .  This can result in complicated constraints, especially when
$h(\cdot)$ depends on covariates $x$.

\textit{The basis functions used in this paper allow monotonicity
constraints to be formulated as linear constraints on the parameters (page
11 and 12). Standard algorithms for optimisation under linear constraints
can be applied. The necessary constraints are better explained in Sections~4
and 5 in the revision.}

\item no computation details are given for the maximization. Indeed,
computation is a nontrivial issue in transformation models for censored
data.

\textit{Appendix ``Computational Details'' contains a reference to the
optimiser used for the experiments. We have tried other linear
constraint optimisers as well. All of them worked resonably well; the implementation in package 
\textbf{BB} was chosen because of its convenient user interface. One
advantage of the framework presented in this paper is the simplicity of the
corresponding optimisation problem and thus standard optimisers can be
applied (as mentioned in the Discussion). In fact, the evaluation of the
likelihood is easier for censored observations as one only needs to compute
the distribution function, instead of the density. Note that the core of our
reference implementation including all models presented here was written in 
less than 1000 lines of pure R code, this also highlights the computational
advantages of the full likelihood estimation procedure.}

\item in general, when we choose different distributions for $Z$, it will result
in different transformations so parameterization can be different.  Would
results be sensitive to the choice of $Z$'s distribution?

\textit{It depends on the model: In three out of four examples presented in
Section~5.1 the distribution of $Y$ is invariant wrt the distribution of
$Z$. Thus, for flexible enough parameterisations of the transformation
function, any choice of $Z$ (subject to the restrictions given in the paper)
will practically lead to the same estimated distribution, but of course with
different parameters $\vartheta$. In such a situation, we are not interested
in the parameter estimates but in the estimated (conditional) distribution, density,
hazard or cumulative hazard functions (as in the the density estimation and
quantile regression applications). In contrast, 
the minimum extreme value distribution was
used to allow comparisons with a standard Cox model in the survival
application. In this model, the
choice of $Z$ as the minimum extreme value distribution allows an
interpretation of the regression coefficients $\beta$ as log-hazard ratios;
all of this is explained and discussed in Section~4.2.}

\end{enumerate}

\subsubsection*{Referee 2}

This paper develops a new theoretical and computational likelihood based
framework, which embeds a large class of transformation models (parametric,
semi-parametric, and non-parametric) commonly estimated by other procedures. 
The approach also allows extensions to novel transformation models.  Besides
the theoretical foundations (complemented by an extensive and exhaustive
bibliography), the paper provides several persuasive applications based on
real data.  It is well written and deserves publication.  No revision is
necessary.

\subsection*{Review History: Version 2 by Journal 2 (January--April 2016)}

\noindent Review for a condenced variant of version 2 (v2, \url{http://arxiv.org/abs/1508.06749v2}).
Comments by the authors are printed in \textit{italics}.

\subsubsection*{Associate Editor}

This paper aims to develop a unified transformation based framework for
modeling, estimation and inference with a wide range of data types.
While the presented formulation of the transformation model seems
interesting,  as noted by all reviewers, it confines the proposed model to a
fully parametric setting through specifying the basis function and $F_Z$ (an
analog to "error distribution"). This poses a major limitation of
the proposed work. For example, in survival applications, the proposed model
may only correspond to a proportional hazards model with
parametric baseline hazard function, which is more restrictive than the
commonly adopted semi-parametric proportional hazards model. 

\textit{This statement is true from a theoretical point of view.
Practically, the regression coefficients in a Cox model obtained from the
partial likelihood and the full likelihood are numerically equivalent even
for rather low-dimensional approximations of $h_Y$. Empirical evidence was
presented in Section 5 of v2 (Figure 6) and we now report on more detailed
investigations of this issue in the package vignette for the R add-on package
\textbf{mlt.docreg} distributed at CRAN. Parameterising the transformation
function allows much more flexible models to be formulated and estimated
(for example distribution regression models as in the growth curve analysis
example) than anything we are aware of in the NPMLE world. As a
side-comment, the dominant semiparametric view on survival analysis made 
estimation of the Cox
model for interval censored data rather troublesome for over 40 years (software became available
only very recently). As a consequence, many analyses of, most prominently,
disease-free survival times treating events observed at follow-up
examinations as ``exact'' event times are probably flawed. Convenient estimation of the
Cox model under \textbf{all} forms of censoring is a by-product of our
approach (and is implemented in the \textbf{mlt} package).}

When positioning the proposed transformation model as a general class of
parametric models, there are some important issues left unaddressed:

\begin{enumerate}
\item How to decide the form of the basis function in real data analysis?

\textit{As explained in sections 4 and 5, there is not much to decide. For
        discrete observations, each element of the support (except one) gets
        one parameter assigned and for the Bernstein polynomials used in the
        empirical parts, $M$ is not a hyper parameter. New empirical
        evidence can be found in the above mentioned vignette.}

\item Similarly, how to choose the distribution function for $Z$?

\textit{Some conditional transformation models are practically invariant wrt $Z$ (in the
        Old Faithful and growth curve and partially the GBSG-2 examples), as
the fitted distribution functions are numerically the same for different
choices if $F_Z$. For
        others, $Z$ is part of the model formulation. For example, in the
        Cox model we need $Z$ to follow a minimum extreme value distribution
        if we want to interpret the regression coefficients $\beta$ as
        log-hazard ratios. As mentioned in the discussion, it is possible to
        estimate $F_Z$ and McLain and Ghosh (2013) introduce a suitable family
        of distributions (newly added to discussion) for this exercise.}

\item  What is the advantage of adopting the transformation model formulation
(presented in Definition 4) over direct parametric modeling of the
distribution function? 

\textit{Section~2 explains that we can generate \textbf{all} distributions
from a suitable $h$ and not just the onces we can find in textbooks. In this
sense, the framework is nonparametric at its heart.}

\end{enumerate}

In addition, I find the presentation of the conditional transformation model
quite vague. Is there an implicit assumption that $h(Y|X)$ has a
distribution independent of $X$?

\textit{$F_Z(h(Y | x))$ has a (discrete or continuous) uniform distribution
in $[0, 1]$ for every $x$ by the probability integral transform.}

\subsubsection*{Referee 1}

This paper provides a mathematical abstraction of the transformation model,
and considers inference through fitting the most likely transformation as a
maximum likelihood approach.
Regularity conditions were provided so that the usual maximum likelihood
mathematical results hold for this transformation model as stated in three
theorems. Then the authors discussed
connections to transformation models in literature, and demonstrated
applications on three example data sets. The transformation model has been
used extensively in literature, but the
mathematical abstraction to summarize the general approach does not seem to
appear before. However, I do have concerns about this manuscript.

\begin{enumerate}
\item The results in section 3 are rather restatements of standard results.
More importantly, they are only stated in the fully parametrized model,
which do not include the most
practically useful and academically interesting semi-parametric
(non-parametric) transformation model. For fully parametrized model, direct
likelihood approach is natural, and it is not
clear that this restatement in terms of transformation model brings any
extra insight. 

\textit{See the comment to overall statement by AE.}

\item  The authors claimed that this likelihood approach clearly illustrate
the connections the general transformation models, while the lack of general
understanding in previous
literature was due to presentation in the relatively narrow contexts of
censored or ordinal data (page 4). However, seems to me that a main reason
for the authors' simple approach is to
ignore the semi-parametric/non-parametric setting that are the main
mathematical challenges. 

\textit{The likelihood does not illustrate the connections between models,
the model formulation does.  We study and extend transformation models
common in the semi- and non-parametric world while making model estimation
much simpler (also numerically) by introducing appropriate
parameterisations.  See also comment to overall statement by AE.}

The authors deal with the
semi-parametric/nonparametric case by using Bernstein polynomials
of fixed degree to approximate the transformation function, then treat this
as a parametric fitting. However, for practical purpose, the choice of the
degree are critical for data
analysis.

\textit{New empirical evidence presented in the package
vignette of package \textbf{mlt.docreg} suggests that this is actually not
the case.}

And for theoretical correctness, the degree should increase with
sample size, not fixed.  

\textit{A new paragraph 3.2. was introduced making the link between the
estimation procedure applied here and the asymptotic theory for
unconditional and linear transformation models presented by McLain and Ghosh
(2013).}

The authors simply state that: cannot choose M ``too
large'' (page 18, line 6). This
does not really prove any practical guidance nor provide theoretical
guarantee for applications.  The three examples in section 5, used degrees
8, 3, 5 respectively. There are no
discussion on why these degrees are appropriate, and how would the results
change with a different choice of these tuning parameters?

\textit{See above.}

\item Related to the above point, page 20 discussed the connection to the
survival time models. But those are rather misleading. For proportion
hazards/odds models commonly used in
survival analysis, the transformation $h_Y$ is generally assumed unknown
(rather than parameterized with Bernstein polynomial here). So this model
does not capture the most general PH/PO
model, where the nonparametric maximal likelihood estimator was proven to
work. 

\textit{See the comment to overall statement by AE.}

Also, for the AFT model, the current model just correspond to some
parametric submodels with both $(F,h)$
assumed known. Those are specific parametric models that do not need this
new transformation model abstraction. 

\textit{In our opinion it is very interesting, at least from a practical,
computational, and
educational point of view, to unify all models in this framework. 
Understanding the connections between models allows one to estimate
\textbf{all} models under \textbf{all} forms of random censoring and
truncation using just about 1000 ``smart'' lines of R code (now available in
package \textbf{mlt}).}

For the more general AFT model, the
$F$ should be assumed unknown, which the
transformation model approach in this paper does not cover how to do
statistical inference.

\textit{See AE point 3.}

\end{enumerate}

\subsubsection*{Referee 2}

In this paper, the authors show that one can always write a potentially complex distri-
bution function FY as the composition of a priori specified distribution function
$F_Z$ and a
strictly monotone transformation function $h$ (unknown). The task of estimating
$F_Y$ is then
reduced to obtaining an estimate of $h$, which is parameterised in terms of basis functions.
The likelihood function of the transformed variable can then be characterised by transforma-
tion function and the parameters of interest can be estimated and analysed in the maximum
likelihood framework. The authors establish the asymptotic normality of the proposed es-
timators. A reference software implementation of maximum likelihood-based estimation for
conditional transformation models was employed to illustrate the wide range of possible
applications.

\paragraph{Major:}

\begin{enumerate}
\item My major concern is the specification of $F_Z$, which is assumed to be known. Since
$Y$
can be subject to censoring/truncation, how do we choose appropriate $F_Z$ based on cen-
sored/truncated data? In particular, for interval-censored or truncated data, how do we
decide an appropriate $F_Z$ in the likelihood function ? The authors should address this issue.

\textit{$F_Z$ defines the transformation \textbf{model} and has nothing to do with
its \textbf{estimation} under censoring or
truncation. Only the likelihood function takes censoring (by integrating
over all possible realisations) and truncation (by conditioning) into
account. For example, the Cox \textbf{model} has nothing to do with censoring, only its
\textbf{likelihood} (full or partial) allows for this. In the same spirit,
we can fit a normal linear model under censoring and truncation without
changing anything in the model formulation.}

\end{enumerate}

\paragraph{The Other Comments:}

\begin{enumerate}
\item page 8, lines -1 through -5: The authors only demonstrate likelihood contribution of
censored data. How about truncated data ?

\textit{The likelihood for truncated observations is given on page 6, line
2, of v2 and page 9, line 12 in your version.}

\item page 9, lines -1 through -6: It is not clear to me that $Y \in C_i$, where
$C_i$ denotes an independent sample of possibly censored or truncated observations. 

\textit{$C_i$ are intervals, ie the observations, and $\rY$ is the rv.}

What is the definition of $H$ ? 

\textit{The space $H$ is defined on page 4, line -8 in v2 and on page 7, line 14 in your version.}

In practice, can we obtain $\hat{h}_N$ ? 

\textit{Via plug-in, as explained on page 7, line -4 in v2 and page 11, line -8 in
your version.}

When truncation is present (such as left truncation
or double truncation), does the concavity of the log-likelihood hold ? How does this impact
the uniqueness of $\hat{h}_N$ ?

\textit{The sum of two concave functions is again concave.}

\item page 10, line 20: What is the advantage of parameterising the transformation function
$h(y)$ in terms of basis functions. Is there any other alternative (e.g. nonparametric estimator)? 
Please comment on this.

\textit{Parameterising $h$ leads to models with well-defined parameters and 
allows efficient evaluation of the likelihood function.}

\item page 11, line 6: When truncation is present, is $\hat{\vartheta}_N$ unique ?

\textit{See your point 2.}

\item page 13: In Theorems 1 and 2, the authors do not mention any assumption regarding
the left and right endpoints of the distribution functions of lifetime, censoring time and
truncation time. Are they related to the asymptotic properties of
$\hat{\vartheta}_N$ ?

\textit{Sure, but as stated on page 9, line 3 in v2 and page 12, line -3 in
your version, the theorems describe the asymptotic distribution for the
``exact'' case.}

\item page 18: The authors should point out how the conditional transformation model relate
to semiparametric transformation models (Cheng et al. 1995). How does the proposed
method differ from that of Zeng and Lin (2006,2007).

\textit{Zeng and Lin (2007) study linear transformation models
without explicit parameterisation of the baseline transformation function
$\lambda$ enriched with random effects. The model (without random effects) is a
simple linear transformation model known for decades and their estimation
procedure is completely different from the one discussed here.}

\item page 24: In empirical evaluation, the data illustration can be improved. Please compare
the proposed method with the other approaches.

\textit{We removed the simulation study presented in v2 from the condenced
version, please see Section~5 in all versions of the arXiv report.}

\end{enumerate}

\subsection*{Review History: Version 3 by Scandinavian Journal of Statistics (November 2016--March 2017)}

\noindent Review for a condenced variant of version 3 (v3, \url{http://arxiv.org/abs/1508.06749v3}).
Comments by the authors are printed in \textit{italics}.

\subsubsection*{Associate Editor}

\textit{
All reviewers think that the paper is interesting and that the unified
framework includes a variety of data types and models.  The comments are
mostly minor.  However, all reviewers and myself think that the paper
presentation should be much improved in next version.  Please follow the
reviewer's specific comments to make revision.
}

We revised the manuscript taking all suggestions made by the three referees
into account. The main changes relate to the following issues
\begin{itemize}
\item Connections between Models. Two referees suggested to improve the
  presentation of connections between models. We added a new Table 1
  containing a systematic overview on linear and response-varying
  transformation models. We also stress that models, parameterisation
  and estimation are not a unity but should be understood as separate
  issues.
\item Parameter Interpretability. In some models, transformation functions
   or shift parameters $\beta$ or response-varying parameter functions 
   $\beta(y)$ have a clear meaning. We mention interpretability wherever
   possible and also included this information in our new Table~1.
\item Constraints. We explicitly present linear constraints for all models
  discussed in this manuscript.
\item Choice and Meaning of $F_Z$. Following a suggestion by referee 3, we
  added the mean one exponential distribution and a presentation of
  corresponding additive hazard models. Distributions $F_Z$ with free
  parameters are discussed in a dedicated paragraph in the last section now.
\end{itemize}

\subsubsection*{Referee 1}

\textit{
This study develops a unified framework for studying transformation models
and the corresponding (downstream) estimates. What is really nice about
this study is that multiple data types, distributions, and models are
comprehensively investigated under one framework. Some numerical studies
are conducted. I have some minor comments.
}

\begin{enumerate}

\item 
\textit{The introduction is ``too comprehensive'' in that a lot of existing studies
are mentioned. With so much information, it gets confusing. Please consider
better organising discussions on the existing studies. Especially please
make it clear which existing studies can be considered under the proposed
framework, and which are beyond.}

The common denominator of all studies cited are transformation models.  The
differences lie in (1) different estimation techniques, such as
non-parametric maximum likelihood estimation, estimation equations or
optimisation of scoring rules and (2) different parameterisations (or lack
thereof).  We agree that it is hard to see and understand the connections
(we explicitly comment on this problem in the 3rd paragraph) because models,
parameterisation and estimation are often treated as a unity. We addressed
this issue in two ways. First, the introduction makes it clear that the
transformation \emph{models} cited in the introduction can be understood as
members of the family of transformation models presented in this manuscript.
With the specific parameterisations proposed here, one can also
\emph{estimate} this models by maximum likelihood. Second, we added a
systematic overview on transformation models and their connections to other
models following a suggestion by referee 3 (new Table 1). We hope that
this overview will make it easier to understand the connections between
different models.

\item \textit{
For the density of $f_Z$, the authors mentioned normal, logistic, and minimum
extreme value distributions. And there should be quite a few other choices.
Theoretically, computationally, and numerically, does this choice make a
difference?
}

The main difference is with respect to the interpretation of regression
parameters $\beta$ in linear transformation models. Following a suggestion
by Referee 3, we now present a systematic overview on different flavours of
transformation models and corresponding parameter interpretations (Table 1).
We also added the mean one exponential distribution $F_Z$ because it allows
fitting additive hazard models. Some authors studied distribution functions
$F_Z$ with free parameters, these approaches are now discussed in a new
paragraph in the last section.

\item \textit{
The nonparametric inference in Section 3.2 is discussed relatively briefly.
For smooth nonparametric estimation, there are multiple candidate
approaches. Does that choice make a difference?}

Only from a computational point of view. It is very easy to impose
monotonicity on a Bernstein polynomial (make sure the parameters are
monotone increasing). This is much harder to achieve for a B-spline, for
example. These issues are discussed in Section~4.1.

\item \textit{
The illustrations in Section 5.1 do not provide much insight. They can be
moved after simulation.}

Feedback from readers showed that many people have difficulties following
our path through the transformation model jungle. Presentation of
applications usually helps them to get started much easier. We therefore
wanted to put strong emphasis on applications. Section 5 illustrates the
broadness of transformation models without the need to understand every
technical detail. We therefore prefer to keep the current order of
presentation.

\end{enumerate}

\subsubsection*{Referee 2}

\textit{Authors studied the properties of maximum likelihood estimators in the class of
conditional transformation models. The distribution of response variable was written
in the priori specified distribution, such as standard normal, standard log-logistic and
the minimum extreme value distribution after certain transformation $h(y)$. Then, the
same theoretical and computational framework can be applied simply by choosing an
appropriate transformation function and parameterizations thereof. For discrete and
continuous responses, they established the asymptotic normality of the proposed 
estimators. The ``mlt'' package in R can implement the proposed method and empirical
evaluation illustrated the wide range of possible applications.}

\begin{enumerate}

\item \textit{
The model was determined by the priori specified distribution and $h_Y(y) +
h_x(x)$ and $h_Y(y)$ can be specified or approximated by the spline. The authors have tested the
performance of data without censoring and right censoring. Theoretically, this MLE
approach can be extended to the interval censored data. However, usually, there are
issues in estimation for the interval censored data with spline from the direct MLE,
which may still exist in the proposed method. The authors may add discussions on this.
}

In fact, linear transformation models of the form $h_Y(y) + h_x(x)$ are the
simplest members of the much more general class of conditional
transformation models covered in our manuscript. For linear and more complex
transformation models, we did not encounter theoretical or computational
difficulties when estimating most likely transformations. In her master
thesis, Mariia Dobrynina systematically investigated finite sample
properties of most likely transformations under interval censoring. She
found the R add-on package "mlt" to be correct also for interval
censored data. The thesis can be downloaded from \\
\texttt{http://user.math.uzh.ch/hothorn/docs/MasterThesisSfS.pdf}.

\item \textit{
There is no definition of $\tilde{h}$ in definition 1.
}

$\tilde{h}$ is an element of $\mathcal{H}$ (underneath argmax).

\end{enumerate}

\subsubsection*{Referee 3}

\textit{
In the paper ``Most Likely Transformations'' the authors propose the use of a
flexible class of parametric sieve-type models for conditional and 
unconditional density estimation.  The theory of the models is presented, along
with the examples of how the approach would be applied.  I found this to be
a well written paper on an interesting topic.  Please find my main comments
on the manuscript below.}

\begin{enumerate}

\item \textit{
I think it's important to give some differences between the proposed
approach and the well-known transformation models from Zeng and
Lin. It appears that this is all in the introduction, but a sentence
focusing and clarifying this aspect would be helpful.}

Zeng and Lin study \emph{nonparametric} maximum likelihood estimator without
explicit parameterisations of baseline transformation functions $h_Y$. We
highlighted this fact. The models discussed by Zeng and Lin, however, can be
understood as conditional transformation models and thus our
likelihood-based estimation procedure can be used for parameterisation and 
estimation. This is now better explained in the introduction.

\item \textit{
Your model also has a clear connection with some grouped survival
methods. For example
\begin{itemize}
\item Aranda-Ordaz FJ. An extension of the proportional- hazards model for
grouped data. Biometrics 1983; 39: 109--117 and
\item Tibshirani RJ and Ciampi A. A family of proportional- and additive-
hazards models for survival data. Biometrics 1983; 39: 141--147.
\end{itemize}
The connection to these models with some discussion would be a nice
addition.
}

Thank you very much for this suggestion. We added the corresponding mean one
exponential distribution $F_Z$ and discuss the resulting
additive hazard models in Section 4 and in Table 1 now. We also added a
paragraph discussing the possibility to have free parameters in $F_Z$ in the
last section. The piecewise constant parameterisation employed in these two
papers is, however, only motivated by the limited computing resources
available in the 1980ies. To quote Aranda-Ordaz: "It is often reasonable to
consider survival time as essentially a continuous random variable." Thus,
with todays computing power, we think it is much more adequate to 
estimate smooth transformation functions, for example using the method
introduced in our paper.

\item \textit{
I think you could add a table summarizing the correspondence of the
proposed approach to the various popular modeling procedures by
$(F_Z , a, \vartheta)$. This is in the paper, but more scattered. A difficult 
issue with a paper that is as broad as this one, is the presentation of the
various of models and methods in a systematic fashion. Section 4 does
an honorable job of this, but making the connections is still difficult. A
focused table would be a nice start to clarifying the methods, but I 
recommend that authors look at this issue more generally (see comment
6 as well).
}

We thank the referee for this suggestion. We added a table covering 
linear transformation models as well as response-varying transformation
models for binary, polytomous, count and continuous regression, the latter
including survival analysis, to Section 4. Indeed, this more systematic overview
helps to identify white spots in the model landscape much easier.

\item \textit{
The first equation in ``Count Regression Without Tears'' the first
equation has an extra ')'.
}

Thank you, fixed.

\item \textit{
Please explain the issue of the linear constraints of the model in a little
more detail. They are mentioned in passing in Section 4, then they
appear as a much bigger issue in Section 5. A more focused discussion
of the roles that linear constraints can play would be helpful.
}

Linear constraints on the model parameters are necessary to ensure
monotonicity of the transformation function. For linear transformation
models, monotonicity of the ``baseline'' transformation $h_Y$ is explained
in Section 4 and we now added a more detailed discussion to section 4.3 and
4.5 (and added a reference to a technical report explaining the linear
constraints for different models). The specific constraints for the
empirical studies in Section 5 are now made explicit.

\item \textit{
Is there any form that can be given to help the interpretation (even
graphically) of covariate parameters in some tangible way? This is
attempted in Figure 3, but having ``Transformation deviation'' as the
y-axis makes the plot uninterpretable. Is there an equation that will
give some tangible meaning to the coefficients for all possible models (or
maybe just all continuous models)? Further, is the sign interpretable.
Maybe there is a broad subset of models where a coefficient with a
particular sign has a monotonic influence on the outcome (e.g., negative
coefficient always extends survival)? Please include some discussion on
the interpretation. It does not have to hit these specific points, but a
discussion of the specific and general (i.e., sign) interpretations would
be beneficial.
}

The new table presenting an overview on models also contains interpretations
for the baseline transformation function $h_Y$ and the regression
coefficients $\beta$ in linear transformation models of the form $h(y \mid x)
 = h_Y(y) - x^\top \beta$. We parameterised a negative shift, because of
the interpretation $E(h(Y \mid x)) = x^\top \beta + c$, ie larger values of
the linear predictor correspond to larger conditional expectations (but
smaller risks in survival analysis). This relationship is discussed in
Section 4.3.

\end{enumerate}

\end{appendix}

\end{document}